\documentclass[11pt]{amsart}
\usepackage[english]{babel}
\usepackage[utf8]{inputenc}
\usepackage[inner=1in, outer=1in, bottom =2.9cm, top =2.9cm]{geometry}
\makeatletter
\newcommand*\Cdot{\mathpalette\Cdot@{.5}}
\newcommand*\Cdot@[2]{\mathbin{\vcenter{\hbox{\scalebox{#2}{$\m@th#1\circled{$1$}$}}}}}
\makeatother

\usepackage{amsmath, amssymb, amsfonts, amsthm}
\usepackage{graphicx, overpic}
\usepackage{mathrsfs}
\usepackage{mathtools}
\usepackage[usenames,dvipsnames,svgnames,table]{xcolor}
\usepackage{enumerate}
\usepackage{enumitem}
\usepackage{shuffle} %Shuffle
\usepackage{lipsum}
\usepackage{color}
\usepackage{phoenician}
\usepackage{xkeyval}
\usepackage{tikz}
\usepackage{pst-node} % Commutative Diagrams
\usepackage{tikz-cd} % Commutative Diagrams
\usepackage[OT2,T1]{fontenc} %Shuffle
\usepackage{thm-restate} 
\usepackage{float}
\usepackage{bold-extra}
\usepackage{microtype}
\usepackage[all,cmtip]{xy}
\usepackage{etoolbox}
\usepackage[strict]{changepage} %Margins
\usepackage{esvect}
\usepackage{wrapfig}
\usepackage{caption}
\usepackage{cite}
\usepackage{eucal}
\usepackage{mathrsfs}
\DeclareMathAlphabet{\mathpzc}{OT1}{pzc}{m}{it}
\usepackage{bbding}
\usepackage{array}
\usepackage{pict2e}
\usepackage{makecell}
\usepackage{stmaryrd}
\usepackage{lmodern}
\usepackage{amsbsy}
\usepackage{esvect}
\usepackage{bbding}
\usepackage{thm-restate}
\usepackage[toc,page, titletoc,title]{appendix}
\graphicspath{ {figures/} }
\usepackage{url}
\usepackage[pagebackref]{hyperref} %PUT LAST
\newtagform{red}{\color{blue}(}{)}

\colorlet{linkequation}{blue}
\newcommand*{\refeqq}[1]{%
  \begingroup
    \hypersetup{
      linkcolor=linkequation,
      linkbordercolor=linkequation,
    }%
    \ref{#1}%
  \endgroup
}
\makeatletter
\newcommand{\colim@}[2]{%
  \vtop{\m@th\ialign{##\cr
    \hfil$#1\operator@font colim$\hfil\cr
    \noalign{\nointerlineskip\kern1.5\ex@}#2\cr
    \noalign{\nointerlineskip\kern-\ex@}\cr}}%
}
\newcommand{\colim}{%
  \mathop{\mathpalette\colim@{\rightarrowfill@\scriptscriptstyle}}\nmlimits@
}
\renewcommand{\varprojlim}{%
  \mathop{\mathpalette\varlim@{\leftarrowfill@\scriptscriptstyle}}\nmlimits@
}
\renewcommand{\varinjlim}{%
  \mathop{\mathpalette\varlim@{\rightarrowfill@\scriptscriptstyle}}\nmlimits@
}
\makeatother

\makeatletter
\providecommand*{\twoheadrightarrowfill@}{%
  \arrowfill@\relbar\relbar\twoheadrightarrow
}
\providecommand*{\twoheadleftarrowfill@}{%
  \arrowfill@\twoheadleftarrow\relbar\relbar
}
\providecommand*{\xtwoheadrightarrow}[2][]{%
  \ext@arrow 0579\twoheadrightarrowfill@{#1}{#2}%
}
\providecommand*{\xtwoheadleftarrow}[2][]{%
  \ext@arrow 5097\twoheadleftarrowfill@{#1}{#2}%
}
\makeatother
\makeatletter
\newcommand*{\relrelbarsep}{.386ex}
\newcommand*{\relrelbar}{%
  \mathrel{%
    \mathpalette\@relrelbar\relrelbarsep
  }%
}
\newcommand*{\@relrelbar}[2]{%
  \raise#2\hbox to 0pt{$\m@th#1\relbar$\hss}%
  \lower#2\hbox{$\m@th#1\relbar$}%
}
\providecommand*{\rightrightarrowsfill@}{%
  \arrowfill@\relrelbar\relrelbar\rightrightarrows
}
\providecommand*{\leftleftarrowsfill@}{%
  \arrowfill@\leftleftarrows\relrelbar\relrelbar
}
\providecommand*{\xrightrightarrows}[2][]{%
  \ext@arrow 0359\rightrightarrowsfill@{#1}{#2}%
}
\providecommand*{\xleftleftarrows}[2][]{%
  \ext@arrow 3095\leftleftarrowsfill@{#1}{#2}%
}
\makeatother
\DeclareSymbolFont{cyrletters}{OT2}{wncyr}{m}{n}
\DeclareMathSymbol{\Sh}{\mathalpha}{cyrletters}{"58}

\usetikzlibrary{tikzmark,decorations.pathreplacing}
\usetikzlibrary{shapes,shadows,arrows}
\usetikzlibrary{decorations.markings}
\usetikzlibrary{positioning,calc}
\tikzset{near start abs/.style={xshift=1cm}}
\DeclareSymbolFont{symbolsC}{U}{txsyc}{m}{n}
\DeclareMathSymbol{\Searrow}{\mathrel}{symbolsC}{117}

\hypersetup{
	colorlinks=true,
	linkcolor=black,
	filecolor=black,      
	urlcolor=blue,
	citecolor= 	red,
}
\DeclareSymbolFont{extraup}{U}{zavm}{m}{n}
\DeclareMathSymbol{\varheart}{\mathalpha}{extraup}{86}
\DeclareMathSymbol{\vardiamond}{\mathalpha}{extraup}{87}
 \DeclareMathSymbol{\varclub}{\mathalpha}{extraup}{84} 
\DeclareMathSymbol{\varspade}{\mathalpha}{extraup}{85}

\newcommand{\bigslant}[2]{{\raisebox{.2em}{$#1$}\left/\raisebox{-.2em}{$#2$}\right.}}
\theoremstyle{definition}
\newtheorem{thm}{Theorem}[section]
\newtheorem{cor}{Corollary}[thm]

\newtheorem{prop}[thm]{Proposition}

\theoremstyle{definition}

\newtheorem{ex}{Example}[section]

\newtheorem{remark}{Remark}[section]

\newcommand{\ga}{\alpha}

\newcommand{\sfVec}{\textsf{Vec}}
\newcommand{\sfS}{\textsf{S}}

\newcommand{\tH}{\mathtt{H}}
\newcommand{\tQ}{\mathtt{Q}}

\newcommand{\ta}{\mathtt{a}}

\newcommand{\bR}{\mathbb{R}}
\newcommand{\bN}{\mathbb{N}}
\newcommand{\bC}{\mathbb{C}}

\newcommand{\bZ}{\mathbb{Z}}

\newcommand{\cF}{\CMcal{F}}

\newcommand{\cH}{\CMcal{H}}

\newcommand{\cS}{\CMcal{S}}
\newcommand{\cA}{\CMcal{A}}

\newcommand{\cX}{\CMcal{X}}

\newcommand{\op}{\operatorname{op}}

\newcommand{\Hom}{\operatorname{Hom}}

\newcommand{\bE}{\textbf{E}}

\newcommand{\Lie}{ \textbf{Lie} }
\newcommand{\Zie}{\textbf{Zie}}
\newcommand{\Com}{\textbf{Com}}
\newcommand{\Ass}{\textbf{Ass}}

\newcommand{\Gam}{\boldsymbol{\Gamma}}
\newcommand{\Sig}{\boldsymbol{\Sigma}}  
\DeclareFontFamily{U}{mathx}{\hyphenchar\font45}
\DeclareFontShape{U}{mathx}{m}{n}{
      <5> <6> <7> <8> <9> <10>
      <10.95> <12> <14.4> <17.28> <20.74> <24.88>
      mathx10
      }{}
\DeclareSymbolFont{mathx}{U}{mathx}{m}{n}
\DeclareFontSubstitution{U}{mathx}{m}{n}
\DeclareMathAccent{\widecheck}{0}{mathx}{"71}
\DeclareMathAccent{\wideparen}{0}{mathx}{"75}

\newcommand{\res}{\parallel}

\newcommand{\formj}{\emph{\texttt{j}}\, }
\newcommand{\formg}{\emph{\texttt{g}}}

\newcommand{\ssS}{\emph{\textsf{S}}}
\newcommand{\ssA}{\emph{\textsf{A}}}

\makeatletter
\newcommand*\bigcdot{\mathpalette\bigcdot@{.5}}
\newcommand*\bigcdot@[2]{\mathbin{\vcenter{\hbox{\scalebox{#2}{$\m@th#1\bullet$}}}}}
\makeatother

\makeatletter
\newcommand{\adjunction}{\@ifstar\named@adjunction\normal@adjunction}
\newcommand{\normal@adjunction}[4]{%
  #1\colon #2%
  \mathrel{\vcenter{%
    \offinterlineskip\m@th
    \ialign{%
      \hfil$##$\hfil\cr
      \longrightharpoonup\cr
      \noalign{\kern-.3ex}
      \smallbot\cr
      \longleftharpoondown\cr
    }%
  }}%
  #3 \noloc #4%
}
\newcommand{\named@adjunction}[4]{%
  #2%
  \mathrel{\vcenter{%
    \offinterlineskip\m@th
    \ialign{%
      \hfil$##$\hfil\cr
      \scriptstyle#1\cr
      \noalign{\kern.1ex}
      \longrightharpoonup\cr
      \noalign{\kern-.3ex}
      \smallbot\cr
      \longleftharpoondown\cr
      \scriptstyle#4\cr
    }%
  }}%
  #3%
}
\newcommand{\longrightharpoonup}{\relbar\joinrel\rightharpoonup}
\newcommand{\longleftharpoondown}{\leftharpoondown\joinrel\relbar}
\newcommand\noloc{%
  \nobreak
  \mspace{6mu plus 1mu}
  {:}
  \nonscript\mkern-\thinmuskip
  \mathpunct{}
  \mspace{2mu}
}
\newcommand{\smallbot}{%
  \begingroup\setlength\unitlength{.15em}%
  \begin{picture}(1,1)
  \roundcap
  \polyline(0,0)(1,0)
  \polyline(0.5,0)(0.5,1)
  \end{picture}%
  \endgroup
}
\makeatother
\newcommand{\leftrarrows}{\mathrel{\raise.75ex\hbox{\oalign{%
  $\scriptstyle\leftarrow$\cr
  \vrule width0pt height.5ex$\hfil\scriptstyle\relbar$\cr}}}}
\newcommand{\lrightarrows}{\mathrel{\raise.75ex\hbox{\oalign{%
  $\scriptstyle\relbar$\hfil\cr
  $\scriptstyle\vrule width0pt height.5ex\smash\rightarrow$\cr}}}}
\newcommand{\Rrelbar}{\mathrel{\raise.75ex\hbox{\oalign{%
  $\scriptstyle\relbar$\cr
  \vrule width0pt height.5ex$\scriptstyle\relbar$}}}}

\makeatletter
\def\leftrightarrowsfill@{\arrowfill@\leftrarrows\Rrelbar\lrightarrows}
\newcommand{\xleftrightarrows}[2][]{\ext@arrow 3399\leftrightarrowsfill@{#1}{#2}}
\makeatother

\newcommand{\la}{\langle}
\newcommand{\ra}{\rangle}
\newcommand{\wt}{\widetilde}

\definecolor{Red}{rgb}{0.8,0,0.2}
\newcommand{\GG}[1]{}

\makeatletter
\def\@footnotecolor{red}
\define@key{Hyp}{footnotecolor}{%
 \HyColor@HyperrefColor{#1}\@footnotecolor%
}
\def\@footnotemark{%
    \leavevmode
    \ifhmode\edef\@x@sf{\the\spacefactor}\nobreak\fi
    \stepcounter{Hfootnote}%
    \global\let\Hy@saved@currentHref\@currentHref
    \hyper@makecurrent{Hfootnote}%
    \global\let\Hy@footnote@currentHref\@currentHref
    \global\let\@currentHref\Hy@saved@currentHref
    \hyper@linkstart{footnote}{\Hy@footnote@currentHref}%
    \@makefnmark
    \hyper@linkend
    \ifhmode\spacefactor\@x@sf\fi
    \relax
  }%
\makeatother

\hypersetup{footnotecolor=blue}
\makeatletter
\patchcmd{\@startsection}
  {\@afterindenttrue}
  {\@afterindentfalse}
  {}{}
\makeatother
\title[Hopf Monoids in Perturbative Algebraic Quantum Field Theory]{Hopf Monoids in Perturbative Algebraic\\ Quantum Field Theory}
\author{William Norledge}
\address{Pennsylvania State University}
\email{wxn39@psu.edu}
\thanks{This paper is an abridged version of `Species-theoretic foundations of perturbative quantum field theory', arXiv:2009.09969}
\begin{document}

\usetagform{red}

% % % % % % % % % % % % % % % % % % % % % % % % %
% Autoref names
\renewcommand{\chapterautorefname}{Chapter}
\renewcommand{\sectionautorefname}{Section}
\renewcommand{\subsectionautorefname}{Section}
% % % % % % % % % % % % % % % % % % % % % % % % % 

% % % % % % % % % % % % % % % % % % % % % % % % % % % % % 
% Autoref names
\renewcommand{\chapterautorefname}{Chapter}
\renewcommand{\sectionautorefname}{Section}
\renewcommand{\subsectionautorefname}{Section}
% % % % % % % % % % % % % % % % % % % % % % % % % % % % % 

\begin{abstract}
We develop an algebraic formalism for perturbative quantum field theory (pQFT) which is based on Joyal's combinatorial species. We show that certain basic structures of pQFT are correctly viewed as algebraic structures internal to species, constructed with respect to the Cauchy monoidal product. Aspects of this formalism have appeared in the physics literature, particularly in the work of Bogoliubov-Shirkov, Steinmann, Ruelle, and \hbox{Epstein-Glaser-Stora}. In this paper, we give a fully explicit account in terms of modern theory developed by \hbox{Aguiar-Mahajan}. We describe the central construction of causal perturbation theory as a homomorphism from the Hopf monoid of set compositions, decorated with local observables, into the Wick algebra of microcausal polynomial observables. The operator-valued distributions called (generalized) time-ordered products and (generalized) retarded products are obtained as images of fundamental elements of this Hopf monoid under the curried homomorphism. The perturbative \hbox{S-matrix} scheme corresponds to the so-called universal series, and the property of causal factorization is naturally expressed in terms of the action of the Hopf monoid on itself by Hopf powers, called the Tits product. Given a system of fully renormalized time-ordered products, the perturbative construction of the corresponding interacting products is via an up biderivation of the Hopf monoid, which recovers Bogoliubov's formula. 
%then correspond to various elements of the algebra. In particular, time-ordered products are the universal series, operator products of time-ordered products are the H-basis, and generalized retarded products are the Dynkin elements.
%The S-matrix is the image of a decorated universal series in the cocommutative Hopf monoid of set compositions under the homomorphism into the Wick algebra which is determined by the time-ordered products, being Wick algebra-valued functions of the decorations. 
\end{abstract}

\maketitle

% % % % % % % % % % % % % % % % % % % % % % % % % % % % %  gap between abstract and contents
\vspace{-6.5ex}
% % % % % % % % % % % % % % % % % % % % % % % % % % % % % 

% % % % % % % % % % % % % % % % % % % % % % % % % % % % % 
\setcounter{tocdepth}{1} % what appears
\hypertarget{foo}{ }
\tableofcontents
%\setlength{\cftbeforesecskip}{1pt} % spacing
% % % % % % % % % % % % % % % % % % % % % % % % % % % % % 

%%%%%%%%%%%%%%%%%%%%%%%%%%%%%%%%%%%%%%%%%%%%%%%%%%%%%%%%%%%%%%%%%%%%%%%%
%%%%%%%%%%%%%%%%%%%%%%%%%%%%%%%%%%%%%%%%%%%%%%%%%%%%%%%%%%%%%%%%%%%%%%%%
\section*{Introduction}\label{intro}
%%%%%%%%%%%%%%%%%%%%%%%%%%%%%%%%%%%%%%%%%%%%%%%%%%%%%%%%%%%%%%%%%%%%%%%%
%%%%%%%%%%%%%%%%%%%%%%%%%%%%%%%%%%%%%%%%%%%%%%%%%%%%%%%%%%%%%%%%%%%%%%%%

The theory of species is a richer, categorified version of analyzing combinatorial structures in terms of generating functions, going back to André Joyal \cite{joyal1981theorie}, \cite{joyal1986foncteurs}, \cite{bergeron1998combinatorial}. In this approach, one sees additional structure by encoding processes of \emph{relabeling} combinatorial objects, that is by modeling combinatorial objects as presheaves on the category $\sfS$ of finite sets $I$ (the labels) and bijections $\sigma$ (relabelings). In this paper, we are concerned with species $\textbf{p}$ valued in complex vector spaces, i.e. functors of the form
\[
\textbf{p}:\sfS^{\op}\to \sfVec, \qquad I\mapsto \textbf{p}[I] 
,\quad 
\sigma \mapsto \textbf{p}[\sigma]
\]
where $\sfVec$ is the category of complex vector spaces. Explicitly, $\textbf{p}$ consists of a complex vector space $\textbf{p}[I]$ for each finite set $I$, and a bijective linear map $\textbf{p}[\sigma]:\textbf{p}[I]\to \textbf{p}[J]$ for each bijection $\sigma:J\to I$ such that composition of bijections is preserved.

A highly structured theory of gebras\footnote{\ meaning (co/bi/Hopf)algebras and Lie (co)algebras} internal to vector species has been developed by \hbox{Aguiar-Mahajan} \cite{aguiar2010monoidal}, \cite{aguiar2013hopf}, building on the work of Barratt \cite{barratt1978twisted}, Joyal \cite{joyal1986foncteurs}, Schmitt \cite{Bill93}, Stover \cite{stover1993equivalence}, and others. For the internalization, one uses the Day convolution monoidal product $\textbf{p}\bigcdot\textbf{q}$ with respect to disjoint union and tensor product, given by
\[
\textbf{p}\bigcdot\textbf{q}[I] =  \textbf{p} \otimes_{\text{Day}} \textbf{q}[I]= \bigoplus_{S\sqcup T=I} \textbf{p}[S]\otimes \textbf{q}[T]
.\]
This may be viewed as a categorification of the Cauchy product of formal power series.\footnote{\ from the perspective of $\textsf{S}$-colored (co)operads, as defined in e.g. \cite[Section 3]{MR3134040}, there is an equivalent description of these gebras as (co)algebras over the left (co)action (co)monads of the (co)operads $\Com^{ (\ast) }$, $\Ass^{ (\ast) }$, $\Lie^{ (\ast) }$ \cite[Appendix B.5]{aguiar2010monoidal}, which relates the gebras of this paper to structures such as cyclic operads, which already appear in mathematical physics} Various decategorifications of \hbox{Aguiar-Mahajan's} theory recovers the plethora of graded combinatorial Hopf algebras which have been studied \cite[Chapter 15]{aguiar2010monoidal}. %In particular, the species analogs of the Connes-Kreimer graded Hopf algebras of rooted trees are constructed in \cite[Section 13.3]{aguiar2010monoidal}. 

On the other hand, quantum field theory (QFT) may be viewed as a kind of modern infinite dimensional calculus. Perturbative quantum field theory (pQFT) is the part of QFT which considers Taylor series approximations of smooth functions. By an argument of Dyson \cite{Dyson52}, Taylor series of realistic pQFTs are expected to have vanishing radius of convergence.
%They are formal power series in both Planck's constant $\hbar$, measuring the size of quantum effects, and the coupling constant $\formg$, measuring the size of interactions.\footnote{\ thus, in pAQFT, we study the formal infinitesimal neighborhoods of classical free field theories} 
Nevertheless, if an actual smooth function of a non-perturbative quantum field theory is being approximated, then they are asymptotic series, and so one might expect their truncations to agree to reasonable precision with experiment. This is indeed the case. 

%In contrast to pQFT, there are currently no known examples of realistic (non-perturbative) QFTs.

There are two main synthetic approaches to (non-perturbative) QFT, which grew out of the failure to make sense of the path integral analytically. There is functorial quantum field theory (FQFT), which formalizes the Schr\"odinger picture by assigning time evolution operators to cobordisms between spacetimes. There is also algebraic quantum field theory (AQFT), going back to \cite{haagkas64}, which formalizes the Heisenberg picture by assigning $\text{C}^\ast$-algebras of observables to regions of spacetime. Low dimension examples of AQFTs/Wightman field theories were rigorously constructed in seminal work of \hbox{Glimm-Jaffe} and others \cite{MR247845}, \cite{MR272301}, \cite{MR363256}. 
%See also \cite{glimm2012quantum}. See \cite{Fred19} for some recent developments. 

Perturbative algebraic quantum field theory (pAQFT) \cite{rejzner2016pQFT}, \cite{dutsch2019perturbative}, \cite[\href{https://ncatlab.org/nlab/show/geometry+of+physics+--+perturbative+quantum+field+theory}{nLab}]{perturbative_quantum_field_theory}, due to Brunetti, D\"utsch, Fredenhagen, Hollands, Rejzner, Wald, and others, is (mathematically precise, realistic) pQFT based on causal perturbation theory \cite{steinbook71}, \cite{ep73roleofloc}, \cite{MR1359058}, due to St\"uckelberg, Bogoliubov, Steinmann, Epstein, Glaser, Stora, and others. See \cite[Foreword]{dutsch2019perturbative} for an account of the history. Following \cite{Slavnov78}, \cite{klaus2000micro}, \cite{dutfred00}, in which one takes the algebraic adiabatic limit to handle \hbox{IR-divergences}, pAQFT satisfies the Haag-Kastler axioms of AQFT, but with \hbox{$\text{C}^\ast$-algebras} replaced by formal power series $\ast$-algebras, reflecting the fact that pQFT deals with Taylor series approximations. In this paper, we show that the construction and structure of these formal power series algebras is naturally described in terms of gebra theory internal to species. 

For simplicity, we restrict ourselves to the Klein-Gordan real scalar field on Minkowski spacetime $\cX\cong \bR^{p,1}$, $p\in \bN$ (pAQFT may be applied in more general settings, see e.g. \cite{MR2455327}). Therefore for us, an off-shell field configuration $\Phi$ is a smooth function
\[
\Phi:\cX\to \bR
,\qquad
x\mapsto \Phi(x)
.\] 
%and to be on-shell, $\Phi$ must be in the kernel of the Klein-Gordan differential operator 
%\[
%\square + (m\text{c}/\hbar)^2, \qquad  m\in \bR_{\geq 0}
%.\] 
In particular, we do not impose conditions on the asymptotic behaviour of $\Phi$ at infinite times. Let $\mathcal{F}_{\text{loc}}$ denote the space of local observables $\ssA\in \mathcal{F}_{\text{loc}}$; these are functionals of field configurations which are obtained by integrating polynomials in $\Phi$ and its derivatives against bump functions on $\cX$. Let $\mathcal{F}$ denote the commutative $\ast$-algebra of microcausal polynomial observables $\emph{\textsf{O}}\in \mathcal{F}$; these are polynomial functionals of field configurations satisfying a \hbox{microlocal-theoretic} condition known as microcausality, with multiplication the pointwise multiplication of functionals, sometimes called the normal-ordered product. Then $\mathcal{F}[[\hbar]]$ is a formal power series $\ast$-algebra in formal Planck's constant $\hbar$, called the (abstract, off-shell) Wick algebra, with multiplication the Moyal star product for the Wightman propagator $\Delta_{\text{H}}$ of the Klein-Gordan field
\[
\mathcal{F}[[\hbar]] \otimes \mathcal{F}[[\hbar]] \to \mathcal{F}[[\hbar]]
,\qquad
\emph{\textsf{O}}_1\otimes \emph{\textsf{O}}_2 \mapsto \emph{\textsf{O}}_1 \star_{\text{H}}\! \emph{\textsf{O}}_2
,\] 
sometimes called the operator product. 

%The calculation of the Moyal star product involves the pointwise multiplication of generalized functions, which by H\"ormander's criterion \cite[Theorem 8.2.10]{horm90} is not well-defined without microcausality. 

Perhaps the most fundamental Hopf monoid of Aguiar-Mahajan's theory is the cocommutative Hopf algebra\footnote{\ we say `algebra' and not `monoid' since vector species form a linear category} of compositions $\Sig$, see \autoref{hopfofsetcomp}, which is a Hopf monoid internal to vector species defined with respect to the Day convolution. (More familiar is perhaps a certain decategorification of $\Sig$, which is the graded Hopf algebra of noncommutative symmetric functions $\textbf{NSym}$, see \cite[Section 17.3]{aguiar2010monoidal}.) A composition $F$ of $I$ is a surjective function of the form
\[
F:I\to   \{1,\dots,k\}   
,\qquad 
\text{for some} \quad  k\in \bN
.\] 
The ordering $1>\dots>k$ is understood, so that $F$ models the $k^{\text{th}}$ ordinal with $I$-marked points. We let $S_j=F^{-1}(j)$, called the lumps of $F$, and write $F=(S_1,\dots, S_k)$. Each component $\Sig[I]$ is the space of formal linear combinations of compositions $F$ of $I$,
\[
\Sig[I]  =  \Big \{  \mathtt{a}=\sum_{F} c_F \tH_F \ \big | \  c_F\in \bC  \Big \}
.\] 
The multiplication 
\[
\mu_{S,T}: \Sig[S]\otimes \Sig[T]\to \Sig[I]
,\qquad
\tH_F\otimes \tH_G\mapsto \tH_{FG}
\] 
is the linearization of concatenating compositions (`gluing' via ordinal sum), and the comultiplication   
\[\Delta_{S,T}:\Sig[I] \to\Sig[S]\otimes \Sig[T]
,\qquad
\tH_F \mapsto \tH_{F|_S} \otimes \tH_{F|_T}
\] 
is the linearization of restricting compositions to subsets (`forgetting marked points'), where $S\sqcup T=I$. 

Aspects of $\Sig$ have appeared in the physics literature as follows. Firstly, \hbox{Epstein-Glaser-Stora's} algebra of proper sequences \cite[Section 4.1]{epstein1976general} is the action of $\Sig$ on itself by Hopf powers, called the Tits product \cite[Section 13]{aguiar2013hopf}, going back to Tits \cite{Tits74}. Secondly, the primitive part $\Zie=\mathcal{P}(\Sig)$\footnote{\ the name `Zie' comes from \cite{aguiar2017topics}}, which is a Lie algebra internal to species, is essentially the Steinmann algebra from e.g. \cite[Section 6]{Ruelle}, \cite[Section III.1]{bros}. More precisely, the Steinmann algebra is a graded Lie algebra based on the structure map of the adjoint realization of $\Zie$, see \autoref{sec:Ruelle's Identity and the GLZ Relation}. Thirdly and fourthly, and outside the scope of this paper, see below regarding work of Losev-Manin and Feynman integrals.

The central idea of this paper is to formalize the construction of a system of interacting \hbox{time-ordered} products in causal perturbation theory as the construction of a homomorphism $\widetilde{\text{T}}$ of algebras internal to species of the form
\[       
\widetilde{\text{T}}:\Sig  \otimes \textbf{E}_{\mathcal{F}_{\text{loc}}[[\hbar]]}     
\to 
\textbf{U}_{\mathcal{F}[[\hbar, \formg]]} 
.\]
We describe this construction in a clean abstract setting in \autoref{sec:T-Products, Generalized T-Products, and Generalized R-Products}, and then specialize to QFT in \autoref{sec:Time-Ordered Products}. Here, $\otimes$ is the Hadamard monoidal product (=componentwise tensoring), $\textbf{E}_{\mathcal{F}_{\text{loc}}[[\hbar]]}$ is the species given by $I\mapsto (\mathcal{F}_{\text{loc}}[[\hbar]])^{\otimes I}$, and $\textbf{U}_{\mathcal{F}[[\hbar, \formg]]}$ is the algebra in species which has the Wick algebra, with formal coupling constant $\formg$ adjoined, in each $I$-component, 
\[
\textbf{E}_{\mathcal{F}_{\text{loc}}[[\hbar]]}[I] = (\mathcal{F}_{\text{loc}}[[\hbar]])^{\otimes I} ,\qquad  \textbf{U}_{\mathcal{F}[[\hbar, \formg]]}[I] = \mathcal{F}[[\hbar, \formg]]
.\]
It follows that the data of a system of products $\widetilde{\text{T}}$ is equivalently a homomorphism of \hbox{$\bC$-algebras}
\[
\hat{\Sig}(\mathcal{F}_{\text{loc}}[[\hbar]])\to \mathcal{F}[[\hbar, \formg]]
\]
where $\hat{\Sig}(-): \textsf{Vec}\to \textsf{Vec}$ is the analytic endofunctor, or Schur functor, on vector spaces associated to $\Sig$ \cite[Section 19.1.2]{aguiar2010monoidal}.\footnote{\ the hat $\hat{\Sig}$ is meant to suggest a kind of categorified Fourier transform} Decategorified versions of this formalization appear in graded Hopf algebra approaches to pQFT \cite{Brouder10}, \cite[p. 635]{Borcherds10}. In particular, there is an interpretation of the Moyal deformation quantization in terms of Laplace pairings (=coquasitriangular structures) \cite{Fauser01}, \cite[Section 2.4]{Brouder10}. 

Also related is the notion of a Losev-Manin cohomological field theory \cite[Theorem 3.3.1]{losevmanin}, \cite[Definition 1.3]{shadrin2011group}, where finite ordinals are replaced by strings of Riemann spheres glued at the poles, giving a Hopf monoid structure on the toric variety of the permutohedron, and $\Sig$ is replaced by the ordinary homology of this toric variety. The Hopf monoid structure of this toric variety is also central to modern approaches to Feynman integrals \cite[p.6]{MR3713351}, \cite{schultka2018toric}. We shall study this Hopf monoid in future work.

Explicitly, the homomorphism $\widetilde{\text{T}}$ consists of component linear maps
\[       
\widetilde{\text{T}}_I:\Sig[I]  \otimes (\mathcal{F}_{\text{loc}}[[\hbar]])^{\otimes I}
\to 
\mathcal{F}[[\hbar, \formg]]
,\qquad
\tH_F\otimes \ssA_{i_1}\otimes \dots \otimes  \ssA_{i_n} \mapsto  \widetilde{\text{T}}_I(\tH_F\otimes \ssA_{i_1}\otimes \dots \otimes  \ssA_{i_n})
\]
for each finite set $I=\{i_1,\dots, i_n\}$. This homomorphism should also satisfy causal factorization, which says
\[ 
\widetilde{\text{T}}_I( \mathtt{a} \otimes  \ssA_{i_1}\otimes \dots \ssA_{i_n} )
= 
\widetilde{\text{T}}_I(  \! \! \! \!   \underbrace{\mathtt{a} \triangleright \tH_{G}}_{\text{Tits product}}    \! \! \! \! \otimes  \ssA_{i_1}\otimes \dots \otimes  \ssA_{i_n}  )
\qquad  \text{for all} \quad 
\mathtt{a}\in \Sig[I]     
\]
whenever the local observables $\ssA_{i_1}, \dots ,  \ssA_{i_n}$ respect the ordering of $I$ induced by the composition $G$, see \autoref{prob:causalfac}. Additional properties are often included, such as translation equivariance. 

%\begin{figure}[t]
%	\centering
%	\includegraphics[scale=0.7]{tropperm}
%	\caption{The tropical toric variety which is the toric compactification of the type $A$ root space $\bR^I/(1,\dots,1)=(\bT^\times)^I/\bT^\times$ with respect to the braid arrangement fan. We have shown how the usual $\wt{A}_2$-tiling gets distorted in this picture. On the left is a typical subvariety in blue and its corresponding tropical polynomial underneath, and on the right we have shown the characteristic operation $\tH_{(12,3)}\triangleright (-)$ on the underlying Hopf monoid, see \cite[Introduction]{norledge2019hopf}.}
%	\label{fig:tropperm}
%\end{figure}

We can curry $\widetilde{\text{T}}$ with respect to the internal hom $\cH(-,-)$ for the Hadamard product, giving a homomorphism of algebras
\[
\Sig  \to \cH(  \textbf{E}_{\mathcal{F}_{\text{loc}}[[\hbar]]}  ,\textbf{U}_{ \mathcal{F}[[\hbar, \formg]]} )
,\qquad 
\tH_{F}=\tH_{(S_1,\dots, S_k)} \mapsto \widetilde{\text{T}}(S_1)\dots  \widetilde{\text{T}}(S_k)
.\]
The resulting linear maps
\[
\widetilde{\text{T}}(S_1)\dots  \widetilde{\text{T}}(S_k): (\mathcal{F}_{\text{loc}}[[\hbar]])^{\otimes I} \to \mathcal{F}[[\hbar, \formg]]
\]
are called interacting generalized time-ordered products. For each choice of a field polynomial, the curried homomorphism is a `representation' of $\Sig$ as $\mathcal{F}[[\hbar, \formg]]$-valued generalized functions on $\cX^I$, called operator-valued distributions since the Wick algebra is often represented on a Hilbert space. The composition of the time-ordered products $\widetilde{\text{T}}(I)$ with the Hadamard vacuum state
\[
\la - \ra_{0} :\mathcal{F}[[\hbar, \formg]] \to \bC[[\hbar, \formg]]
,\qquad
\emph{\textsf{O}} \mapsto  \emph{\textsf{O}}(\Phi=0)
\]
are then translation invariant $\bC[[\hbar, \formg]]$-valued generalized functions 
\[
\text{G}_I: \cX^I \to     \bC[[\hbar, \formg]]
,\qquad    
(x_{i_1}, \dots, x_{i_n}) \mapsto \text{G}_I(x_{i_1}, \dots, x_{i_n})   \footnote{\ we have used generalized function notation; $\text{G}_I$ is not a single function, but can be represented by a sequence of functions}
\] 
called time-ordered $n$-point correlation functions. After taking the adiabatic limit, and in the presence of vacuum stability, these functions may be interpreted as the probabilistic predictions made by the pQFT of the outcomes of scattering experiments, called scattering amplitudes, see \autoref{sec:scatterung}. However, their values are formal power series in $\hbar$ and $\formg$, and so have to be truncated. 

Central to Aguiar-Mahajan's work is the interpretation of $\Sig$ (and other Hopf monoids) in terms of the geometry of the type $A$ reflection hyperplane arrangement, called the (essentialized) braid arrangement
\[
\text{Br}[I]=\big\{  \{  x_{i_1}-x_{i_2}=0  \}  \subseteq 
 \! \! \!   \! \! \!  \! 
 \underbrace{\bR^I/\bR \twoheadleftarrow \bR^I}_{\text{quotient by translations}}
\! \! \!  \! \! \!  \! 
 :  (i_1,i_2)\in I^2, \ i_1 \neq i_2  \big   \}
.\]
In causal perturbation theory, the braid arrangement appears as the space of time components of configurations \hbox{$\cX^I$} modulo translational symmetry \cite[Section 2]{Ruelle}, and the reflection hyperplanes are the coinciding interaction points. Every real hyperplane arrangement $\text{A}$ has a corresponding adjoint hyperplane arrangement $\text{A}^\vee$ \cite[Section 1.9.2]{aguiar2017topics}. The free vector space $\bR I$ on $I$ is naturally $\Hom(\bR^I,\bR)$, and so the adjoint of the braid arrangement is given by
\[
\text{Br}^\vee[I]=\bigg\{  \Big\{ \sum_{i\in S} x_i=\sum_{i\in T} x_i=0\Big  \}  \subseteq \underbrace{\Hom(\bR^I/\bR,\bR)\hookrightarrow \bR I}_{\text{sum-zero subspace}}  :  (S,T)\in 2^I,\   S,T\neq \emptyset \bigg  \}
.\]
In causal perturbation theory, the adjoint braid arrangement appears as the space of energy components \cite[Section 2]{Ruelle}, and the hyperplanes correspond to subsets going `on-shell'. The spherical representation of the adjoint braid arrangement is called the Steinmann sphere, or Steinmann planet, e.g. \cite[Figure A.4]{epstein2016}. The chambers of the adjoint braid arrangement are indexed by combinatorial gadgets called cells $\cS$ \cite[Definition 6]{epstein1976general}, also known as maximal unbalanced families \cite{billera2012maximal} and positive sum systems \cite{MR3467341}.

The primitive part Lie algebra $\Zie=\mathcal{P}(\Sig)$ (together with its dual Lie coalgebra $\Zie^\ast$) has a natural geometric realization over the adjoint braid arrangement \cite[Section 6]{Ruelle}, \cite[\href{https://www.youtube.com/watch?v=fUnr0f6mV4c}{Lecture 33}]{oc17}, \cite{lno2019}, \cite{norledge2019hopf}, which results in cells $\cS$ corresponding to certain special primitive elements $\mathtt{D}_\cS\in \Zie[I]$, see \autoref{adjoint}. The special elements were named Dynkin elements by Aguiar-Mahajan \cite[Section 14.1 and 14.9.8]{aguiar2017topics}. It is shown in \cite{norledge2019hopf} that the Dynkin elements span $\Zie$, but they are not linearly independent. The relations which are satisfied by the Dynkin elements are known as the Steinmann relations \cite[Equation 44]{steinmann1960}, see \autoref{stein}, first studied by Steinmann in settings where $\Sig$ is represented as operator-valued distributions. More recently, they have been studied in the context scattering amplitudes, where they appear to be related to cluster algebras \cite{drummond2018cluster}, \cite{caron2019cosmic}, \cite{Caron-Huot:2020bkp}. 

If we restrict a curried system of interacting generalized time-ordered products to the primitive part $\Zie$, then we obtain a Lie algebra homomorphism
\[       
\Zie\to\cH(\textbf{E}_{\mathcal{F}_{\text{loc}}[[\hbar]]},\textbf{U}_{\mathcal{F}[[\hbar, \formg]]}), \qquad \mathtt{D}_\cS \mapsto \widetilde{\text{R}}_\cS
.\]
The operator-valued distributions $\widetilde{\text{R}}_\cS$ which are the images of the Dynkin elements $\mathtt{D}_\cS$ are the interacting generalized retarded products of the system, see e.g. \cite{steinmann1960}, \cite{Huz1}, \cite[Equation 79]{ep73roleofloc}. In this paper, we give an exposition of the Steinman algebra and Steinmann relations in \autoref{sec:Steain}, \autoref{adjoint} and \autoref{stein}.

Let $\textbf{L}\hookrightarrow \Sig$ be the Hopf subalgebra of linear orders (=compositions with singleton lumps), and let $\textbf{E}^\ast\hookrightarrow \Sig$ be the subcoalgebra of compositions with one lump. Then we have the dictionary in \autoref{dic} between products/vacuum expectation values and elements of $\Sig$. In the commutative setting before Moyal deformation quantization, the species $\textbf{X}$ and $\textbf{E}$ are similarly related to the smeared field and polynomial observables, see \autoref{sec:obs}.

\begin{figure}[t] 
		\begin{tabular}{|c|c|c|c|}
		\hline
		&
		\begin{tabular}{@{}c@{}}spanning set\end{tabular}&
		\begin{tabular}{@{}c@{}}operator-valued distributions\end{tabular}&
		\begin{tabular}{@{}c@{}}vacuum expectation values\end{tabular}\\ \hline
		
		\begin{tabular}{@{}c@{}}$\textbf{E}^\ast$\end{tabular}&   
		\begin{tabular}{@{}c@{}}universal series\\$\mathtt{G}_I$\end{tabular}&
		\begin{tabular}{@{}c@{}}time-ordered product\\ $\text{T}(I)$\end{tabular}&
		\begin{tabular}{@{}c@{}}time-ordered $n$-point\\ function\end{tabular}\\ \hline
		
		\begin{tabular}{@{}c@{}}$\textbf{L}$\end{tabular}&
		\begin{tabular}{@{}c@{}}$\tH$-basis linear orders\\$\tH_\ell$\end{tabular}&
		\begin{tabular}{@{}c@{}}$\text{T}(i_1)\dots \text{T}(i_n)$\end{tabular}& 
		\begin{tabular}{@{}c@{}}Wightman $n$-point\\ functions \end{tabular}\\ \hline
		
		\begin{tabular}{@{}c@{}}$\Sig$\end{tabular}&
		\begin{tabular}{@{}c@{}}$\tH$-basis set compositions\\$\tH_F$\end{tabular}&
		\begin{tabular}{@{}c@{}}generalized time-ordered products\\$\text{T}(S_1)\dots \text{T}(S_k)$\end{tabular}&
		\begin{tabular}{@{}c@{}}generalized time-ordered\\ functions\end{tabular}\\ \hline
		
		\begin{tabular}{@{}c@{}}$\Zie$\end{tabular}&
		\begin{tabular}{@{}c@{}}Dynkin elements\\$\mathtt{D}_\cS$\end{tabular}&
		\begin{tabular}{@{}c@{}}generalized retarded products\\$\text{R}_\cS$\end{tabular}&
		\begin{tabular}{@{}c@{}}generalized retarded\\ functions\end{tabular}\\ \hline
	\end{tabular}
	\caption{Dictionary between products/vacuum expectation values and elements of the Hopf algebra $\Sig$.}
	\label{dic} 
\end{figure}

In \autoref{sec:Perturbation of T-Products by Steinmann Arrows} and \autoref{sec:Interactions}, we formalize the \emph{perturbation} of time-ordered products in casual perturbation theory as follows. Our starting point is a fully normalized system of generalized \hbox{time-ordered} products, that is a homomorphism of algebras
\[
\text{T}:\Sig\otimes \textbf{E}_{\mathcal{F}_\text{loc} [[\hbar]]}\to\textbf{U}_{\mathcal{F}((\hbar))}  
\]
satisfying causal factorization, and such that the singleton components $\text{T}_{\{i\}}$ are the natural inclusion 
\[
\mathcal{F}_\text{loc} [[\hbar]]\hookrightarrow  \mathcal{F}((\hbar))
,\qquad
\ssA \mapsto \, \,  :\! \ssA :  .
\]
The corresponding operator-valued distributions are determined everywhere on $\cX^I$ by causal factorization, apart from on the fat diagonal (=coinciding interaction points). In particular, off the fat diagonal, the time-ordered products $\text{T}(I)$ are given by the Moyal star product $\star_{\text{F}}$ with respect to the Feynman propagator $\Delta_{\text{F}}$ for the Klein-Gordon field. The terms of the product $\star_{\text{F}}$ may be encoded in finite multigraphs, i.e. Feynman graphs. The remaining inherent ambiguity means one has to make choices when extending the $\text{T}(I)$ to the fat diagonal, and these choices form a torsor of the \hbox{St\"uckelberg-Petermann} renormalization group. This is Stora's elaboration \cite{stora16}, \cite{stora1993differential}, \cite{klaus2000micro} on \hbox{St\"uckelberg-Bogoliubov-Epstein-Glaser} normalization \cite{ep73roleofloc}, which constructs the $\text{T}(I)$ inductively in $n=|I|$. We leave species-theoretic aspects of renormalization, and possible connections to \hbox{Connes-Kreimer} theory \cite{MR1845168}, \cite{Bondia00}, \cite{Kreimer05}, \cite{FredHopf14}, to future work.

%\footnote{\ alternatively, to define the homomorphism $\text{T}$ one can normalize the operator-valued distributions which are the images of the Dynkin elements $\mathtt{D}_i\in \Sig[Y\sqcup \{i\}]$, as in \cite{steinbook71}, \cite{dutfredretard04}, \cite{dutsch2019perturbative}, called total retarded products $\text{R}(Y,i)$; unlike $\tH_{(I)}$, the $\mathtt{D}_i$ do not freely generate $\Sig$, and so one must include the GLZ relation \cite[Equation 11]{GLZ1957} (\autoref{sec:Ruelle's Identity and the GLZ Relation})} 

In the original formulation by Tomonaga, Schwinger, Feynman and Dyson, would-be \hbox{time-ordered} products are obtained by informally multiplying Wick algebra products by step functions, which is in general ill-defined by H\"ormander's criterion. This leads to the divergence of individual terms of the formal power series, called UV-divergences. Then informal methods are used to obtain finite values from these infinite terms \cite[Preface and Section 4.3]{MR1359058}.

The exponential species $\textbf{E}$, given by $\textbf{E}[I]=\bC$ and $1_\bC\in \textbf{E}[I]$ denoted $\tH_I$, has the structure of an algebra in species by linearizing taking unions of sets,
\[
\mu_{S,T}: \textbf{E}[S]\otimes \textbf{E}[T]\to \textbf{E}[I]
,\qquad
\tH_S\otimes \tH_T\mapsto \tH_{I}
.\] 
An $\textbf{E}$-module $\textbf{m}=(\textbf{m},\rho)$ is an associative and unital morphism 
\[\rho:\textbf{E}\bigcdot\textbf{m}\to \textbf{m}\] 
for $\textbf{m}$ a species. Moreover, taking the inverse of $\mu_{S,T}$ as the comultiplication turns $\textbf{E}$ into a connected (co)commutative bialgebra, and so the category of $\textbf{E}$-modules $\textsf{Rep}(\textbf{E})$ is a symmetric monoidal category with monoidal product the Cauchy product of $\textbf{E}$-modules. In particular, we may consider Hopf/Lie algebras internal to $\textsf{Rep}(\textbf{E})$, which we call Hopf/Lie \hbox{$\textbf{E}$-algebras}.

The retarded $Y\downarrow(-)$ and advanced $Y\uparrow(-)$ Steinmann arrows are (we formalize as) raising operators on $\Sig$, whose precise definition is due to \hbox{Epstein-Glaser-Stora} \cite[p.82-83]{epstein1976general}. They define two $\textbf{E}$-module structures on $\Sig$, 
\[
\textbf{E}\bigcdot \Sig \to \Sig 
,\quad
\tH_Y \otimes \tH_F \, \mapsto\,  Y  \downarrow\tH_F
\qquad \text{and} \qquad 
\textbf{E}\bigcdot \Sig \to \Sig
,\quad
\tH_Y \otimes \tH_F \, \mapsto\,  Y  \uparrow \tH_F
.\]
See \autoref{sec:The Steinmann Arrows}. In particular, the retarded arrow is generated by putting $\{\ast\} \downarrow \tH_{(I)}= -\tH_{(\ast, I)} +\tH_{(\ast I)}$.\footnote{\ $(\ast I)$ denotes the composition of $\{ \ast \}\sqcup I$ which has a single lump} Then
\[
Y\! \downarrow \tH_{(I)}= \underbrace{\sum_{Y_1\sqcup Y_2=Y} \mu_{Y_1, Y_2\sqcup I}\big ( \text{s}(\tH_{(Y_1)}) \otimes \tH_{(Y_2\sqcup I)} \big )}_{\text{denoted $\mathtt{R}_{(Y;I)}$}}
\]
where $\text{s}:\Sig \to \Sig$ is the antipode of $\Sig$. The Steinmann arrows were first studied by Steinmann \cite[Section 3]{steinmann1960}, where $\Sig$ is represented as operator-valued distributions. Here, the \hbox{operator-valued} distribution which corresponds to $\mathtt{R}_{(Y;I)}\in \Sig[Y\sqcup I]$ is called the retarded product $\text{R}(Y;I)$.\footnote{\ note that some authors, e.g. \cite{dutsch2019perturbative}, call $\text{R}(Y;i)$ the retarded product, and then call $\text{R}(Y;I)$ the generalized retarded product}

Since $\{\ast\} \downarrow (-)$ is a commutative biderivation of $\Sig$ (\autoref{steinmannarrowaredercoder}), the retarded Steinmann arrow gives $\Sig$ the structure of a Hopf $\textbf{E}$-algebra, and $\Zie$ the structure of a Lie $\textbf{E}$-algebra (similarly for the advanced arrow). There is an interesting description of these Lie $\textbf{E}$-algebras in terms of the adjoint braid arrangement, see \autoref{sec:The Steinmann Arrows and Dynkin Elements}. The Steinmann arrows are ``two halves'' of the restricted adjoint representation $\textbf{L}\bigcdot \Sig \to \Sig$ of $\Sig$, which is reflected in \cite[Equation 13]{steinmann1960}. This directly corresponds to how the retarded $\Delta_-$ and advanced $\Delta_+$ propagators are two halves of the causal propagator $\Delta_{\text{S}}=\Delta_+ - \Delta_-$. 

Let $\cH^{\bigcdot}(-,-)$ denote the internal hom for the Cauchy product of species, and let 
\[
(-)^{\textbf{E}}=\cH^{\bigcdot} ( \textbf{E} , -)
.\] 
See \autoref{Coalgebras} for a more explicit definition.  See also \cite[Section 2]{norledge2020species} for more details here regarding the differentiation between the $\formj$-colored sets $I$ (physically, the source field) and the $\formg$-colored sets $Y$ (physically, the coupling constant). Then $(-)^{\textbf{E}}$ is an endofunctor on species, which is lax monoidal with respect to the Cauchy product. Therefore $\Sig^{\textbf{E}}$ is naturally an algebra, with multiplication inherited from $\Sig$. Then, by currying the retarded Steinmann action $\textbf{E}\bigcdot \Sig \to \Sig$, we obtain a homomorphism $\Sig \to \Sig^{\textbf{E}}$. Similarly for the setting with decorations, given a choice of adiabatically switched interaction action functional $\ssS_{\text{int}}\in \mathcal{F}_{\text{loc}}[[\hbar]]$, after acting with the retarded Steinmann arrows and currying, we obtain the homomorphism
\begin{align*}
\Sig \otimes \textbf{E}_{\mathcal{F}_{\text{loc}}[[\hbar]]} &\to    (\Sig \otimes \textbf{E}_{\mathcal{F}_{\text{loc}}[[\hbar]]})^{\textbf{E}}\\[6pt]
\tH_F\otimes \ssA_{i_1}\otimes \dots \otimes \ssA_{i_n} &\mapsto \,    \sum_{r=0}^\infty   \underbrace{\downarrow\dots  \downarrow}_{\text{$r$ times}} \tH_F  \otimes  \underbrace{\ssS_{\text{int}}\otimes \dots \otimes \ssS_{\text{int}}}_{\text{$r$ times}} \,  \otimes\,    \ssA_{i_1}\otimes \dots \otimes \ssA_{i_n}.
\end{align*}
Compare this with the formalism for creation-annihilation operators in \cite[Chapter 19]{aguiar2010monoidal}. Then, finally, the corresponding system of perturbed interacting time-ordered products $\widetilde{\text{T}}$ is given by composing this homomorphism with the image of $\text{T}$ under the endofunctor $(-)^{\textbf{E}}$,
\[
\widetilde{\text{T}} : \Sig\otimes \textbf{E}_{\mathcal{F}_\text{loc} [[\hbar]]} \to   (\Sig\otimes \textbf{E}_{\mathcal{F}_\text{loc} [[\hbar]]})^{\textbf{E}}   \xrightarrow{\text{T}^{\textbf{E}}}        (\textbf{U}_{ \mathcal{F}((\hbar))})^{\textbf{E}} \cong   \textbf{U}_{\mathcal{F}((\hbar))[[\formg]]}
.\]
See \autoref{sec:Perturbation of T-Products by Steinmann Arrows}. It is a theorem of pAQFT that this does indeed land in $\textbf{U}_{\mathcal{F}[[\hbar,\formg]]}$.

Finally, in \autoref{sec:T-Exponentials} and \autoref{sec:Time-Ordered Products}, we formalize S-matrices, or time-ordered exponentials, as follows. Let $\Hom(-,-)$ denote the external hom for species, which lands in vector spaces $\sfVec$. We let 
\[
\mathscr{S}(-)=\Hom(\textbf{E},-)
.\] 
This is lax monoidal with respect to the Cauchy product. In the presence of a generic system of products on an algebra $\textbf{a}$,
\[
\varphi:\textbf{a}\otimes \textbf{E}_V\to \textbf{U}_\cA,
\] 
series $\mathtt{s}\in \mathscr{S}(\textbf{a})$ of $\textbf{a}$
\[\mathtt{s}:\textbf{E}\to\textbf{a}
,\qquad 
\tH_I\mapsto \mathtt{s}_I\] 
induce $\mathscr{S}(\textbf{U}_\cA)\cong \cA[[\formj]]$-valued functions $\mathcal{S}_{\mathtt{s}}$ on $V$ as follows,
\[
\mathcal{S}_{\mathtt{s}} :    V \to \cA[[\formj]]
,\qquad
\ssA \mapsto   \mathcal{S}_{\mathtt{s}}(\formj\! \ssA) :=  \sum_{n=0}^{\infty}  \dfrac{\formj^n}{n!}  \varphi_{n} ( \mathtt{s}_{n} \otimes \underbrace{\ssA\otimes \dots \otimes \ssA}_{\text{$n$ times}})
.\]
If $\varphi$ is a homomorphism of algebras, then 
\[
\mathcal{S}_{(-)}: \mathscr{S}(\textbf{a})\to \text{Func}(V, \cA[[\formj]])
\] 
is a homomorphism of $\bC$-algebras. As a basic example, if we put $\textbf{a}=\textbf{E}$, $\cA=C^\infty(V^\ast)$, and set $\formj=1$ at the end, then one can recover the classical exponential function in this way. 

For $c\in \bC$, the so-called (scaled) universal series $\mathtt{G}(c)$ of $\Sig$ is given by sending each finite set to the (scaled) composition with one lump,
\[   
\mathtt{G}(c):  \textbf{E} \to  \Sig
,\qquad  
\tH_{I}\mapsto \mathtt{G}(c)_{I}:= c^n\,  \tH_{(I)}    
.\]
If we set $c=1/\text{i}\hbar$, then the function $\mathcal{S}=\mathcal{S}_{\mathtt{G}(1/\text{i}\hbar)}$ above for a fully normalized system of generalized time-ordered products $\text{T}:\Sig\otimes \textbf{E}_{\mathcal{F}_\text{loc} [[\hbar]]}\to\textbf{U}_{\mathcal{F}((\hbar))}$ recovers the usual perturbative S-matrix scheme of pAQFT,
\[
\mathcal{S}:\mathcal{F}_{\text{loc}}[[\hbar]]\to\mathcal{F}((\hbar))[[\formj]]
,\qquad
\ssA \mapsto   \mathcal{S}(\formj\!  \ssA) =  \sum_{n=0}^{\infty}  \bigg( \dfrac{1}{\text{i} \hbar} \bigg)^n \dfrac{\formj^n}{n!}  \text{T}_{n} ( \tH_{(n)} \otimes \underbrace{\ssA\otimes \dots \otimes \ssA}_{\text{$n$ times}})
.\] 
The image of $\mathcal{S}(\formj\! \ssA)$ after applying perturbation by the retarded Steinmann arrow and a choice of interaction $\ssS_{\text{int}}\in \mathcal{F}_{\text{loc}}[[\hbar]]$ is
\[
\mathcal{Z}_{\formg \ssS_{\text{int}}}(\formj\! \ssA)
=
\sum_{n=0}^\infty \sum_{r=0}^\infty 
\bigg(\dfrac{1}{\text{i}\hbar}\bigg)^{\! r+n}
\dfrac{\formg^{r} \formj^n}{r!\, n!}\,  \text{R}_{r;n} (\underbrace{\ssS_{\text{int}}\otimes \dots \otimes \ssS_{\text{int}}}_{\text{$r$ times}}\,   ;\,  \underbrace{\ssA\otimes \dots \otimes \ssA}_{\text{$n$ times}} )    
\]
where, by our previous expression for $\mathtt{R}_{(Y;I)}=Y\downarrow \tH_{(I)}$ (and letting $\overline{\text{T}}$ denote the precomposition of $\text{T}$ with the antipode of $\Sig \otimes \textbf{E}_{\mathcal{F}_{\text{loc}}[[\hbar]]}$), we have
\[
\text{R}_{Y;I}(\ssS_{\text{int}}^{\, Y};\ssA^I)
= 
\text{T}_{Y\sqcup I}(Y\downarrow \tH_{(I)} \otimes   \ssS_{\text{int}}^{\, Y} \otimes \ssA^I)
=
\sum_{Y_1 \sqcup Y_2=Y}  \overline{\text{T}}_{Y_1}(\ssS_{\text{int}}^{\, Y_1}) \star_{\text{H}} \text{T}_{Y_2\sqcup I}( \ssS_{\text{int}}^{\, Y_2}\otimes  \ssA^I)
.\]
Then, since
\[
\mathcal{S}_{(-)}:\mathscr{S}(\Sig)\to \text{Func}\big (\mathcal{F}_{\text{loc}}[[\hbar]] , \mathcal{F}((\hbar))[[\formg]]\big )
\] 
is a homomorphism of $\bC$-algebras, it follows that $\mathcal{Z}_{\formg \ssS_{\text{int}}}$ is given by
\[
\mathcal{Z}_{\formg \ssS_{\text{int}}}(\formj\! \ssA)
=
\mathcal{S}^{-1}( \formg \ssS_{\text{int}})\star_{\text{H}} \mathcal{S}(\formg \ssS_{\text{int}} +\formj\! \ssA )
.\]
This is the generating function, or partition function, for time-ordered products of interacting field observables, see e.g. \cite[Section 8.1]{ep73roleofloc}, \cite[Section 6.2]{dutfred00}, going back to Bogoliubov \cite[Chapter 4]{Bogoliubov59}. In this paper, we arrive at the generating function $\mathcal{Z}_{\formg \ssS_{\text{int}}}$ through purely Hopf-theoretic considerations. However, it was originally motivated by attempts to make sense of the path integral synthetically. For some recent developments, see \cite{collini2016fedosov}, \cite{MR4109798}.

%%%%%%%%%%%%%%%%%%%%%%%%%%%%%%%%%%%%%%%%%%%%%%%%%%%%%%%%%%%%%%%%%%%%%%%%
\subsection*{Structure.} 
%%%%%%%%%%%%%%%%%%%%%%%%%%%%%%%%%%%%%%%%%%%%%%%%%%%%%%%%%%%%%%%%%%%%%%%%

This paper is divided into two parts. In part one, we focus on developing theory for the Hopf algebra of compositions $\Sig$ and its primitive part $\Zie$. In part two, we specialize to pAQFT for the case of a real scalar field on Minkowski spacetime.

%This paper is divided into three parts. In part one we develop general theory of species, in part two we focus on the theory for the Hopf algebra of compositions $\Sig$ and its primitive part $\Zie$, and in part three we specialize to pAQFT for the case of a real scalar field on Minkowski spacetime.

%By translation invariance, or equivalently momentum conservation, the generalized time-ordered functions are translation invariant. If we restrict to the time component, then we obtain a realization of $\Sig$ as distributions on the braid arrangement, or functions on the adjoint braid arrangement. 

%%%%%%%%%%%%%%%%%%%%%%%%%%%%%%%%%%%%%%%%%%%%%%%%%%%%%%%%%%%%%%%%%%%%%%%%
\subsection*{Acknowledgments.} 
%%%%%%%%%%%%%%%%%%%%%%%%%%%%%%%%%%%%%%%%%%%%%%%%%%%%%%%%%%%%%%%%%%%%%%%%

We thank Adrian Ocneanu for his support and useful discussions. This paper would not have been written without Nick Early's discovery that certain relations appearing in Ocneanu's work were known in quantum field theory as the Steinmann relations. We thank Yiannis Loizides and Maria Teresa Chiri for helpful discussions during an early stage of this project. We thank Arthur Jaffe for his support, useful suggestions, and encouragement to pursue this topic. We thank Penn State maths department for their continued support.

%%%%%%%%%%%%%%%%%%%%%%%%%%%%%%%%%%%%%%%%%%%%%%%%%%%%%%%%%%%%%%%%%%%%%%%%
%%%%%%%%%%%%%%%%%%%%%%%%%%%%%%%%%%%%%%%%%%%%%%%%%%%%%%%%%%%%%%%%%%%%%%%%
%%%%%%%%%%%%%%%%%%%%%%%%%%%%%%%%%%%%%%%%%%%%%%%%%%%%%%%%%%%%%%%%%%%%%%%%
%%%%%%%%%%%%%%%%%%%%%%%%%%%%%%%%%%%%%%%%%%%%%%%%%%%%%%%%%%%%%%%%%%%%%%%%
%%%%%%%%%%%%%%%%%%%%%%%%%%%%%%%%%%%%%%%%%%%%%%%%%%%%%%%%%%%%%%%%%%%%%%%%
\part{Hopf Monoids}
%%%%%%%%%%%%%%%%%%%%%%%%%%%%%%%%%%%%%%%%%%%%%%%%%%%%%%%%%%%%%%%%%%%%%%%%
%%%%%%%%%%%%%%%%%%%%%%%%%%%%%%%%%%%%%%%%%%%%%%%%%%%%%%%%%%%%%%%%%%%%%%%%
%%%%%%%%%%%%%%%%%%%%%%%%%%%%%%%%%%%%%%%%%%%%%%%%%%%%%%%%%%%%%%%%%%%%%%%%
%%%%%%%%%%%%%%%%%%%%%%%%%%%%%%%%%%%%%%%%%%%%%%%%%%%%%%%%%%%%%%%%%%%%%%%%
%%%%%%%%%%%%%%%%%%%%%%%%%%%%%%%%%%%%%%%%%%%%%%%%%%%%%%%%%%%%%%%%%%%%%%%%

%%%%%%%%%%%%%%%%%%%%%%%%%%%%%%%%%%%%%%%%%%%%%%%%%%%%%%%%%%%%%%%%%%%%%%%%
\section{The Algebras} 
%%%%%%%%%%%%%%%%%%%%%%%%%%%%%%%%%%%%%%%%%%%%%%%%%%%%%%%%%%%%%%%%%%%%%%%%
We recall the Hopf algebra of compositions $\Sig$, together with its Lie algebra of primitive elements $\Zie\hookrightarrow \Sig$. We show that $\Sig$ and $\Zie$ are naturally algebras over the exponential species $\textbf{E}$. This will be a \hbox{species-theoretic} formalization of mathematical structure discovered by Steinmann \cite{steinmann1960} and \hbox{Epstein-Glaser-Stora} \cite{epstein1976general}, which, combined with a certain `perturbation of systems of products' construction using the $\textbf{E}$-action, will recover the perturbative construction of interacting fields in pAQFT, as in \cite[Section 8.1]{ep73roleofloc}, \cite[Section 6.2]{dutfred00}, going back to Bogoliubov \cite[Chapter 4]{Bogoliubov59}. 
  
%%%%%%%%%%%%%%%%%%%%%%%%%%%%%%%%%%%%%%%%%%%%%%%%%%%%%%%%%%%%%%%%%%%%%%%%
\subsection{Compositions}  \label{comp}
%%%%%%%%%%%%%%%%%%%%%%%%%%%%%%%%%%%%%%%%%%%%%%%%%%%%%%%%%%%%%%%%%%%%%%%%
Let $I$ be a finite set of cardinality $n$. We think of $I$ as having `color' $\formj$ (physically, the source field). As a particular example of the set $I$, we have the set of integers $[n]:=\{ 1, \dots, n \}$ (formally, we have picked a section of the decategorification functor $I\mapsto n$). For $k\in \bN$, let
\[
(k):=\{1,\dots,k\}
\]
equipped with the ordering \hbox{$1>\dots> k$}. A \emph{composition} $F$ of $I$ of \emph{length} $l(F)=k$ is a surjective function $F:I\to (k)$. The set of all compositions of $I$ is denoted $\Sigma[I]$,
\[  
\Sigma[I]:=  \bigsqcup_{k\in \bN}  \big \{  \text{surjective functions}\ F:I \to (k) \big\}
.\]
We often denote compositions by $k$-tuples 
\[  
F=  (S_1, \dots, S_k)
\]
where $S_j:= F^{-1}(j)$, $1\leq j \leq k$. The $S_j$ are called the \emph{lumps} of $F$. In particular, we have the length one composition $(I)$ for $I\neq \emptyset$, and the length zero composition $(\, )$ which is the unique composition of the empty set. The \emph{opposite} $\bar{F}$ of $F$ is defined by 
\[
\bar{F}:=(S_k,\dots, S_1), \qquad  \text{i.e.} \quad \bar{F}^{-1}(j)=F^{-1}(k+1-j)
.\]

%We may give $\Sigma[I]$ the structure of a category by including a single morphism $G\to F$ whenever $G\leq F$. 

Given a decomposition $I\! =S\sqcup T$ of $I$ ($S,T$ can be empty), for $F=(S_1, \dots , S_{k})$ a composition of $S$ and $G=(T_1,\dots, T_{l})$ a composition of $T$, their \emph{concatenation} $FG$ is the composition of $I$ given by
\[ 
FG: =   ( S_1, \dots , S_{k},  T_1,\dots, T_{l}  ) 
.\]
For $S\subseteq I$ and $F=(S_1, \dots , S_{k}) \in \Sigma[I]$, the \emph{restriction} $F|_S$ of $F$ to $S$ is the composition of $S$ given by
\[
F|_S:=  (  S_1 \cap S, \dots, S_k\cap S )_+
\]
where $(-)_+$ means we delete any sets from the list which are the empty set.

For compositions $F,G\in \Sigma[I]$, we write $G\leq F$ if $G$ can be obtained from $F$ by iteratively merging contiguous lumps. Given compositions $G\leq F$ with $G=(T_1, \dots, T_l)$, we let
\[
l(F/G):=\prod^{k}_{j=1} l( F|_{T_j} )     
\qquad \text{and} \qquad  
(F/G)!:=\prod^{k}_{j=1} l( F|_{T_j} )!\,    
.\]
%%%%%%%%%%%%%%%%%%%%%%%%%%%%%%%%%%%%%%%%%%%%%%%%%%%%%%%%%%%%%%%%%%%%%%%%
\subsection{The Cocommutative Hopf Monoid of Compositions} \label{hopfofsetcomp}
%%%%%%%%%%%%%%%%%%%%%%%%%%%%%%%%%%%%%%%%%%%%%%%%%%%%%%%%%%%%%%%%%%%%%%%%
Let
\[   
\Sig[I] 
:= \big\{\text{formal $\bC$-linear combinations of compositions of $I$}\big\}
.\]
The vector space $\Sig[I]$ is naturally a right module over the symmetric group on $I$, and these actions extend to a contravariant functor from the category $\textsf{S}$ of finite sets and bijections into the category $\textsf{Vec}$ of vector spaces over $\bC$,
\[
\Sig:\textsf{S}^\text{op} \to \textsf{Vec}
,\qquad 
I \mapsto \Sig[I]
.\] 
For $F$ a composition of $I$, let $\tH_F\in \Sig[I]$ denote the basis element corresponding to $F$. The sets $\{\tH_F: F\in \Sigma[I]\}$ form the \emph{$\tH$-basis} of $\Sig$. 

In general, functors $\textbf{p}:\textsf{S}^\text{op} \to \textsf{Vec}$ are called (complex) \emph{vector species}, going back to Joyal \cite{joyal1981theorie}, \cite{joyal1986foncteurs}. Morphisms of vector species $\eta:\textbf{p}\to \textbf{q}$ are natural transformations; they consist of a linear map $\eta_I:   \textbf{p}[I]\to \textbf{q}[I]$ for each finite set $I$ which commutes with the action of the bijections. When $I=[n]:=\{ 1,\dots,n\}$, we abbreviate $\eta_n:= \eta_{[n]}$.

We equip vector species with the tensor product $\textbf{p}\bigcdot \textbf{q}$ known as the \emph{Cauchy product} \cite[Definition 8.5]{aguiar2010monoidal}, given by
\begin{equation}\label{eq:Cauchy}
\textbf{p}\bigcdot \textbf{q} [I] := \bigoplus_{I=S\sqcup T  }   \textbf{p}[S] \otimes \textbf{q}[T]
.
\end{equation}
This is the Day convolution with respect to disjoint union of sets and tensor product of vector spaces. In this paper, we consider algebraic structures on species which are constructed using this tensor product. In particular, a multiplication on a species $\textbf{p}$ consists of linear maps
\[
\mu_{S,T} : \textbf{p} [S] \otimes \textbf{p}[T] \to \textbf{p}[I] 
\]
and a comultiplication on $\textbf{p}$ consists of linear maps 
\[
 \Delta_{S,T} : \textbf{p}[I] \to  \textbf{p} [S] \otimes \textbf{p}[T]
,\]
where we have a map for each choice of decomposition $I=S\sqcup T$ ($S,T$ can be empty). We can then impose conditions like (co)associativity, see e.g. \cite[Section 1.3]{norledge2020species}.

Following \cite[Section 11]{aguiar2013hopf}, $\Sig$ is a connected\footnote{\ a species $\textbf{p}$ is \emph{connected} if $\textbf{p}[\emptyset]=\bC$} bialgebra, meaning it is naturally equipped with an associative, unital multiplication and a coassociative, counital comultiplication, which are compatible in the sense they satisfy the bimonoid axiom. See \cite[Section 8.3.1]{aguiar2010monoidal} for details. The multiplication and comultiplication are given in terms of the $\tH$-basis by
\[
\mu_{S,T}(\tH_F\otimes \tH_G):=\tH_{FG} \qquad \text{and} \qquad   \Delta_{S,T}  (\tH_F) :=  \tH_{F|_S} \otimes    \tH_{F|_T}
.\]
We sometimes abbreviate $\tH_F \tH_G:= \mu_{S,T}(\tH_F \otimes\tH_G)$. The unit and counit are given by 
\[
\mathtt{1}_{\Sig}:=\tH_{(\, )} \qquad  \text{and} \qquad \epsilon_\emptyset(\tH_{(\, )}):=1_\bC
.\] 
Let
\begin{equation}\label{antipode}
\overline{\tH}_F:= \sum_{G\geq \bar{F}}  (-1)^{l(G)}\, \tH_G 
.
\end{equation}
Then \cite[Theorem 11.38]{aguiar2010monoidal} (in the case $\textbf{q}=\textbf{E}^\ast_+$ and $q=1$) shows that
\begin{equation}\label{eq:inversion relation for reverse time-ordered products}
\sum_{S\sqcup T=I} \tH_{F|_S} \overline{\tH}_{F|_T}
=0
\qquad \text{and} \qquad 
\sum_{S\sqcup T=I}\overline{\tH}_{F|_S}  \tH_{F|_T}
=0
.
\end{equation}
In general, connected bialgebras are automatically Hopf algebras, and it follows from \textcolor{blue}{(\refeqq{eq:inversion relation for reverse time-ordered products})} that the antipode $s:\Sig\to \Sig$ is given by 
\[
\text{s}_I(\tH_F)=\overline{\tH}_F
.\] 
The Hopf algebra $\Sig$ is the free cocommutative Hopf algebra on the positive coalgebra $\textbf{E}^\ast_+$ \cite[Section 11.2.5]{aguiar2010monoidal}, and so $\Sig\cong \textbf{L}\boldsymbol{\circ} \textbf{E}^\ast_+$ where `$\boldsymbol{\circ}$' is plethysm of species and $\textbf{L}\hookrightarrow \Sig$ is the subspecies of singleton lump compositions ($=$linear orders). 

There is a second important basis of $\Sig$, called the \emph{$\tQ$-basis}. The $\tQ$-basis is also indexed by compositions, and is given by  
%\begin{equation}\label{eq:1}
\[ 
\tQ_F:= \sum_{G\geq F}   (-1)^{ l(G)-l(F) }  \dfrac{1}{    l(G/F) }   \tH_G\qquad \text{or equivalently} \qquad  \mathtt{H}_F=: \sum_{G\geq F} \dfrac{1}{( G/F )!}  \tQ_G 
.\]
%\end{equation}
For $S\subseteq I$ and $F\in \Sigma[I]$, we have \emph{deshuffling}
\[F\res_S\, :=  
\begin{cases}
F|_S &\quad  \text{if $S$ is a union of lumps of $F$}\footnote{\ }\\
0\in \Sig[S] &\quad \text{otherwise.}
\end{cases}
\]
\footnotetext{\ not necessarily contiguous}The multiplication and comultiplication of $\Sig$ is given in terms of the $\tQ$-basis by  
\[         
\mu_{S,T} ( \tQ_F\otimes \tQ_G ) =  \tQ_{FG} 
\qquad \text{and} \qquad  
\Delta_{S,T}  (\tQ_F) =  \tQ_{F\res_S} \otimes \,   \tQ_{F\res_T}
.\]

%The \emph{adjoint representation} of $\Sig$ is the $\Sig$-module given by
%\[
%\text{Ad}:\Sig \bigcdot \Sig \to \Sig 
%,\qquad
%\text{Ad}_{Y,I}(\tH_G \otimes \tH_{F})=   \sum_{Y_1\sqcup Y_2=Y} \tH_{G|_{Y_1}} \tH_F \,  \overline{\tH}_{G|_{Y_2}}
%.\]
%One can check that this gives $\Sig$ the structure of a Hopf $\Sig$-algebra (everything goes through in the same way as the classical case).

%We have the following Cauchy product, for which we use slightly different naming conventions for our finite sets,
%\[
%\textbf{E}\bigcdot \Sig [X]  =  \bigoplus_{X= Y\sqcup I} \textbf{E}[Y] \otimes  \Sig [I] = \bigoplus_{I\subseteq X}   \Sig [I]
%.\]
%%%%%%%%%%%%%%%%%%%%%%%%%%%%%%%%%%%%%%%%%%%%%%%%%%%%%%%%%%%%%%%%%%%%%%%%
\subsection{Decorations}\label{Decorations}
%%%%%%%%%%%%%%%%%%%%%%%%%%%%%%%%%%%%%%%%%%%%%%%%%%%%%%%%%%%%%%%%%%%%%%%%

Given a complex vector space $V$, we can use $V$ to `decorate' $\Sig$ in order to obtain an enlarged Hopf algebra $\Sig\otimes \textbf{E}_V$. This goes as follows.

We have the species denoted $\textbf{E}_V$, given by
\[
\textbf{E}_V[I] := V^{\otimes I}=  \! \! \! \! \! \! \! \! \!  \underbrace{V\otimes \dots \otimes V}_{\text{a copy of $V$ for each $i\in I$}}
\! \! \! \! \! \! \! \! \!  \! \!
.\]
The action of bijections is given by relabeling tensor factors. 

\begin{remark}
Notice species of the form $\textbf{E}_V$ are exactly the monoidal functors \hbox{$\textbf{E}_V: \textsf{S}^{\text{op}} \to \textsf{Vec}$}.
\end{remark}

We denote vectors by $\ssA,\ssS\in V$, and we denote simple tensors of $V^{\otimes I}$ by
\[
\ssA_I=\ssA_{i_1}\otimes \cdots \otimes \ssA_{i_n} \in V^{\otimes I}
\] 
where $I=\{ i_1,\dots, i_n\}$. If $\ssA_i=\ssA$ for all $i\in I$, then we write    
\begin{equation}\label{eq:simpletensors} 
	\ssA^{I}:= \ssA\otimes \cdots \otimes \ssA\in V^{\otimes I}
	\qquad \text{and} \qquad 
	\ssA^{ n }:=\ssA^{[n]}\in V^{\otimes [n]}
\end{equation}
where $[n]=\{1,\dots,n\}$ as usual.

We let `$\otimes$' denote the Hadamard product of species, which is given by componentwise tensoring, see e.g. \cite[Section 1.2]{norledge2020species}. Then the species of $V$-\emph{decorated compositions} $\Sig\otimes \textbf{E}_V$ is given by
\[
\Sig\otimes \textbf{E}_V[I] =   \Sig[I]\otimes \textbf{E}_V[I]=   \Sig[I] \otimes V^{\otimes I} 
.\]
Following \cite[Section 8.13.4]{aguiar2010monoidal}, $\Sig\otimes \textbf{E}_V$ is a connected bialgebra, with multiplication given by
\[
\mu_{S,T}\big((\tH_F\otimes\ssA_S) \otimes (\tH_G \otimes \ssA_T)\big)
:=
\tH_F \tH_G \otimes \ssA_S \otimes \ssA_T 
\]
and comultiplication given by 
\[
\Delta_{S,T}(\tH_F\otimes \ssA_I)
:=
(\tH_{F|_S}\otimes{\ssA_{I}}|_S)\otimes(\tH_{F|_T} \otimes {\ssA_{I}}|_T) 
.\]
The unit and counit are given by 
\[
\mathtt{1}_{\Sig\otimes \textbf{E}_V}:=\tH_{(\, )}\otimes 1_\bC\qquad  \text{and}\qquad \epsilon_\emptyset(\tH_{(\, )}\otimes 1_\bC):=1_\bC
.\] 
For $\tH_F\otimes \ssA_I\in \Sig\otimes \textbf{E}_V[I]$, we have
\[
\sum_{S\sqcup T=I} 
\mu_{S,T}\big ((\tH_{F|_S}\otimes {\ssA_{I}}|_S)\otimes (\overline{\tH}_{F|_T}\otimes {\ssA_{I}}|_T)\big) 
=
\underbrace{\sum_{S\sqcup T=I} \tH_{F|_S}  \overline{\tH}_{F|_T}}_{\text{$=0$ by \textcolor{blue}{(\refeqq{eq:inversion relation for reverse time-ordered products})}}} \otimes\,  \ssA_{I}
=0
\] 
and
\[
\sum_{S\sqcup T=I} 
\mu_{S,T}\big ((\overline{\tH}_{F|_S}\otimes {\ssA_{I}}|_S)\otimes (\tH_{F|_T}\otimes {\ssA_{I}}|_T)\big) 
=
\underbrace{\sum_{S\sqcup T=I} \overline{\tH}_{F|_S} \tH_{F|_T}}_{\text{$=0$ by \textcolor{blue}{(\refeqq{eq:inversion relation for reverse time-ordered products})}}}\otimes\,  \ssA_{I}
=0
.\]
It follows that the antipode of $\Sig\otimes \textbf{E}_V$ is given by
\begin{equation}\label{eq:antipode}
\text{s}_I(\tH_F\otimes \ssA_I)= \overline{\tH}_F \otimes \ssA_I.  
\end{equation}

%\begin{remark}
%	in applications to pAQFT, $V$ will be the space of local observables $V=\mathcal{F}_{\text{loc}}[[\hbar]]$.
%\end{remark}

%%%%%%%%%%%%%%%%%%%%%%%%%%%%%%%%%%%%%%%%%%%%%%%%%%%%%%%%%%%%%%%%%%%%%%%%
\subsection{The Steinmann Algebra}\label{sec:Steain}
%%%%%%%%%%%%%%%%%%%%%%%%%%%%%%%%%%%%%%%%%%%%%%%%%%%%%%%%%%%%%%%%%%%%%%%% 
The Hopf algebra $\Sig$ is connected and cocommutative, and so the CMM Theorem applies, see \cite[Section 1.4]{norledge2020species}. We now describe the positive\footnote{\ a species $\textbf{p}$ is \emph{positive} if $\textbf{p}[\emptyset]=0$} Lie algebra of primitive elements 
\[
\mathcal{P}(\Sig)\subset \Sig
.\] 
For $I\in \sfS$ a finite set, let a \emph{tree} $\mathcal{T}$ over $I$ be a planar\footnote{\ i.e. a choice of left and right child is made at every node} full binary tree whose leaves are labeled bijectively with the blocks of a partition of $I$ (a \emph{partition} $P$ of $I$ is a set of disjoint nonempty subsets of $I$, called \emph{blocks}, whose union is $I$). The blocks of this partition, called the \emph{lumps} of $\mathcal{T}$, form a composition called the \emph{debracketing} $F_\mathcal{T}$ of $\mathcal{T}$, by listing them in order of appearance from left to right. We denote trees by nested products $[\, \cdot\, ,\, \cdot\, ]$ of subsets or trees, see \autoref{fig:tree}. We make the convention that no trees exist over the empty set $\emptyset$. 
\begin{figure}[H]
\centering
\includegraphics[scale=0.6]{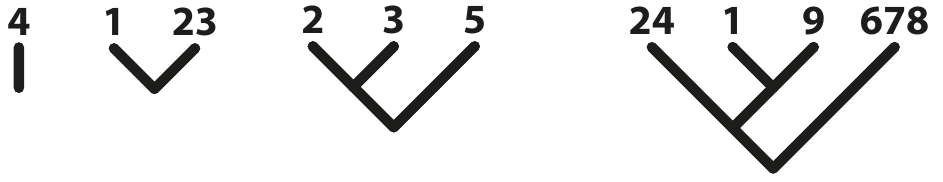}	\caption{Let $I$ be various subsets of $\{1,2,3,4,5,6,7,8,9\}$. The trees $[4]$, $[1,23]$ ($\neq[23,1]$), $[[2,3],5]$, $[[24,[1,9]],678]$ are shown. The debracketing of $[[24,[1,9]],678]$ is the composition $(24,1,9,678)$. If we put $\mathcal{T}_1=[24,[1,9]]$ and $\mathcal{T}_2=[678]$, then $[\mathcal{T}_1, \mathcal{T}_2]$ would also denote this tree.}
\label{fig:tree}
\end{figure} 
\noindent We define the positive species $\textbf{Zie}$ by letting $\textbf{Zie}[I]$ denote the vector space of formal $\bC$-linear combinations of trees over $I$, modulo the relations of antisymmetry and the Jacobi identity as interpreted on trees in the usual way. Explicitly, 
\begin{enumerate}
\item(antisymmetry) for all trees of the form $[\dots [\mathcal{T}_1, \mathcal{T}_2 ]\dots]$ (writing a tree in this form is equivalent to picking a node) we have
\[
[\dots [\mathcal{T}_1, \mathcal{T}_2 ]\dots]
+ [\dots [\mathcal{T}_2, \mathcal{T}_1 ]\dots]=
0
.\]
\item(Jacobi Identity) for all trees of the form $[\dots [[\mathcal{T}_1,\mathcal{T}_2],\mathcal{T}_3] \dots ]$ we have
\[
[\dots  [[\mathcal{T}_1,\mathcal{T}_2],\mathcal{T}_3]\dots ]+
[\dots [[\mathcal{T}_3,\mathcal{T}_1],\mathcal{T}_2]\dots ]+
[\dots [[\mathcal{T}_2,\mathcal{T}_3],\mathcal{T}_1]\dots ]=
0
.\]
\end{enumerate}
Then $\Zie$ is a positive Lie algebra in species, with Lie bracket $\partial^\ast$ given by
\[     
\partial_{S,T}^\ast(\mathcal{T}_1\otimes \mathcal{T}_2):=[\mathcal{T}_1,\mathcal{T}_2] 
.\]

\begin{remark}
We have that $\Zie$ is the free Lie algebra on the positive exponential species $\textbf{E}^\ast_+$, and so the species $\Zie$ is also given by
\[   
\Zie[I]
=   
\Lie \boldsymbol{\circ} \bE^\ast_+[I]= \bigoplus_{P}   \Lie[P]
\]
where $\textbf{Lie}$ is the species of the Lie operad, and the direct sum is over all partitions $P$ of $I$. 
\end{remark}

The Lie algebra in species $\Zie$ is closely related to the Steinmann algebra from the physics literature \cite[Section III.1]{bros}, \cite[Section 6]{Ruelle}. Precisely, the Steinmann algebra is an ordinary graded Lie algebra based on the structure map for the adjoint braid arrangement realization of $\Zie$. The adjoint braid arrangement realization of $\Zie$ is the topic of \cite{lno2019}, and the fact that the Lie algebra there is indeed $\Zie$ was shown in \cite{norledge2019hopf}. 

Via the commutator bracket, $\Sig$ is a Lie algebra in species, given by
\[
[\tH_F,\tH_G] =\tH_F \tH_G -\tH_G \tH_F
.\]
Let 
\[
[I; \text{2}]:= \big\{    \text{surjective functions}\ I\to \{1,2 \}   \big\} 
\]
denote the set of compositions of $I$ with two lumps. Since $\Sig$ is connected, its positive Lie subalgebra of primitive elements $\mathcal{P}(\Sig)\subset \Sig$ is given on nonempty $I$ by
\[
\mathcal{P}(\Sig)[I]
=\bigcap_{(S,T)\in [I; \text{2}]}  \text{ker} 
\big(    
\Delta_{S,T} :\Sig[I]\to \Sig[S]\otimes \Sig[T]    
\big)
.\]
In particular, $\tQ_{(I)}\in \mathcal{P}(\Sig)[I]$ for $I$ nonempty. Since $\Zie$ is freely generated by stick trees $[I]$, we can define a homomorphism of Lie algebras by 
\[\Zie\to \mathcal{P}(\Sig), \qquad [I]\mapsto \tQ_{(I)}.\] 
To describe this explicitly, given a tree $\mathcal{T}$, let $\text{antisym}(\mathcal{T})$ denote the set of $2^{l(F_{\mathcal{T}})-1}$ many trees which are obtained by switching left and right branches at nodes of $\mathcal{T}$. For $\mathcal{T}' \in \text{antisym}(\mathcal{T})$, let $(\mathcal{T}, \mathcal{T}')\in \bZ/2\bZ$ denote the parity of the number of node switches required to bring $\mathcal{T}$ to $\mathcal{T}'$. Then the homomorphism is given in full by
\[      
\textbf{Zie} \to \mathcal{P}(\Sig), \qquad \mathcal{T} \mapsto   \tQ_\mathcal{T}  := \sum_{\mathcal{T}' \in \text{antisym}(\mathcal{T})}  (-1)^{ (\mathcal{T},\mathcal{T}') }  \tQ_{F_{\mathcal{T}'}}
.\]
By \cite[Corollary 11.46]{aguiar2010monoidal}, this is an isomorphism. From now on, we make the identification
\[
\Zie= \mathcal{P}(\Sig)
\]
and retire the notation $\mathcal{P}(\Sig)$.

%\begin{ex}
%If $\mathcal{T}=[[1,23],4]$, then the trees $\mathcal{T}'\in \text{antisym}(\mathcal{T})$, written with prefactor $(-1)^{(\mathcal{T}, \mathcal{T}')}$, are 
%\[[[1,23],4],\qquad  -[[23,1],4],\qquad  -[4,[1,23]],\qquad [4,[23,1]].\]
%\end{ex}

%Recall that $\tE_I=\tQ_{(I)}$ for $I$ nonempty, and that the multiplication of the $\tQ$-basis of $\Sig$ is given by the concatenation of compositions. 

%%%%%%%%%%%%%%%%%%%%%%%%%%%%%%%%%%%%%%%%%%%%%%%%%%%%%%%%%%%%%%%%%%%%%%%%
\subsection{Type $A$ Dynkin Elements}\label{adjoint}
%%%%%%%%%%%%%%%%%%%%%%%%%%%%%%%%%%%%%%%%%%%%%%%%%%%%%%%%%%%%%%%%%%%%%%%% 

Recall that the set of minuscule weights of (the root datum of) $\text{SL}_I(\bC)$ is in natural bijection with $[I; \text{2}]$. We denote the minuscule weight corresponding to $(S,T)$ by $\lambda_{ST}$. See \cite[Section 3.1]{norledge2019hopf} for more details.

A \emph{cell}\footnote{\ also known as maximal unbalanced families \cite{billera2012maximal} and positive sum systems \cite{MR3467341}} \cite[Definition 6]{epstein1976general} over $I$ is (equivalent to) a subset $\cS\subseteq [I; \text{2}]$ such that for all $(S,T)\in [I; \text{2}]$, exactly one of 
\[
(S,T)\in \cS \qquad \text{and} \qquad (T,S)\in \cS
\] 
is true, and whose corresponding set of minuscule weights is closed under conical combinations, that is
\[
\lambda_{UV}\in \text{coni}\big \la  \lambda_{ST} : (S,T)\in \cS \big \ra 
\quad \implies \quad 
(U,V)\in \cS
.\]
By dualizing conical spaces generated by minuscule weights, cells are in natural bijection with chambers of the adjoint of the braid arrangement, see \cite[Section 3.3]{norledge2019hopf}, \cite[Definition 2.5]{epstein2016}. Their number is sequence \href{https://oeis.org/A034997}{A034997} in the OEIS. We denote the species of formal $\bC$-linear combinations of cells by $\textbf{L}^\vee$.

Associated to each composition $F$ of $I$ is the subset $\cF_F\subseteq [I; \text{2}]$ consisting of those compositions $(S,T)$ which are obtained by merging contiguous lumps of $F$, 
\[
\cF_F:=\big \{  (S,T)\in [I; \text{2}] : (S,T) \leq F\big \} 
.\]
More geometrically, $\cF_F$ is the subset corresponding to the set of minuscule weights which are contained in the closed braid arrangement face of $F$. Let us write $F\subseteq \cS$ as abbreviation for $\cF_F \subseteq \cS$.  

Consider the morphism of species given by
\begin{equation}     \label{eq:hbasisexp}
\textbf{L}^\vee\to \Sig
,\qquad 
\cS\mapsto  \mathtt{D}_\cS
:= 
-\sum_{\bar{F}\subseteq \cS} (-1)^{l(F)} \tH_{F}
.
\end{equation} 
The element $\mathtt{D}_\cS$ is called the \emph{Dynkin element} associated to the cell $\cS$. These special elements were defined by Epstein-Glaser-Stora in \cite[Equation 1, p.26]{epstein1976general}, and the name is due to \hbox{Aguiar-Mahajan} \cite[Equation 14.1]{aguiar2017topics} (see \autoref{Rem:dny}). In fact, $\mathtt{D}_\cS$ is a primitive element \cite[Proposition 14.1]{aguiar2017topics}, and so we actually have a morphism $\textbf{L}^\vee\to \Zie$.

For $i\in I$, let $\cS_i$ denote the cell given by
\[     \cS_i:=\big \{ (S,T)\in [I, \text{2}]: i\in S     \big \}  .    \]
This is the cell corresponding to the adjoint braid arrangement chamber which contains the projection of the basis element $e_i\in \bR I$ onto the sum-zero hyperplane. Let the \emph{total retarded} Dynkin element $\mathtt{D}_i$ associated to $i$ be given by
\[   
\mathtt{D}_i:= \mathtt{D}_{\cS_i}  =-\sum_{\substack{F\in \Sigma[I]\\ i\in S_k}} (-1)^{l(F)} \tH_F   
.\]
These Dynkin elements are considered in \cite[Section 14.5]{aguiar2013hopf}. For $i\in I$, let
\[     
\bar{\cS}_i
:=
\big \{ 
(S,T)\in [I, \text{2}]: i\in T     
\big \}  
.\]
This is the cell corresponding to the adjoint braid arrangement chamber which is opposite to the chamber of $\cS_i$. Let the \emph{total advanced} Dynkin element $\mathtt{D}_{\bar{i}}$ associated to $i$ be given by
\[   
\mathtt{D}_{\bar{i}}:= \mathtt{D}_{\bar{\cS}_i}=-\sum_{\substack{F\in \Sigma[I]\\ i\in S_1}} (-1)^{l(F)} \tH_F   
.\]

\begin{figure}[t]
	\centering
	\includegraphics[scale=0.7]{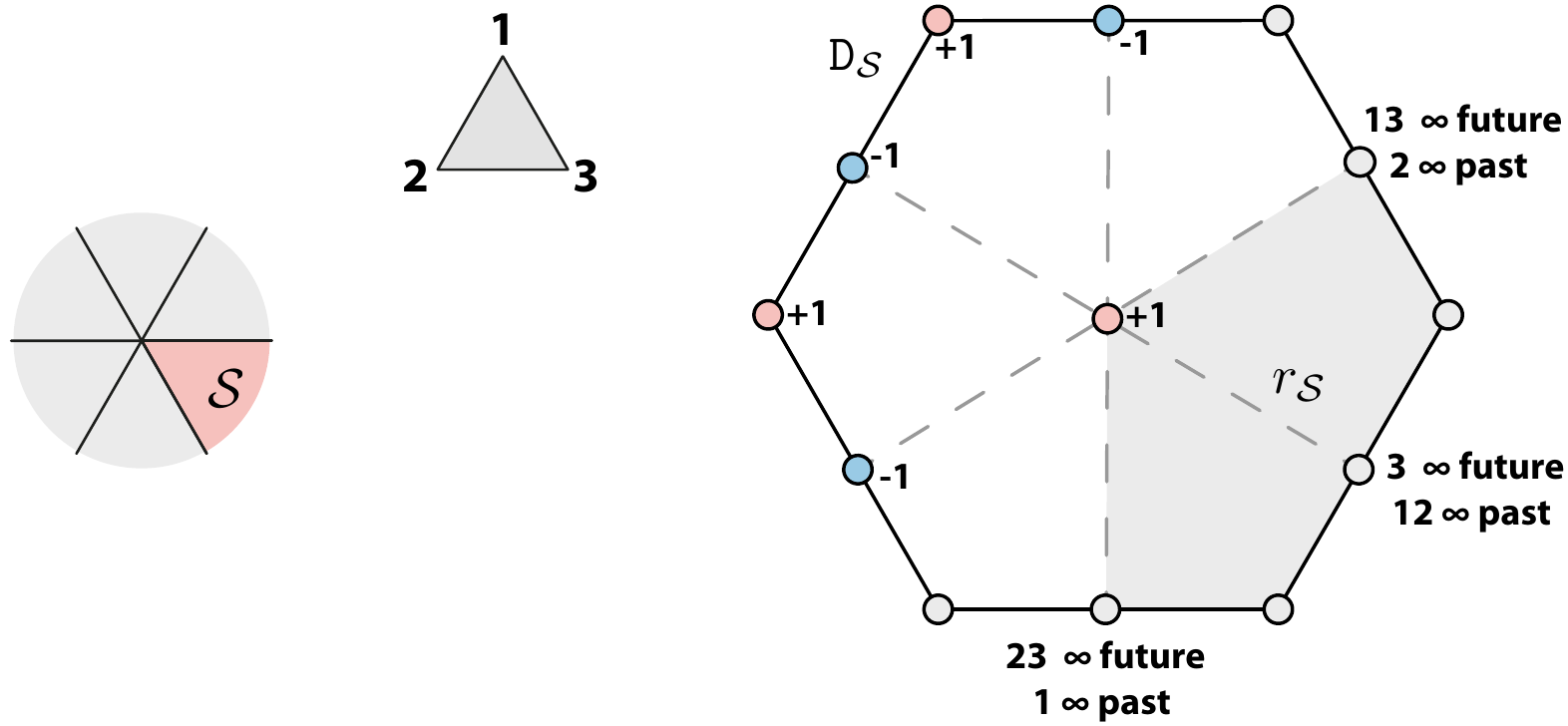}
	\caption{A cell $\cS$ over $\{1,2,3\}$ (on the adjoint braid arrangement) and its Dynkin element $\mathtt{D}_\cS$ (on the tropical geometric realization of $\boldsymbol{\Sigma}$, where the multiplication embeds facets and the comultiplication projects onto facets, see \cite[Introduction]{norledge2019hopf})). In the presence of causal factorization, the time component of the corresponding generalized retarded function $r_\cS$ is a \hbox{$\bC[[\hbar, \formg]]$-valued} generalized function on the braid arrangement with support the gray cone. The Dynkin element shown is $\mathtt{D}_\cS=\mathtt{D}_3=\mathtt{R}_{(12;3)}$. Its support consists of those configurations such that the event labeled by $3$ can be causally influenced by the events labeled by $1$ and $2$.}
	\label{fig:supp}
\end{figure}

\begin{remark}\label{Rem:dny}
More generally, Dynkin elements are certain Zie elements of generic real hyperplane arrangements, which are indexed by chambers of the corresponding adjoint arrangement. They were introduced by Aguiar-Mahajan in \cite[Equation 14.1]{aguiar2017topics}. Specializing to the braid arrangement, one recovers the type $A$ Dynkin elements $\mathtt{D}_\cS$. 
\end{remark}

In \cite{norledge2019hopf},  the following perspective on the Dynkin elements is given. The Hopf algebra $\Sig^\ast$ which is dual to $\Sig$ is realized as an algebra $\hat{\Sig}^\ast$ of piecewise-constant functions on the braid arrangement. Then its dual, in the sense of polyhedral algebras \cite[Theorem 2.7]{MR1731815}, is an algebra $\check{\Sig}^\ast$ of certain functionals of piecewise-constant functions on the adjoint braid arrangement, i.e. those coming from evaluating on permutohedral cones. We have the morphism of species
\[   
\check{\Sig}^\ast\to  (\textbf{L}^\vee)^\ast    
\]
defined by sending functionals to their restrictions to piecewise-constant functions on the complement of the hyperplanes. Since the multiplication of $\check{\Sig}^\ast$ corresponds to embedding hyperplanes, this morphism is the indecomposable quotient of $\check{\Sig}^\ast$ \cite[Theorem 4.5]{norledge2019hopf}. Then, in \cite[Proposition 5.1]{norledge2019hopf}, we see that taking the linear dual of this morphism recovers the Dynkin elements map,
\[        
 \textbf{L}^\vee\to \Sig, \qquad \cS\mapsto  \mathtt{D}_\cS
.\] 
(Here we have identified $\Sig^\ast=  \check{\Sig}^\ast$.) Therefore we obtain the following.

\begin{thm}[$\! \! ${\cite{norledge2019hopf}}]
The morphism of species $\textbf{L}^\vee\to \Zie$ is surjective. Therefore the Dynkin elements $\{\mathtt{D}_\cS: \cS\ \text{is a cell over $I$} \}$ span $\Zie$.
\end{thm} 
%%%%%%%%%%%%%%%%%%%%%%%%%%%%%%%%%%%%%%%%%%%%%%%%%%%%%%%%%%%%%%%%%%%%%%%%
\subsection{The Steinmann Relations}\label{stein}
%%%%%%%%%%%%%%%%%%%%%%%%%%%%%%%%%%%%%%%%%%%%%%%%%%%%%%%%%%%%%%%%%%%%%%%% 
The Dynkin elements span $\Zie$, but they are not linearly independent. The relations which are satisfied by the Dynkin elements are generated by relations known in physics as the Steinmann relations, introduced in \cite{steinmann1960zusammenhang}, \cite{steinmann1960}. 
%We now abstract out the essential feature of the above example to describe the Steinmann relations in general. 

Let a pair of \emph{overlapping channels} over $I$ be a pair $(S,T),(U,V)\in [I; \text{2}]$ of two-lump compositions of $I$ such that 
\[     
S\cap U\neq \emptyset \qquad \text{and}  \qquad      T\cap U \neq \emptyset   
.\]
Let $\cS_1$, $\cS_2$, $\cS_3$, $\cS_4$ be four cells over $I$ with $(S,T),(U,V)\in \cS_1$, and such that $\cS_2$, $\cS_3$, $\cS_4$ are obtained from $\cS_1$ by replacing, respectively, 
\[        (S,T), (U,V)   \mapsto  (T,S), (U,V)         \]
\[        (S,T), (U,V)   \mapsto  (T,S), (V,U)         \]
\[        (S,T), (U,V)   \mapsto  (S,T), (V,U).        \]
Then, by inspecting the definition of the Dynkin elements \textcolor{blue}{(\refeqq{eq:hbasisexp})}, we see that\footnote{\ we go through the argument for the basic $4$-point case in \autoref{ex:stein}, which is sufficient to exhibit the general phenomenon}
\[    
 \mathtt{D}_{\cS_1} -\mathtt{D}_{\cS_2}    +\mathtt{D}_{\cS_3}    -\mathtt{D}_{\cS_4}    =0.        
\]
In general, a \emph{Steinmann relation} is any relation between Dynkin elements obtained in this way, i.e. an alternating sum of four Dynkin elements which are obtained from each other by switching overlapping channels only. This definition of the Steinmann relations can be found in \cite[Seciton 4.3]{epstein1976general} (it is given slightly more generally there for paracells). 

An alternative characterization of the Steinmann relations in terms of the Lie cobracket of the dual Lie coalgebra $\Zie^\ast$ is \cite[Definition 4.2]{lno2019}. Here, the Steinmann relations appear in the same way one can arrive at generalized permutohedra, i.e. by insisting on type $A$ `factorization' in the sense of species-theoretic coalgebra structure. See \cite[Theorem 4.2 and Remark 4.2]{norledge2019hopf}. 

%The Steinmann algebra $\Zie$ appears most explicitly in the work of Ruelle \cite[Section 6]{Ruelle}. However the approach there uses certain `cycles', and the connection is explained in \cite[Section 4.4]{epstein1976general}.

Thus, Dynkin elements satisfy the Steinmann relations. Moreover, they are sufficient.

\begin{thm}
The relations which are satisfied by the Dynkin elements are generated by the Steinmann relations. That is, if 
\[
\textbf{Stein}[I]
:=
\big \la
\mathtt{D}_{\cS_1}-\mathtt{D}_{\cS_2}+\mathtt{D}_{\cS_3}-\mathtt{D}_{\cS_4}
:
\mathtt{D}_{\cS_1}-\mathtt{D}_{\cS_2}+\mathtt{D}_{\cS_3}-\mathtt{D}_{\cS_4}=0 \text{ is a Steinmann relation}
\big \ra\footnote{\ angled brackets denote $\bC$-linear span} 
\] 
then
\[
\Zie\cong \bigslant{\textbf{L}^\vee}{\textbf{Stein}} 
.\] 
\end{thm}
\begin{proof}
This follows by combining \cite[Theorem 4.3]{lno2019} with \cite[Theorems 4.2 and 4.5]{norledge2019hopf}. 
\end{proof}

\begin{ex}\label{ex:stein}
Let us give the basic $4$-point example $I=\{1,2,3,4\}$, which takes place on a square facet of the type $A$ coroot solid \cite[Figure 1]{lno2019}. Consider the following four cells over $I$ (we have marked where they differ, the names `$s$-channel' and `$u$-channel' are from physics and refer to Mandelstam variables),
\[   
\cS_1=\big\{\underbrace{(23,14)}_{u\text{-channel}}, (12,34), (1,234), (13,24), (13,24), (134,2), (3,124)  \}    
\]
\[ 
\cS_2=\big\{(23,14), \underbrace{(34,12)}_{s\text{-channel}}, (1,234), (13,24), (13,24), (134,2), (3,124)  \big \} 
\] 
\[ 
\cS_3=\big\{\underbrace{(14,23)}_{u\text{-channel}}, (34,12), (1,234), (13,24), (13,24), (134,2), (3,124)  \big \} 
\]
\[ 
\cS_4=\big\{(14,23), \underbrace{(12,34)}_{s\text{-channel}}, (1,234), (13,24), (13,24), (134,2), (3,124)  \big \}
.\] 
The $s$-channel and the $u$-channel overlap, and so we should now have
\[    
\mathtt{D}_{\cS_1}-\mathtt{D}_{\cS_2}+\mathtt{D}_{\cS_3}-\mathtt{D}_{\cS_4}=0 
.\]    
To see this, let us assume throughout that $\tH_{F}$ appears in the $\tH$-basis expansion \textcolor{blue}{(\refeqq{eq:hbasisexp})} of $\mathtt{D}_{\cS_1}$, i.e. $\bar{F} \subseteq \cS_1$. Then we have
\begin{equation*}
\bar{F} \subseteq \cS_1 \setminus \{  (12,34), (23,14)  \}   \quad \implies \quad      \bar{F}\subseteq\cS_1, \   \cS_2,\ \cS_3,\ \cS_4.  \tag{$\spadesuit$} \label{1}
\end{equation*}
If $\bar{F} \nsubseteq \cS_1 \setminus \{  (12,34), (23,14)  \}$, then either $(12,34)\in \bar{F}$ or $(23,14)\in \bar{F}$ but not both, since the channels overlap. We then have
\begin{equation*}
(12,34)\in \bar{F}   \implies    \bar{F}\subseteq\cS_1, \ \bar{F}\nsubseteq\cS_2, \ \bar{F}\nsubseteq\cS_3, \   \bar{F}\subseteq\cS_4.    \tag{$\varheart$}  \label{2}
\end{equation*}
We also have
\begin{equation*}
(23,14)\in \bar{F} \implies      \bar{F}\subseteq\cS_1, \ \bar{F}\subseteq\cS_2, \ \bar{F}\nsubseteq\cS_3, \   \bar{F}\nsubseteq\cS_4   .              \tag{$\vardiamond$}  \label{3}
\end{equation*}
Notice that in all three cases  \textcolor{blue}{(\refeqq{1})}, \textcolor{blue}{(\refeqq{2})}, \textcolor{blue}{(\refeqq{3})}, the prefactors of $\tH_{F}$ sum to zero in the four term alternating sum of the Steinmann relation. 
\end{ex}

\begin{remark}
In \cite{norledge2019hopf}, the Steinmann condition is seen to be equivalent to the restriction to generalized permutohedra in a certain local (or spherical) sense. Ocneanu \cite{oc17} and Early \cite{early2019planar} have studied an affine version of the Steinmann condition, in the context of higher structures and matroid subdivisions. Here, one observes that the (translated) hyperplanes of the adjoint braid arrangement for the Mandelstam variables give three subdivisions of the hypersimplex $\Delta(2,4)$ (octahedron). 
\begin{figure}[H]
\centering
\includegraphics[scale=0.5]{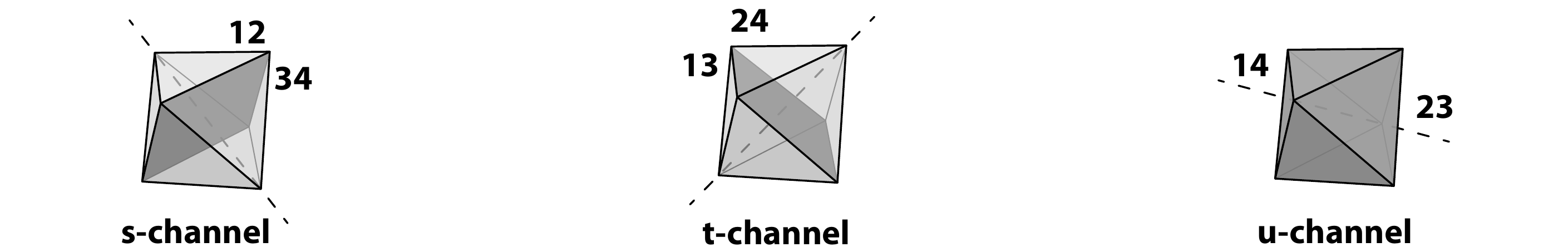}	
\label{fig:affine}
\end{figure}
\noindent See \cite{borges2019generalized}, \cite{cachazo2019planar} for the closely related study of generalized Feynman diagrams in generalized biadjoint $\Phi^3$-theory.
\end{remark}

%%%%%%%%%%%%%%%%%%%%%%%%%%%%%%%%%%%%%%%%%%%%%%%%%%%%%%%%%%%%%%%%%%%%%%%%
\subsection{Ruelle's Identity} \label{sec:Ruelle's Identity and the GLZ Relation}
%%%%%%%%%%%%%%%%%%%%%%%%%%%%%%%%%%%%%%%%%%%%%%%%%%%%%%%%%%%%%%%%%%%%%%%% 
Since the Dynkin elements span $\Zie$, we can ask what is the description of the Lie bracket of $\Zie$ in terms of the Dynkin elements. The answer is known in the physics literature as Ruelle's identity.

In order to state Ruelle's identity, we need to notice the following. For $S\sqcup T=I$, if $\cS_1$ is a cell over $S$ and $\cS_2$ is a cell over $T$, then $\cS_1 \sqcup \cS_2$ describes a collection of codimension one faces of the adjoint braid arrangement which are supported by the hyperplane orthogonal to $\lambda_{ST}$ (in \cite{lno2019}, such faces were called \emph{Steinmann equivalent}). A cell $\cS^{[S,T]}$ over $I$ which satisfies
\[
\cS^{[S,T]} \supseteq \cS_1\sqcup\cS_2
\qquad \text{and} \qquad
(S,T)\in \cS^{[S,T]}
\]
corresponds to a chamber arrived at by moving (by an arbitrarily small amount) from an interior point of a face of $\cS_1 \sqcup \cS_2$ in the $\lambda_{ST}$ direction. In particular, such cells always exist, but they are not unique (the Steinmann relations exactly quotient out this ambiguity). The chamber obtained by moving in the opposite direction corresponds to the cell obtained by replacing $(S,T)$ with $(T,S)$ in $\cS^{[S,T]}$.

\begin{prop}[Ruelle's Identity {\cite[Equation 6.6]{Ruelle}}] \label{prop:ruelle}
For $S\sqcup T=I$, let $\cS_1$ be a cell over $S$ and let $\cS_2$ be a cell over $T$. Let $\cS^{[S,T]}$ be a cell over $I$ which satisfies 
\[
\cS^{[S,T]} \supseteq \cS_1\sqcup\cS_2
\qquad \text{and} \qquad
(S,T)\in \cS^{[S,T]}
.\] 
Let $\cS^{[T,S]}$ denote the cell obtained by replacing $(S,T)$ with $(T,S)$ in $\cS^{[S,T]}$. Then the Lie bracket of $\Zie$ is given by
\begin{equation}\label{eq:ruelleiden}
[\mathtt{D}_{\cS_1},\mathtt{D}_{\cS_2}]   
=   
\mathtt{D}_{\cS^{[S,T]}}-\mathtt{D}_{\cS^{[T,S]}}
.
\end{equation} 
\end{prop}
\begin{proof}
%Notice that $\cS$ is the adjoint family for a codimension one shard, and $\cS^{[S,T]}$ and $\cS^{[T,S]}$ are its two adjacent maximal shards.
This result is clear from \cite[Section 5.2]{lno2019}; the Lie bracket which was given to the adjoint braid arrangement realization of $\textbf{Z}\textbf{ie}$ (denoted there by $\Gam$) coincides with \textcolor{blue}{(\refeqq{eq:ruelleiden})}. Alternatively, we can just explicitly check, as in \cite[Section 4.3]{epstein1976general}. 
\end{proof}

%If a system of interacting generalized time-ordered products is constructed as in \cite{steinbook71}, \cite{dutfredretard04}, \cite{dutsch2019perturbative} via total retarded products, then one includes the GLZ relation \cite[Equation 11]{GLZ1957}, \cite[Proposition 1.10.1]{dutsch2019perturbative}, which is a consequence of the following relation satisfied Dynkin elements $\mathtt{D}_i$,
%\[     
%\mathtt{D}_{i_1}- \mathtt{D}_{i_2}=   \sum_{\substack{(S,T)\in [\text{2};I]\\[1pt] i_1\in S,\ i_2\in T }}[\mathtt{D}_{i_1}  , \mathtt{D}_{i_2}  ]
%.\]

%%%%%%%%%%%%%%%%%%%%%%%%%%%%%%%%%%%%%%%%%%%%%%%%%%%%%%%%%%%%%%%%%%%%%%%%
\section{$\Sig$ as a Hopf $\textbf{E}$-Algebra} \label{Sig as a Hopf E-Algebra} 
%%%%%%%%%%%%%%%%%%%%%%%%%%%%%%%%%%%%%%%%%%%%%%%%%%%%%%%%%%%%%%%%%%%%%%%%
%Bogoliubov-Shirkov approach to perturbative QFT

We now recall the Steinmann arrows, which are (or we interpret as) actions of the exponential species $\textbf{E}$ on $\Sig$. We show that they give $\Sig$ the structure of a Hopf $\textbf{E}$-algebra (=Hopf monoid internal to $\textbf{E}$-modules) in two ways, and thus the primitive part $\Zie=\mathcal{P}(\Sig)$ the structure of a Lie $\textbf{E}$-algebra in two ways. %As we shall see, the resulting perturbation of products on $\Sig$ recovers Bogoliubov's formula \cite[Chapter 4]{Bogoliubov59}. 
%%%%%%%%%%%%%%%%%%%%%%%%%%%%%%%%%%%%%%%%%%%%%%%%%%%%%%%%%%%%%%%%%%%%%%%%
\subsection{Derivations and Coderivations of $\Sig$}  \label{Derivations and Coderivations of Sig} 
%%%%%%%%%%%%%%%%%%%%%%%%%%%%%%%%%%%%%%%%%%%%%%%%%%%%%%%%%%%%%%%%%%%%%%%%

Let $Y=\{  y_1,\dots, y_r \}$ be a finite set with cardinality $r\in \bN$. We think of $Y$ as having `color' $\formg$ (physically, the coupling constant). Given a species $\textbf{p}$, we have the \hbox{$Y$-\emph{derivative}} $\textbf{p}^{[Y]}$ of $\textbf{p}$, which is the species given by
\[
\textbf{p}^{[Y]}[I] :=   \textbf{p}[Y \sqcup I] \qquad \text{and} \qquad  \textbf{p}^{[Y]}[\sigma] :=   \textbf{p}[\text{id}_Y \sqcup \sigma] 
.\]
A \emph{raising operator} $u$ on $\textbf{p}$ is a morphism of species of the form\footnote{\ for raising operators, we often abbreviate $u(\mathtt{a}):=u_I(\mathtt{a})$}
\[
u: \textbf{p} \to \textbf{p}^{[Y]}, \qquad \ta \mapsto u(\ta)
.\]

\begin{remark}
Moreover, there is an endomorphism algebra of raising operators \cite[Section 2.4]{norledge2020species}, which features when considering modules internal to species, see \cite[Section 5.1]{norledge2020species}.
\end{remark}

As a particular example of the set $Y$, we have the set of formal symbols $[r]:=\{ \ast_1, \dots, \ast_r \}$ (formally, we have picked a section of the decategorification functor $Y\mapsto r$). We often abbreviate $\ast=\ast_1$, also $\ast=\{\ast\}$ and $\ast I = \{ \ast \}\sqcup I$. The \emph{derivative} $\textbf{p}'$ of $\textbf{p}$ is the $Y$-derivative in the singleton case $Y=\{\ast\}$, thus
\[
\textbf{p}'[I]:= \textbf{p}^{[\ast]}[I]= \textbf{p}[\ast I]
.\]
Following \cite[Section 8.12.1]{aguiar2010monoidal}, an \emph{up operator} $u$ on $\textbf{p}$ is a raising operator of the form $u:\textbf{p} \to \textbf{p}'$. Writing $u_\ast(\ta)=u(\ta)$ in order to specify the name of the adjoined singleton, we call an up operator \emph{commutative} if
\[
	u_{\ast_2}(u_{\ast_1}(\mathtt{a})) = u_{\ast_1} (u_{\ast_2}(\mathtt{a}))
	.
\]
Raising operators can be obtained by iteratively applying commutative up operators, see \cite[Section 5.4]{norledge2020species}. Following \cite[Section 8.12.4]{aguiar2010monoidal}, an up operator on an algebra $\textbf{a}$ is called an \emph{up derivation} if
\begin{equation}\label{eq:upder}
	u\big(\mu_{S,T}(\mathtt{a} \otimes \mathtt{b})\big) 
	=\mu_{\ast S,T}\big (u(\mathtt{a}) \otimes    \mathtt{b}\big)+\mu_{S,\ast T}\big(\mathtt{a} \otimes u(\mathtt{b})\big)
\end{equation}
(it follows that $u(\mathtt{1}_{\ta})=0$ if $\textbf{a}$ is unital) and an up operator on a coalgebra $\textbf{c}$ is called an \emph{up coderivation} if
\begin{equation}\label{eq:upcoder}
	\big(  
	u\otimes \text{id} + \text{id} \otimes u
	\big)
	\circ  
	\Delta_{S,T}(\mathtt{a})
	=
	\Delta_{\ast S,T} \big(u(\mathtt{a})\big)
	+
	\Delta_{S,\ast T} \big(u(\mathtt{a})\big)
	.
\end{equation}
An \emph{up biderivation} on a bialgebra $\textbf{h}$ is an up operator which is both an up derivation and an up coderivation. The data of an up (co/bi)derivation on a connected species $\textbf{h}$ is equivalent to giving $\textbf{h}$ the structure of an \hbox{$\textbf{L}$-(co/Hopf)algebra} (= an (co/Hopf)monoid internal to $\textbf{L}$-modules). The data of a commutative up (co/bi)derivation on $\textbf{h}$ is equivalent to giving $\textbf{h}$ the structure of an \hbox{$\textbf{E}$-(co/Hopf)algebra}. See \cite[Section 5]{norledge2020species} for more details and proofs.

Thus, an up derivation $u$ of $\Sig$ is a morphism of species 
\[
u:\Sig\to \Sig',
\qquad 
\tH_F\mapsto u(\tH_F) 
\qquad \quad \text{such that} \qquad 
u(\tH_F \tH_G )=
u(\tH_F)\tH_G+\tH_F u(\tH_G)
.\] 
An up derivation of $\Sig$ is determined by its values on the elements $\tH_{(I)}$, $I\in \sfS$, since then
\[
u(\tH_F)= u(\tH_{(S_1)})\tH_{(S_2)}\dots \tH_{(S_k)} 
+\ \,  \cdots\ \,  +
\tH_{(S_1)} \dots \tH_{(S_{k-1})}   u(\tH_{(S_k)}) 
.\] 
An up derivation must have $u(\tH_{(\, )})=0$, since $\mathtt{1}_{\Sig}= \tH_{(\, )}$. An up coderivation $u$ of $\Sig$ is a morphism of species
\[
u:\Sig\to \Sig', \qquad \tH_F\mapsto u(\tH_F) 
\qquad \quad \text{such that} \qquad
\Delta_{\ast S,T} \big(  u( \tH_F )\big)=
u(\tH_{F|_S}) \otimes \tH_{F|_T}
.\]
In particular, an up coderivation must have
\[
\Delta_{\ast S,T}\big(u( \tH_{(I)} )\big)=
u(\tH_{(S)}) \otimes \tH_{(T)}  
.\]
Therefore, an up biderivation $u$ of $\Sig$ must have
\[
u(\tH_{(i)})=  a_1 \tH_{(\ast,i)}+a_2 \tH_{(\ast i)}+a_3 \tH_{(i, \ast)} 
\qquad \quad \text{where} \qquad
a_1+a_2+a_3=0\in \bC
.\]
Motivated by this, given $a,b\in \bC$, we define an up derivation $u_{a,b}$ of $\Sig$ by
\begin{equation}\label{eq:defbider}
u_{a,b}:\Sig\to \Sig'
,\qquad 
u_{a,b}(\tH_{(I)}):
= -a \tH_{(\ast,I)}+(a+b) \tH_{(\ast I)}-b \tH_{(I, \ast)}
.
\end{equation}
Towards an explicit description, consider the following example for $I=\{1,2,3\}$,
\begin{align*}
u_{a,b}(\tH_{(12,3)})
&
=\, u_{a,b}(\tH_{(12)})\tH_{(3)}+ \tH_{(12)}u_{a,b}(\tH_{(3)})\\
&
=(-a\tH_{(\ast, 12)}+(a+b)\tH_{(\ast 12)}-b\tH_{(12,\ast)})\tH_{(3)}+
\tH_{(12)}(-a\tH_{(\ast, 3)}+(a+b)\tH_{(\ast 3)}-b\tH_{(3,\ast)})\\
&
=-a\tH_{(\ast, 12,3)}+(a+b)\tH_{(\ast 12,3)}-b\tH_{(12,\ast,3)})
-a\tH_{(12,\ast, 3)}+(a+b)\tH_{(12,\ast 3)}-b\tH_{(12,3,\ast)}
.
\end{align*}
From this, we see that in general
\[
u_{a,b}(\tH_F)
= 
\sum_{1\leq m\leq k}
-a\mathtt{H}_{(S_1,\dots   ,\ast, S_m,\dots,S_k)}
+(a+b)\mathtt{H}_{(S_1,\dots, \ast S_m,\dots,S_k)}
-b\mathtt{H}_{(S_1,\dots, S_m,\ast,\dots,S_k)}
.\]

\begin{thm}\label{steinmannarrowaredercoder}
Given $a,b\in \bC$, the morphism of species 
\[
\Sig\to \Sig', 
\qquad 
\tH_F \mapsto u_{a,b}(\tH_{F})
\]
is an up biderivation of $\Sig$ (it follows this gives $\Sig$ the structure of a Hopf $\textbf{L}$-algebra).
\end{thm}
\begin{proof}
In the following, for $F=(S_1,\dots,S_k)$ a composition of $I$ and $S\subseteq I$, we write
\[ 
(U_1,\dots,U_k):=(S_1\cap S, \dots ,S_k\cap S)
.\]
In general, $(U_1,\dots,U_k)$ is a decomposition of $I$. 

First, $u_{a,b}$ defines a derivation of $\Sig$ by construction. To see that $u_{a,b}$ also defines a coderivation, we have
\begin{align*}
\Delta_{\ast S,T} \big(  u_{a,b}( \tH_F )\big)\ =\ 
&
\ \ \ \   \Delta_{\ast S,T} \Bigg( \sum_{1\leq m\leq k}
-a\mathtt{H}_{(S_1,\dots,\ast, S_m,\dots,S_k)}
+(a+b)\mathtt{H}_{(S_1,\dots, \ast S_m,\dots,S_k)}
-b\mathtt{H}_{(S_1,\dots, S_m,\ast,\dots,S_k)}\Bigg )\\[7pt]
=\
&
\ \ \ \    
\Bigg(\sum_{1\leq m\leq k}-a\mathtt{H}_{(U_1,\dots,\ast, U_m,\dots,U_k)_+}+
(a+b)\mathtt{H}_{(U_1,\dots, \ast U_m,\dots,U_k)_+}
-b\tH_{(U_1,\dots, U_m,\ast,\dots,U_k)_+}\Bigg)\otimes \tH_{F|_T}\\[7pt]
=\
&
\ \ \ \     
\Bigg(\sum_{\substack{1\leq m\leq k \\[2pt] U_m \neq \emptyset}}-a\mathtt{H}_{(U_1,\dots,\ast, U_m,\dots,U_k)_+}+
(a+b)\mathtt{H}_{(U_1,\dots, \ast U_m,\dots,U_k)_+}
-b\tH_{(U_1,\dots, U_m,\ast,\dots,U_k)_+}\Bigg)\otimes \tH_{F|_T}\\
&+\ \underbrace{ 
\Bigg(\sum_{\substack{1\leq m\leq k\\[2pt]U_m=\emptyset}}\big (-a+(a+b)-b\big)\, \mathtt{H}_{(U_1,\dots,U_{m-1},\ast,U_{m+1},\dots,U_k)_+}\Bigg)}_{=0} \otimes\, \tH_{F|_T}\\[7pt]
=\  
&
\ \ \ \ u(\tH_{F|_S})\otimes \tH_{F|_T}.
\end{align*}
Therefore $u_{a,b}$ is a biderivation of $\Sig$. %Then \autoref{prop:hopfalgebra}, which is for Hopf $\textbf{E}$-algebras, goes through in the same way for  Hopf $\textbf{L}$-algebras, and so $u_{a,b}$ equips $\Sig$ with the structure of a Hopf $\textbf{L}$-algebra.
\end{proof}

%\begin{conj}
%Whilst \textcolor{blue}{(\refeqq{eq:defbider})} may still seem ad hoc, there are indications that this definition is forced, and also $ab=0$ is forced, if one additionally requires commutativity of the biderivation. 
%\end{conj}
%%%%%%%%%%%%%%%%%%%%%%%%%%%%%%%%%%%%%%%%%%%%%%%%%%%%%%%%%%%%%%%%%%%%%%%%
\subsection{The Steinmann Arrows} \label{sec:The Steinmann Arrows}
%%%%%%%%%%%%%%%%%%%%%%%%%%%%%%%%%%%%%%%%%%%%%%%%%%%%%%%%%%%%%%%%%%%%%%%%
We now recall the Steinmann arrows for $\Sig$, whose precise definition is due to Epstein-Glaser-Stora \cite[p.82-83]{epstein1976general}. The Steinmann arrows were first considered by Steinmann in settings where $\Sig$ is represented as operator-valued distributions \cite[Section 3]{steinmann1960}.

Let the \emph{retarded Steinmann arrow} be the up biderivation of $\Sig$ given by
\begin{equation}\label{steindown}
\ast\downarrow(-) :\Sig\to \Sig', 
\qquad    
\ast\downarrow \tH_F:=  u_{1,0}(\tH_F)=
\sum_{1\leq m\leq k}
-\mathtt{H}_{(S_1,\dots,\ast, S_m,\dots,S_k)}
+\mathtt{H}_{(S_1,\dots, \ast S_m,\dots, S_k)} 
.
\end{equation} 
Let the \emph{advanced Steinmann arrow} be the up biderivation of $\Sig$ given by
\begin{equation}\label{steinup}
\ast\uparrow(-) : \Sig\to \Sig' , 
\qquad    
\ast\uparrow \tH_F:=  u_{0,1}(\tH_F)=
\sum_{1\leq m\leq k}
\mathtt{H}_{(S_1,\dots, \ast S_m,\dots, S_k)}
-\mathtt{H}_{(S_1, \dots, S_m,\ast, \dots,S_k)}
.
\end{equation} 
We use this arrow notation from now on instead of `$u$' in order to match the physics literature. In particular
\[
\ast \downarrow \tH_{(I)}=-\tH_{(\ast,I)} + \tH_{(\ast I)}
\qquad \text{and} \qquad
\ast \uparrow \tH_{(I)}=\tH_{(\ast I)} -\tH_{(I,\ast )}
.\]
We have
\[
\ast \uparrow \tH_{F}\, -\, \ast \downarrow \tH_{F}  =  u_{-1,1}(\tH_F) = [ \tH_{(\ast)} , \tH_{F}]
.\]
This identity appears often in the physics literature for operator-valued distributions, e.g. \cite[Equation 13]{steinmann1960}, \cite[Equation 83]{ep73roleofloc}. The biderivation $u_{-1,1}$ gives $\Sig$ the structure of a Hopf $\textbf{L}$-algebra. This $\textbf{L}$-action is the restriction of the adjoint representation of $\Sig$. Notice the Steinmann arrows are commutative up operators. By \cite[Proposition 5.4]{norledge2020species}, we can restrict them to obtain up derivations of $\Zie$,
\[
\ast\downarrow(-): \Zie\to \Zie', 
\qquad 
\mathtt{D}_\cS \mapsto  \ast  \downarrow \mathtt{D}_\cS
\qquad \text{and} \qquad
\ast\uparrow(-):\Zie\to \Zie', 
\qquad 
\mathtt{D}_\cS \mapsto  \ast  \uparrow \mathtt{D}_\cS
.\]

%Thus, we have
%\[
%\ast \downarrow [ \mathtt{D}_{\cS_1}, \mathtt{D}_{\cS_2} ]  =   [   \ast \downarrow  \mathtt{D}_{\cS_1},   \mathtt{D}_{\cS_2} ]+  [     \mathtt{D}_{\cS_1},   \ast \downarrow \mathtt{D}_{\cS_2} ]
%\qquad \text{and} \qquad 
%\ast \uparrow [ \mathtt{D}_{\cS_1}, \mathtt{D}_{\cS_2} ]  =   [   \ast \uparrow  \mathtt{D}_{\cS_1},   \mathtt{D}_{\cS_2} ]+  [     \mathtt{D}_{\cS_1},   \ast \uparrow\mathtt{D}_{\cS_2} ]
%.\]

Following \cite[Section 5]{norledge2020species}, the Steinmann arrows equip $\Sig$ with the structure of a Hopf \hbox{$\textbf{E}$-algebra} (and $\Zie$ with the structure of a Lie $\textbf{E}$-algebra) in two ways. The details are as follows. First, $\textbf{E}$ is the \emph{exponential species}, given by 
\[
\textbf{E}[I]:=\bC \qquad \text{for all} \quad I\in \textsf{S}
.\]
We denote $\tH_I:=1_\bC\in  \textbf{E}[I]$. The exponential species is an algebra in species when equipped with the trivial multiplication
\[
\mu_{S,T}: \textbf{E}[S] \otimes \textbf{E}[T] = \bC \otimes \bC \xrightarrow{\sim} \bC = \textbf{E}[I]
,\qquad 
\tH_S \otimes \tH_T \mapsto \tH_I
.\]
We have the following $\textbf{E}$-modules induced by the Steinmann arrows, as defined in \cite[Equation 23]{norledge2020species},
\[
\textbf{E}\bigcdot \Sig \to \Sig,
\qquad
\tH_Y\otimes \mathtt{a}\, \mapsto \, 
Y\! \downarrow \mathtt{a}:=
\underbrace{y_r \downarrow 
\circ\cdots \circ     
y_1 \downarrow}_{\text{invariant of the order}}(\mathtt{a}) 
\]
and 
\[
\textbf{E}\bigcdot \Sig \to \Sig,
\qquad
\tH_Y\otimes \mathtt{a}\, \mapsto \, 
Y\! \uparrow \mathtt{a}:=
\underbrace{y_r \uparrow 
\circ\cdots \circ     
y_1 \uparrow}_{\text{invariant of the order}}(\mathtt{a}) 
\]
where $Y=\{y_1,\dots,y_r\}$ as usual. In particular, $Y\downarrow(-)$ and $Y\uparrow(-)$ are the Steinmann arrow raising operators obtained from iterating the Steinmann arrow up operators $\ast \downarrow(-)$ and $\ast\downarrow(-)$, as mentioned in \autoref{Derivations and Coderivations of Sig}. For example, the retarded arrow $Y\downarrow(-)$ consists of a linear map of the form
\[
\Sig[I] \to \Sig[Y\sqcup I]
\]
for each choice of finite set $I$. For $Y=[r]:=\{\ast_1 , \dots ,\ast_r \}$, we abbreviate 
\[
\downarrow(-) :=\ast\downarrow(-), \qquad  \downarrow\downarrow(-) :=\{  \ast_1, \ast_2 \} \downarrow(-),\quad  \dots
\] 
and similarly for the advanced arrow. Since the arrows are derivations, they respect the multiplication of $\Sig$, and since the arrows are coderivations, they respect the comultiplication of $\Sig$. It follows that these $\textbf{E}$-actions give $\Sig$ the structure of a Hopf monoid constructed internal to $\textbf{E}$-modules. 

By inspecting the definitions, we see that
\begin{equation}\label{eq:retardadvan}
Y\! \downarrow \tH_{(I)}=
\mathtt{R}_{(Y;I)}:=\! \sum_{Y_1\sqcup Y_2=Y}\overline{\tH}_{(Y_1)}  \tH_{(Y_2\sqcup I)} 
\qquad \text{and} \qquad  
Y\! \uparrow \tH_{(I)}=
\mathtt{A}_{(Y;I)}:=\! \sum_{Y_1\sqcup Y_2=Y}   \tH_{(Y_1\sqcup I)}    \overline{\tH}_{(Y_2)}
.
\end{equation}
It follows that 
\[     
Y\! \downarrow \tH_F
=  
\sum_{Y_1 \sqcup\dots\sqcup Y_k =Y} \mathtt{R}_{(Y_1;S_1)}\dots \mathtt{R}_{(Y_{k};S_{k})} 
\qquad \text{and} \qquad 
Y\! \uparrow \tH_F
=  
\sum_{Y_1 \sqcup\dots\sqcup Y_k =Y} \mathtt{A}_{(Y_1;S_1)}\dots \mathtt{A}_{(Y_{k};S_{k})} 
.\]
The sums are over all decompositions $(Y_1,\dots, Y_k)$ of $Y$ of length $l(F)$. We call \hbox{$\mathtt{R}_{(Y;I)}, \mathtt{A}_{(Y;I)}\in \Sig[Y\sqcup I]$} the \emph{retarded} and \emph{advanced} elements respectively. The \emph{total retarded} and \emph{total advanced} elements are given by
\[          
Y\! \downarrow \tH_{(i)}=
\mathtt{R}_{(Y;i)}  =\sum_{Y_1\sqcup Y_2=Y}\overline{\tH}_{(Y_1)}\,  \tH_{(Y_2 i)}
\qquad\text{and}\qquad    
Y\! \uparrow \tH_{(i)}=
\mathtt{A}_{(Y;i)}=\sum_{Y_1\sqcup Y_2=Y}   \tH_{(Y_2 i )}\,    \overline{\tH}_{(Y_1)}
\]
respectively. 
\begin{remark}\label{rem:double}
If we put $I=J\sqcup \{i\}$, then we have
\[
\mathtt{R}_{(J;i)}=
\sum_{\substack{S\sqcup T=I\\ i\in T}}\overline{\tH}_{(S)}\,  \tH_{(T)}
=-\sum_{\substack{F\in \Sigma[I]\\ i\in S_k}} (-1)^{l(F)} \tH_F
=\mathtt{D}_i 
\]
and
\[
\mathtt{A}_{(J;i)}=
\sum_{\substack{S\sqcup T=I\\ i\in T}} \tH_{(T)}\, \overline{\tH}_{(S)}=
-\sum_{\substack{F\in \Sigma[I]\\ i\in S_1}} (-1)^{l(F)} \tH_F=
\mathtt{D}_{\bar{i}}
.\]
\end{remark}

%%%%%%%%%%%%%%%%%%%%%%%%%%%%%%%%%%%%%%%%%%%%%%%%%%%%%%%%%%%%%%%%%%%%%%%%
\subsection{Currying the Steinmann Arrows}\label{Coalgebras}  
%%%%%%%%%%%%%%%%%%%%%%%%%%%%%%%%%%%%%%%%%%%%%%%%%%%%%%%%%%%%%%%%%%%%%%%%

Given a species $\textbf{p}$, we let $\textbf{p}^\textbf{E}$ denote the species given by
\[
\textbf{p}^{\textbf{E}}[I] := \prod_{r=0}^\infty\big (\textbf{p}^{[r]}[I]\big )^{ \sfS_r }
.\]
Here, $\textbf{p}^{[r]}$ is the $Y$-derivative of $\textbf{p}$ for $Y=[r]$, and $(-)^{ \sfS_r }$ denotes the subspace of $\sfS_r$-invariants, where $\sfS_r$ is the symmetric group on $[r]$. We denote elements of $\textbf{p}^{\textbf{E}}[I]$ using formal power series notation
\[
\sum_{r=0}^\infty \mathtt{x}_r, \qquad    \mathtt{x}_r\in \textbf{p}^{[r]}[I]
.\] 
Explicitly, $\mathtt{x}_r$ is an element of the vector space $\textbf{p} [   \{ \ast_1, \dots, \ast_r  \}   \sqcup I]$ which is invariant under the action of permuting $\{ \ast_1, \dots, \ast_r  \}$ and leaving $I$ fixed. 

The mapping $\textbf{p}\mapsto \textbf{p}^{\textbf{E}}$ extends to an endofunctor on species. In particular, given a morphism of species $\eta:\textbf{p} \to \textbf{q}$, we have the morphism $\eta^{\textbf{E}}$ given by
\begin{equation}\label{eq:endo}
\eta^{\textbf{E}} : \textbf{p}^{\textbf{E}} \to \textbf{q}^{\textbf{E}} 
,\qquad
\sum_{r=0}^\infty \mathtt{x}_r\mapsto      \sum_{r=0}^\infty   \eta_{[r]\sqcup I}(\mathtt{x}_r)
.
\end{equation}
A \emph{series} of a species $\textbf{p}$ is a morphism of species of the form $s:\textbf{E}\to \textbf{p}$. Notice the elements of $\textbf{p}^{\textbf{E}}[I]$ are naturally series of the species $Y \mapsto \textbf{p}^{[Y]}[I]$. See \cite[Section 3.2]{norledge2020species} for more details.	For the connection between $\textbf{p}^\textbf{E}$ and the internal hom for the Cauchy product, see \cite[Section 2.3]{norledge2020species}.

If $\textbf{a}$ is an algebra in species, then so is $\textbf{a}^\textbf{E}$, see \cite[Equation 12]{norledge2020species}. In particular, $\Sig^\textbf{E}$ is an algebra, with multiplication given by
\[
\sum_{r=0}^\infty \mathtt{x}_r \otimes \sum_{r=0}^\infty \mathtt{y}_r \ \mapsto \ 
	\sum_{r=0}^\infty \sum_{r_1 + r_2 =r}
	\dfrac{r!}{r_1 !\,  r_2 !}
	\mu_{[r_1] \sqcup S, [r_2] \sqcup T}(\mathtt{x}_{r_1} \otimes \mathtt{y}_{r_2})
.\]

\begin{thm} \label{coalgahomo}
We have the following homomorphisms of algebras in species,
\[
\Sig\to \Sig^{\textbf{E}}, 
\qquad
\tH_F\mapsto
\sum_{r=0}^\infty\,  \sum_{Y_1 \sqcup\dots\sqcup Y_k =[r]}   
\mathtt{R}_{(Y_1;S_1)}  \dots \mathtt{R}_{(Y_{k};S_{k})}
\]
and
\[
\Sig\to \Sig^{\textbf{E}}, 
\qquad
\tH_F\mapsto 
\sum_{r=0}^\infty\,  \sum_{Y_1 \sqcup\dots\sqcup Y_k =[r]}   
\mathtt{A}_{(Y_1;S_1)}  \dots \mathtt{A}_{(Y_{k};S_{k})}
.\]
\end{thm}
\begin{proof}
The Steinmann arrows are commutative up biderivations of $\Sig$, and so give $\Sig$ the structure of a Hopf $\textbf{E}$-algebra. This result is then a special case of \cite[Theorem 5.1]{norledge2020species}.
\end{proof}

The homomorphisms of \autoref{coalgahomo} are the unique extensions of the maps
\[
\tH_{(I)}\mapsto \sum_{r=0}^\infty \mathtt{R}_{(r;I)}
\qquad \text{and} \qquad
\tH_{(I)}\mapsto \sum_{r=0}^\infty \mathtt{A}_{(r;I)}
\]
to homomorphisms. In the application to causal perturbation theory, we shall be interested in the decorated analog of these homomorphisms, see \autoref{sec:Perturbation of T-Products by Steinmann Arrows}.

\begin{remark}
These homomorphisms $\Sig\to \Sig^\textbf{E}$ come from currying the $\textbf{E}$-actions of the Steinmann arrows. See \cite[Section 5.1]{norledge2020species} for details.
\end{remark}

%%%%%%%%%%%%%%%%%%%%%%%%%%%%%%%%%%%%%%%%%%%%%%%%%%%%%%%%%%%%%%%%%%%%%%%%
\subsection{The Steinmann Arrows and Dynkin Elements} \label{sec:The Steinmann Arrows and Dynkin Elements}
%%%%%%%%%%%%%%%%%%%%%%%%%%%%%%%%%%%%%%%%%%%%%%%%%%%%%%%%%%%%%%%%%%%%%%%%

\begin{figure}[t]
	\centering
	\includegraphics[scale=0.6]{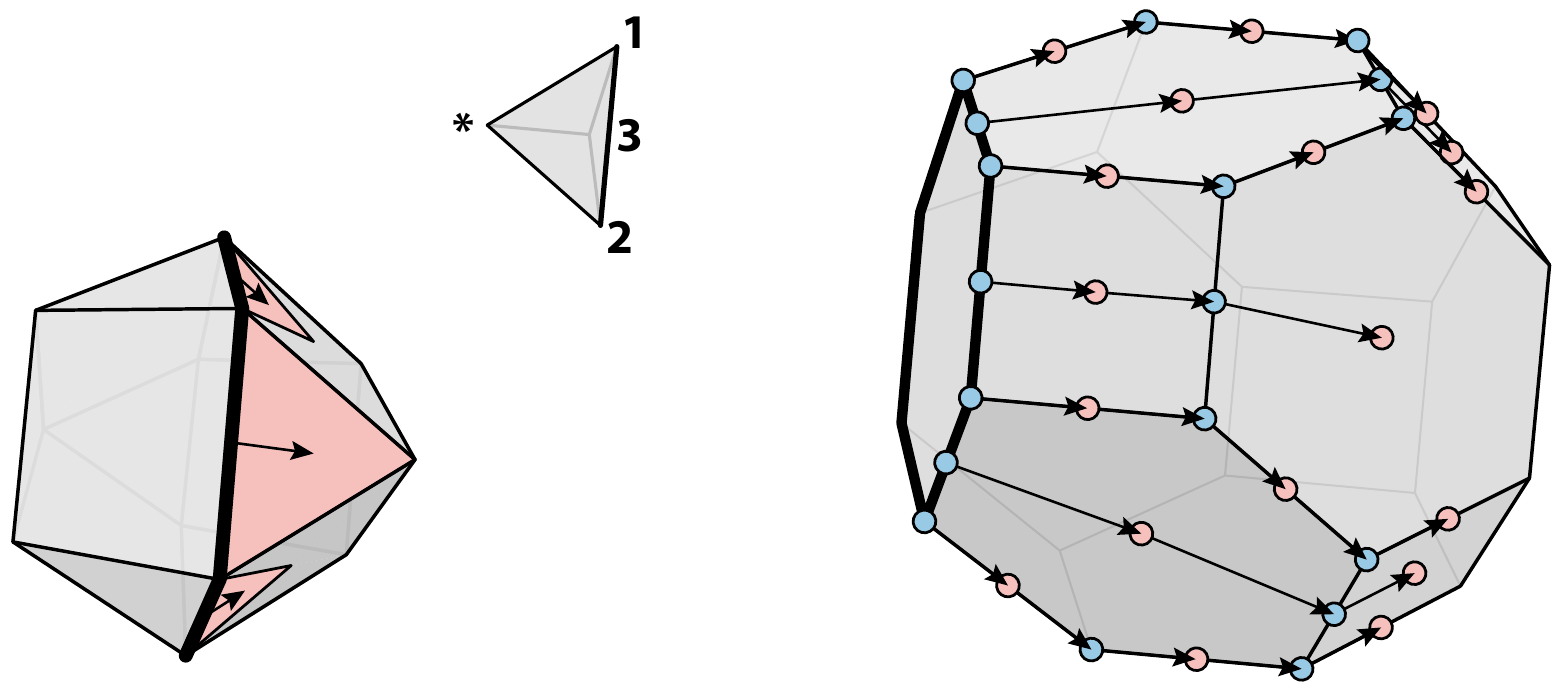}
	\caption{Schematic for the action of the retarded Steinmann arrow $\ast \downarrow$ for $I=\{1,2,3\}$ on the Steinmann sphere (left) and the tropical geometric realization of $\boldsymbol{\Sigma}$ (right, see \cite[Introduction]{norledge2019hopf}).}
	\label{fig:steinarrow}
\end{figure}

We now show that the restriction of the Steinmann arrows to $\Zie$, which are derivations for its Lie bracket, have an interesting description in terms of cells, i.e. chambers of the adjoint braid arrangement.

Following \cite[Section 2]{epstein2016}, we define the commutative up operators
\[
\ast \downarrow(-):\textbf{L}^\vee\to {\textbf{L}^\vee}',
\qquad
\ast \downarrow \cS:=\big \{(\ast S,T),(S,\ast T),(I,\ast):(S,T)\in \cS\big\}
\]
and
\[
\ast\uparrow(-):\textbf{L}^\vee\to {\textbf{L}^\vee}',
\qquad
\ast\uparrow \cS:=\big \{(\ast S,T),(S,\ast T),(\ast,I):(S,T)\in \cS\big\}
.\]
These are indeed well-defined; $\ast \downarrow \cS$ corresponds to the adjoint braid arrangement chamber on the $I$ side of the hyperplane $\lambda_{\ast,I}=0$ which has the face of $\cS$ as a facet, and $\ast\uparrow \cS$ corresponds to the chamber on the $\ast$ side of the hyperplane $\lambda_{\ast,I}=0$ which has the face of $\cS$ as a facet. See around \cite[Remark 2.2]{lno2019} for more details. Thus, it follows from \autoref{prop:ruelle} (Ruelle's identity) that
\[     
[ \tH_{(\ast)}, \mathtt{D}_\cS]= \mathtt{D}_{\ast \uparrow \cS} - \mathtt{D}_{\ast \downarrow \cS} 
.\]
The induced $\textbf{E}$-modules are given by
\[
\textbf{E}\bigcdot \textbf{L}^\vee \to \textbf{L}^\vee
,\qquad 
\tH_Y\otimes \cS \mapsto Y\downarrow \cS:=
\big\{ 
( Y_1\sqcup S, Y_2\sqcup T )\in [Y \sqcup I;\text{2}]  : (S,T)\in \cS\ \text{or}\ S=I   
\big\}   
\]
and
\[
\textbf{E}\bigcdot \textbf{L}^\vee \to \textbf{L}^\vee
,\qquad 
\tH_Y\otimes \cS \mapsto Y\uparrow \cS:=
\big\{ 
( Y_1\sqcup S, Y_2\sqcup T )\in [Y \sqcup I;\text{2}]  : (S,T)\in \cS\ \text{or}\ T=I   
\big\}      
.\]

\begin{prop}\label{adjointinterp}
Given a cell $\cS$ over $I$, we have
\[        
Y \downarrow\mathtt{D}_\cS=  \mathtt{D}_{Y \downarrow \cS} 
\qquad \text{and} \qquad
Y \uparrow\mathtt{D}_\cS=  \mathtt{D}_{Y \uparrow \cS}
.\]
\end{prop}
\begin{proof}
We consider the retarded case $Y\downarrow\mathtt{D}_\cS=  \mathtt{D}_{Y \downarrow \cS}$ only, since the advanced case then follows similarly. It is sufficient to consider the case $Y=\{\ast\}$. We have
\[   
\downarrow\mathtt{D}_\cS
=  
-\sum_{ \bar{F} \subseteq \cS } (-1)^{l(F)}  \downarrow\tH_{F} 
\qquad \text{and} \qquad 
\mathtt{D}_{\downarrow \cS}= -\sum_{\bar{F} \subseteq \,   \downarrow\cS } (-1)^{l(F)}\,  \tH_{ F }
.\]
So, the result follows if we have the following equality
\[ 
\sum_{ \bar{F} \subseteq \cS } (-1)^{l(F)} \sum_{1\leq m\leq k}
- \mathtt{H}_{(S_1, \dots,\ast, S_m, \dots,S_k)} 
+\mathtt{H}_{(S_1,\dots, \ast S_m,\dots, S_k)}   
\overset{\mathrm{?}}{=} 
\sum_{\bar{G} \subseteq \,   \downarrow\cS } (-1)^{l(G)}\,  \tH_{G}    
.\]
Indeed, notice that the $\tH$-basis elements $\tH_G\in \Sig[\ast I]$ which appear on the LHS are exactly those such that 
\[\bar{G}\subseteq  \downarrow\cS.\] 
Notice also that each $\tH_G$ appears with total sign $(-1)^{l(G)}$, since when $\ast$ is inserted as a singleton lump, thus increasing $l(G)$ by one, it appears also with a negative sign. 
\end{proof} 

\begin{remark}
This interpretation of the $\textbf{E}$-module structure of $\Sig$ restricted to the primitive part $\Zie=\mathcal{P}(\Sig)$ in terms of the adjoint braid arrangement suggests obvious generalizations of the Steinmann arrows in the direction of \cite{aguiar2017topics}, \cite{aguiar2020bimonoids}, since the generalization of Hopf monoids there is via hyperplane arrangements.
\end{remark}

\begin{cor}
We have the following homomorphisms of Lie algebras in species,
\[
\Zie\to \Zie^{\textbf{E}}, 
\qquad
\mathtt{D}_\cS\mapsto \mathtt{D}_{ (-) \downarrow \cS }
=
\sum_{r=0}^\infty \mathtt{D}_{ [r] \downarrow \cS }
= \mathtt{D}_\cS+ \mathtt{D}_{\downarrow \cS}+ \mathtt{D}_{  \downarrow\downarrow \cS}+ \cdots 
\]
and
\[
\Zie\to \Zie^{\textbf{E}}, 
\qquad
\mathtt{D}_\cS\mapsto \mathtt{D}_{ (-) \uparrow \cS }
=
\sum_{r=0}^\infty \mathtt{D}_{ [r]\uparrow \cS }
= \mathtt{D}_\cS+ \mathtt{D}_{\uparrow  \cS}+ \mathtt{D}_{\uparrow \uparrow \cS}+ \cdots 
.\]
\end{cor}
\begin{proof}
The Steinmann arrows are commutative up biderivations of $\Zie$, and so give $\Zie$ the structure of a Lie $\textbf{E}$-algebra. This result is then a special case of \cite[Theorem 5.1]{norledge2020species}.
\end{proof}
%%%%%%%%%%%%%%%%%%%%%%%%%%%%%%%%%%%%%%%%%%%%%%%%%%%%%%%%%%%%%%%%%%%%%%%%
\section{Products and Series}
%%%%%%%%%%%%%%%%%%%%%%%%%%%%%%%%%%%%%%%%%%%%%%%%%%%%%%%%%%%%%%%%%%%%%%%%
We now recall several basic constructions of casual perturbation theory in the current, clean, abstract setting. We do this without yet imposing causal factorization/causal additivity. We say e.g. `$\text{T}$-product' and `$\text{R}$-product' for now, and then change to `time-ordered product' and `retarded product' in the presence of causal factorization.

%Whether systems of products for $\Sig$ are meaningful without some additional condition like causal factorization seems to be an interesting question.
%%%%%%%%%%%%%%%%%%%%%%%%%%%%%%%%%%%%%%%%%%%%%%%%%%%%%%%%%%%%%%%%%%%%%%%%
\subsection{T-Products, Generalized T-Products, and Generalized R-Products} \label{sec:T-Products, Generalized T-Products, and Generalized R-Products}
%%%%%%%%%%%%%%%%%%%%%%%%%%%%%%%%%%%%%%%%%%%%%%%%%%%%%%%%%%%%%%%%%%%%%%%%
Let $V$ be a vector space over $\bC$. Let $\cA$ be a $\bC$-algebra with multiplication denoted by $\star$. Let $\textbf{U}_\cA$ be the algebra in species given by 
\[
\textbf{U}_\cA[I]:= \cA
.\] 
The action of bijections is trivial, and the multiplication is the multiplication of $\cA$.

The \emph{positive exponential species} $\textbf{E}^\ast_+$ is given by
\[
\textbf{E}^\ast_+[I]:= \bC \qquad \text{if} \quad I\neq \emptyset \qquad \text{and} \qquad  \textbf{E}^\ast_+[\emptyset] =0
.\]
Let a \emph{system of} \hbox{$\text{T}$\emph{-products}} $\text{T}$ be a system of products for the positive exponential species $\textbf{E}^\ast_+$, as defined in \cite[Section 6.2]{norledge2020species}. This means $\text{T}$ is a morphism of species of the form\footnote{\ recall the definition and notation for $\textbf{E}_V$ from \autoref{Decorations}}
\[   
\text{T}:  \textbf{E}^\ast_+  \otimes \textbf{E}_V  \to \textbf{U}_\cA
,\qquad 
\tH_{(I)} \otimes \ssA_I \mapsto 
\text{T}_I(\tH_{(I)} \otimes \ssA_I)
\]
where recall $\textbf{E}^\ast_+  \otimes \textbf{E}_V$ is the Hadamard product of species, given by
\[
\textbf{E}^\ast_+  \otimes \textbf{E}_V [I] := \textbf{E}^\ast_+[I] \otimes \textbf{E}_V[I] 
.\]
Thus, if $I\neq 0$, we have 
\[
\textbf{E}^\ast_+  \otimes \textbf{E}_V [I] \cong   V^{\otimes I} 
.\]
We abbreviate
\begin{equation}\label{eq:abb}
\text{T}_I(\ssA_I):=\text{T}_I(\tH_{(I)} \otimes \ssA_I) 
.
\end{equation}
Let $\cH(\textbf{E}_V, \textbf{U}_\cA)$ denote the species of linear maps between components, given by
\[
\cH(\textbf{E}_V, \textbf{U}_\cA)[I] := \Hom_{\textsf{Vec}}\! \big (  \textbf{E}_V[I], \textbf{U}_\cA[I]   ) =  \Hom_{\textsf{Vec}}\! \big ( V^{\otimes I}, \cA   )
.\]
We have that $\cH(-, -)$ is the hom for the Hadamard product. Therefore we can curry $\text{T}$ to give the morphism of species
\[          
\textbf{E}^\ast_+\to \cH(\textbf{E}_V, \textbf{U}_\cA), \qquad \tH_{(I)}\mapsto \text{T}(I)
\]
where $\text{T}(I)$ is the linear map
\[
\text{T}(I): V^{\otimes I} \to \cA
,\qquad
\ssA_I \mapsto \text{T}_I(\ssA_I)
.\] 
The linear maps $\text{T}(I)$ are called $\text{T}$\emph{-products}. Notice that $\text{T}$-products are commutative in the sense that 
\[
\text{T}_I\big( \textbf{E}_V[\sigma](\ssA_I)\big)=\text{T}_I(\ssA_I) \qquad \quad \text{for all bijections}\quad \sigma:I\to I
.\]
This property holds because the system $\text{T}$ is a morphism of species, and bijections act trivially for $\textbf{U}_\cA$. This commutativity exists despite the fact that the algebra $\cA$ is noncommutative in general. 

\begin{remark}
In applications to QFT, we shall also have a causal structure on $V$. Then $\text{T}$ is meant to first order the vectors of $\ssA_I$ according to the causal structure, and then multiply in $\cA$, giving rise to this commutativity. 
\end{remark}

%The space of system of $\text{T}$-products forms the convolution algebra $\Hom(\textbf{E}^\ast_+\otimes \textbf{E}_V)$. 
%We say system of $\text{T}$-products beccause in QFT, this commutativity corresponds to time-ordering. 

Let the \emph{system of generalized} $\text{T}$\emph{-products} associated to a system of $\text{T}$-products be the unique extension to a system of products for $\Sig=\textbf{L}\boldsymbol{\circ}\textbf{E}_+^\ast$ which is a homomorphism, as defined in \cite[Section 6.2]{norledge2020species}. Thus
\[  
\text{T}:  \Sig\otimes \textbf{E}_V  \to \textbf{U}_\cA
,\qquad 
\tH_F\otimes \ssA_I\mapsto \text{T}_I(\tH_F\otimes \ssA_I):=
\text{T}_{S_1}(\ssA_{S_1}) \star \dots \star \text{T}_{S_k}(\ssA_{S_k})
.\]
The currying of $\text{T}$ is denoted by
\[          
\Sig\to \cH(\textbf{E}_V, \textbf{U}_\cA), \qquad \tH_{F}\mapsto \text{T}(S_1)\dots \text{T}(S_k)
.\]
The linear maps 
\[
\text{T}(S_1)\dots \text{T}(S_k):V^{\otimes I} \to \cA
,\qquad
\ssA_I\mapsto \text{T}_I(\tH_F\otimes \ssA_I)
\] 
are called \emph{generalized} $\text{T}$-\emph{products}. Let the \emph{system of generalized} $\text{R}$\emph{-products} associated to a system of $\text{T}$-products be the restriction to the Lie algebra of primitive elements $\Zie$,
\[  
\text{R}: \Zie \otimes \textbf{E}_V  \to \textbf{U}_\cA, \qquad  
\mathtt{D}_\cS\otimes \ssA_I\mapsto
\text{R}_I(\mathtt{D}_\cS \otimes \ssA_I):=\text{T}_I(\mathtt{D}_\cS  \otimes \ssA_I)   
.\]
This is a morphism of Lie algebras, where $\textbf{U}_\cA$ is equipped with the commutator bracket. The currying of $\text{R}$ is denoted by
\[          
\Zie \to \cH(\textbf{E}_V, \textbf{U}_\cA), \qquad \mathtt{D}_\cS\mapsto \text{R}_\cS
.\]
The linear maps 
\[
\text{R}_\cS: \textbf{E}_V[I] \to \cA
,\qquad
\ssA_I\mapsto  \text{R}_I(\mathtt{D}_\cS \otimes \ssA_I) 
\] 
are called \emph{generalized} $\text{R}$-\emph{products}. From the expansion \textcolor{blue}{(\refeqq{eq:hbasisexp})} of Dynkin elements $\mathtt{D}_\cS$ in terms of the $\tH$-basis, we recover \cite[Equation 79]{ep73roleofloc},
\[ 
\text{R}_\cS=-\sum_{\cF_{F}\subseteq\bar{\cS}} (-1)^{k}\, \text{T}(S_1)\dots \text{T}(S_k)
.\]

\begin{remark}
Consider a system of products of the form
\[ 
\text{Z}:\textbf{E}^\ast_+\otimes {\textbf{E}_{V}} 
\to 
\textbf{U}_{V}
,\qquad 
\tH_{(I)}\otimes \ssA_I\mapsto \text{Z}_I(\ssA_I)
.\]
Then we obtain a new $\text{T}$-product $\text{T}'$, given by
\[  
\text{T}':  
\textbf{E}^\ast_+\otimes  \textbf{E}_V  \to \textbf{U}_\cA
,\qquad 
\text{T}'_I(\ssA_I)
:= 
\sum_{P}\text{T}_P     
\big(
\text{Z}_{S_1} (\ssA_{S_1})\dots  \text{Z}_{S_k}(\ssA_{S_k})
\big)    
.\]
The sum is over all partitions $P=\{S_1,\dots,S_k\}$ of $I$. This construction underlies renormalization in pAQFT \cite[Section 3.6.2]{dutsch2019perturbative}, which deals with the remaining ambiguity of $\text{T}$-products after imposing causal factorization, and perhaps other renormalization conditions.
\end{remark}

%In the presence of a causal structure and causal factorization, generalized $\text{R}$-products will have certain important support properties, as a consequence of \autoref{actondynkinwithtits}. 
%\begin{remark}
%The definition of a system of $\text{T}$-products does not use the algebraic structure of $\cA$. If $V$ is a space of bump functions on a spacetime $\cX$, then causal factorization is an additional property one can impose which uses the algebraic structure of $\cA$. We shall call such system of $\text{T}$-products time-ordered products.
%\end{remark}

%\begin{remark}
%A system of products of the form
%\[\Sig\otimes \textbf{E}_V  \to \textbf{U}_\cA\] 
%is a natural noncommutative generalization of the system of products for polynomial functions from \autoref{ex:polyfun},
%\[\textbf{E}\otimes \textbf{E}_V  \to \textbf{U}_{\text{PolyFunc}(V^\ast)}.\] 
%Products of the form $\Sig\otimes \textbf{E}_V  \to \textbf{U}_\cA$ appear in QFT, however their significance there relies on the fact that $V$ consists of bump functions on a spacetime, which has a causal structure. In particular, one imposes the property of causal factorization on systems of $\text{T}$-products. %So one can ask, are products $\Sig\otimes \textbf{E}_V  \to \textbf{U}_\cA$ meaningful without this causal structure and causal factorization.
%\end{remark}
  
%%%%%%%%%%%%%%%%%%%%%%%%%%%%%%%%%%%%%%%%%%%%%%%%%%%%%%%%%%%%%%%%%%%%%%%%
\subsection{Reverse T-Products}\label{sec:rev T-Exponentials}
%%%%%%%%%%%%%%%%%%%%%%%%%%%%%%%%%%%%%%%%%%%%%%%%%%%%%%%%%%%%%%%%%%%%%%%%

The system of \emph{reverse generalized} $\text{T}$\emph{-products} $\overline{\text{T}}$ of a system of generalized $\text{T}$-products is given by precomposing $\text{T}$ with the antipode \textcolor{blue}{(\refeqq{eq:antipode})} of $\Sig\otimes \textbf{E}_V$, thus
\[ 
\overline{\text{T}}:\Sig \otimes {\textbf{E}}_V\to \textbf{U}_{\cA^{\op}}, 
\qquad 
\overline{\text{T}}_I(\tH_{F}\otimes \ssA_I)
:=
{\text{T}}_I\big(\overline{\tH}_{F}\otimes \ssA_I\big)
.\]
Since the antipode is a homomorphism $\Sig\otimes \textbf{E}_V\to (\Sig\otimes \textbf{E}_V)^{\op, \text{cop}}$ \cite[Proposition 1.22 (iii)]{aguiar2010monoidal}, this is a system of generalized $\text{T}$-products into the opposite algebra $\textbf{U}_{\cA^{\op}}$. The image of $\tH_{(I)}$ under the currying of $\overline{\text{T}}$ is called the \emph{reverse} $\text{T}$\emph{-product} 
\[
\overline{\text{T}}(I): \textbf{E}_V[I] \to \cA^{\text{op}}
.\] 
From \textcolor{blue}{(\refeqq{antipode})}, we obtain
\[
\overline{\text{T}}(I) 
=
\sum_{F\in \Sigma[I]} (-1)^{k}\, \text{T}(S_1)\dots \text{T}(S_k)
.\]
Note that reverse $\text{T}$-products in \cite[Equation 11]{ep73roleofloc} are defined to be $(-1)^n\, \overline{\text{T}}(I)$. Our definition agrees with \cite[Definition 15.35]{perturbative_quantum_field_theory}.

%%%%%%%%%%%%%%%%%%%%%%%%%%%%%%%%%%%%%%%%%%%%%%%%%%%%%%%%%%%%%%%%%%%%%%%%
\subsection{T-Exponentials} \label{sec:T-Exponentials}
%%%%%%%%%%%%%%%%%%%%%%%%%%%%%%%%%%%%%%%%%%%%%%%%%%%%%%%%%%%%%%%%%%%%%%%%

For details on series in species, see \cite[Section 12]{aguiar2010monoidal}. The (scaled) \emph{universal series} $\mathtt{G}(c)$ is the group-like series of $\Sig$ given by
\[   
\mathtt{G}(c):  \textbf{E} \to  \Sig
,\qquad  
\tH_{I}\mapsto \mathtt{G}(c)_{I}:= c^n\,  \tH_{(I)}    
\qquad  \quad   
\text{for} \quad  
c\in \bC  
.\]
The fundamental nature of this series is described in \cite[Section 13.6]{aguiar2013hopf}. The series $\text{s} \circ \mathtt{G}(c)$ which is the composition of $\mathtt{G}(c)$ with the antipode $\text{s}$ of $\Sig$ is given by
\begin{equation}\label{inverseuni}
\text{s} \circ \mathtt{G}(c): \textbf{E} \to  \Sig
,\qquad  
\tH_I\mapsto \big (\text{s} \circ \mathtt{G}(c)\big )_I= c^n\, \overline{\tH}_{(I)}
.
\end{equation}
Let $\cA[[\formj]]$ denote the $\bC$-algebra of formal power series in the formal symbol $\formj$ with coefficients in $\cA$. Given a system of generalized $\text{T}$-products 
\[
\text{T}:\Sig \otimes {\textbf{E}}_V\to \textbf{U}_\cA
\] 
let the $\text{T}$-\emph{exponential} $\mathcal{S}:=\mathcal{S}_{\mathtt{G}(c)}$ of this system be the $\cA[[\formj]]$-valued function on the vector space $V$ associated to the series $\mathtt{G}(c)$, as constructed in \cite[Section 6.3]{norledge2020species}. Thus, we have\footnote{\ we use the abbreviations \textcolor{blue}{(\refeqq{eq:simpletensors})} and \textcolor{blue}{(\refeqq{eq:abb})}, and also $\text{T}_n:=\text{T}_{[n]}$}
\begin{equation}\label{eq:tsexp}
\mathcal{S}: V\to \cA[[ \formj ]] , 
\qquad        
\ssA\mapsto  \mathcal{S}(\formj\! \ssA)
=
\sum^\infty_{n=0} \dfrac{c^n}{n!} \text{T}_n 
\underbrace{\big ( \formj\! \ssA\otimes  \cdots \otimes \formj\! \ssA \big )}_{ \text{ $n$ times } }   
:=
\sum^\infty_{n=0} \dfrac{\formj^n c^n}{n!} \text{T}_n(\ssA^n)     
.
\end{equation}
%In the presence of causal factorization, $\mathcal{S}$ will be called an S-matrix scheme. Then the following is the result that reverse system of $\text{T}$-products express the inverse S-matrix \cite[Section 2]{ep73roleofloc}. 
By \cite[Equation 34]{norledge2020species} and \textcolor{blue}{(\refeqq{inverseuni})}, the $\text{T}$-exponential for the system of reverse $\text{T}$-products is the inverse of $\mathcal{S}$ as an element of the $\bC$-algebra of functions $\text{Func}(V,\cA[[\formj]])$, given by
\[
\mathcal{S}^{-1}: V\to \cA[[ \formj ]] , 
\qquad        
\ssA\mapsto  \mathcal{S}^{-1}(\formj\! \ssA) :=
\sum^\infty_{n=0} \dfrac{\formj^n c^n}{n!} \overline{\text{T}}_n(\ssA^n)
=
\sum^\infty_{n=0} \dfrac{\formj^n c^n}{n!} \text{T}_n(\overline{\tH}_{(n)}\otimes \ssA^n)
.\]
Therefore
\[
\mathcal{S}(\formj\! \ssA)\star \mathcal{S}^{-1}(\formj\! \ssA)=
\mathcal{S}^{-1}(\formj\! \ssA)\star \mathcal{S}(\formj\! \ssA) = 1_\cA
\]
for all $\ssA\in V$. This appears in e.g. \cite[Equation 2]{ep73roleofloc}.

%%%%%%%%%%%%%%%%%%%%%%%%%%%%%%%%%%%%%%%%%%%%%%%%%%%%%%%%%%%%%%%%%%%%%%%%
\section{Perturbation of T-Products}
%%%%%%%%%%%%%%%%%%%%%%%%%%%%%%%%%%%%%%%%%%%%%%%%%%%%%%%%%%%%%%%%%%%%%%%%
%We now construct three perturbations of T-products by $\textbf{E}$-actions, as considered in \autoref{sec:Perturbation of Products and Series}.

For the perturbation of $\text{T}$-products by a certain up coderivation of $\textbf{E}$ which gives the S-matrix scheme $\mathcal{S}_{\formg \ssS}(\formj\! \ssA)= \mathcal{S}(\formg \ssS+ \formj\! \ssA)$, see \cite[Section 10.1]{norledge2020species}.

%%%%%%%%%%%%%%%%%%%%%%%%%%%%%%%%%%%%%%%%%%%%%%%%%%%%%%%%%%%%%%%%%%%%%%%%
\subsection{Perturbation of T-Products by Steinmann Arrows}  \label{sec:Perturbation of T-Products by Steinmann Arrows}
%%%%%%%%%%%%%%%%%%%%%%%%%%%%%%%%%%%%%%%%%%%%%%%%%%%%%%%%%%%%%%%%%%%%%%%%
Suppose we have a system of generalized \hbox{$\text{T}$-products}
\[   
\text{T}:  \Sig \otimes \textbf{E}_V  \to \textbf{U}_\cA
,\qquad  
\tH_F\otimes \ssA_I \mapsto \text{T}_I(\tH_F\otimes \ssA_I)
.\]
Following \cite[Section 6.4]{norledge2020species}, given a choice of decorations vector $\ssS\in V$, we can use the retarded Steinmann arrow \textcolor{blue}{(\refeqq{steindown})} to perturb $\text{T}$ as follows. 

Recall the decorated Hopf algebra $\Sig\otimes\textbf{E}_V$ from \autoref{Decorations}. Recall also the derivative $(\Sig\otimes\textbf{E}_V)'$ of $\Sig\otimes\textbf{E}_V$ from \autoref{Derivations and Coderivations of Sig}, given by
\[
(\Sig\otimes\textbf{E}_V)'[I] =  \Sig[\ast I]\otimes V\otimes V^{\otimes I}
.\]
We have the up derivation of $\Sig\otimes\textbf{E}_V$ which is the decorated analog of the retarded Steinmann arrow, given by
\[
\Sig\otimes\textbf{E}_V \to (\Sig\otimes\textbf{E}_V)', \qquad  
\tH_F \otimes \ssA_{i_1} \otimes \dots \otimes   \otimes \ssA_{i_n}
\ \mapsto \    
\ast \downarrow \tH_F \otimes  \ssS\otimes  \ssA_{i_1} \otimes \dots \otimes   \otimes \ssA_{i_n}
.\]
This is indeed still an up derivation by \cite[Proposition 6.4]{norledge2020species}. Analogous to the setting without decorations, we have the induced raising operators and associated $\textbf{E}$-action by iterating, which, after currying, give us the homomorphism
\[
\Sig\otimes\textbf{E}_V \to (\Sig\otimes\textbf{E}_V)^{\textbf{E}}
\]
\[ 
\tH_F \otimes   \ssA_{i_1} \otimes \dots \otimes   \otimes \ssA_{i_n} 
\ \mapsto \
\sum^{\infty}_{r=0} \underbrace{\downarrow \dots  \downarrow}_{\text{$r$ times}} \tH_F  \otimes\underbrace{\ssS \otimes \dots \otimes \ssS}_{\text{$r$ times}}\otimes    \ssA_{i_1} \otimes \dots \otimes   \otimes \ssA_{i_n} 
.\]
This is a homomorphism by \cite[Theorem 5.1]{norledge2020species}. Then, a new `perturbed' system of generalized \hbox{$\text{T}$-products} is given by composing this homomorphism with $\text{T}^{\textbf{E}}$ (defined in \textcolor{blue}{(\refeqq{eq:endo})}),
\[ 
\widetilde{\text{T}} : \Sig\otimes\textbf{E}_V \to     (\Sig\otimes\textbf{E}_V )^{\textbf{E}} \xrightarrow{\text{T}^\textbf{E}}  (\textbf{U}_\cA)^{\textbf{E}} \cong  \textbf{U}_{\cA[[\formg]]}
.\]
For the result that $(\textbf{U}_\cA)^{\textbf{E}} \cong  \textbf{U}_{\cA[[\formg]]}$, see \cite[Section 4]{norledge2020species}. 

\begin{remark}
The fact $\widetilde{\text{T}}$ is still a homomorphism, and is thus still a generalized system of products, depends crucially on the fact the Steinmann arrow is a derivation \cite[Theorem 5.1]{norledge2020species}, and that $(-)^\textbf{E}$ is a monoidal functor \cite[Section 2.5]{norledge2020species}. We can similarly perturb a system of generalized $\text{R}$-products, which uses the fact the Steinmann arrow is a biderivation.
\end{remark}

We now unpack all this formalism to give a fully explicit description of the new perturbed system of products. Let us abbreviate
\[
\ssS_Y\ssA_I=  \ssS_{y_1}\otimes \dots \otimes \ssS_{y_r} \otimes \ssA_{i_1} \otimes \dots \otimes \ssA_{i_n} \in \textbf{E}_V[Y\sqcup I]
.\]
Let
\begin{equation}\label{eq:retardprod}
\text{R}_{Y;I}(\ssS_Y;\ssA_I)
:= 
\underbrace{\text{T}_{Y\sqcup I}( \mathtt{R}_{(Y;I)} \otimes \ssS_Y \ssA_I)
=
\sum_{Y_1 \sqcup Y_2=Y}  \overline{\text{T}}_{Y_1 \sqcup \emptyset}(\ssS_{Y_1}) \star \text{T}_{Y_2 \sqcup I}( \ssS_{Y_2} \ssA_I)}_{\text{by } \textcolor{blue}{(\refeqq{eq:retardadvan}})} 
\end{equation}
Then the new perturbed system is given by\footnote{\ we abbreviate $\text{R}_{r;I}(\ssS^{\, r};\ssA_I):=\text{R}_{[r];I}(\ssS^{\, [r]};\ssA_I)=\text{R}_{[r];I}(\underbrace{\ssS\otimes \dots \otimes \ssS}_{\text{$r$ times}} \, ;  \ssA_I)$}
\[ 
\widetilde{\text{T}}:\Sig\otimes\textbf{E}_V \to \textbf{U}_{\cA[[\formg]]}
,\quad 
\tH_F\otimes \ssA_I \mapsto \sum_{r=0}^\infty\    \sum_{r_1 +\, \cdots\,  + r_k=r } \dfrac{\formg^{r}}{r!}  \text{R}_{r_1;S_1}(\ssS^{\, r_1};\ssA_{S_1})\star \cdots\star \text{R}_{r_k;S_k}(\ssS^{\, r_k};\ssA_{S_k}) 
.\]
In particular, the restriction to $\textbf{E}_+^\ast\otimes\textbf{E}_V$, i.e. the new perturbed $\text{T}$-product, is given by
\begin{align*}
\widetilde{\text{T}}_I(\ssA_I)&= \sum_{r=0}^\infty \dfrac{\formg^{r}}{r!}  \text{R}_{r;I}(\ssS^{\, r};\ssA_I) \\
&=\text{T}_I(\ssA_I) + \underbrace{\formg\, \text{T}_{\ast_1 I}(\downarrow \tH_{(I)} \otimes \ssS  \ssA_I) + \dfrac{\formg^2}{2!} \text{T}_{\ast_2 \ast_1 I}(\downarrow \downarrow \tH_{(I)} \otimes \ssS \ssS   \ssA_I)  + \cdots}_{\text{perturbation}}\ .
\end{align*}
Similar, we can perturb a system of generalized T-products using the advanced Steinmann arrow. 

We let $\mathcal{V}_{\! \formg\ssS}$, respectively $\mathcal{W}_{\! \formg\ssS}$, denote the $\text{T}$-exponential (as defined in \textcolor{blue}{(\refeqq{eq:tsexp})}) for the new perturbed system of generalized $\text{T}$-products using the retarded, respectively advanced, Steinmann arrows. Thus
\[
\mathcal{V}_{\! \formg\ssS}: V\to \cA[[\formg,\! \formj]],
\qquad
\mathcal{V}_{\! \formg\ssS}(\formj\! \ssA)
:=
\sum_{n=0}^\infty 
\dfrac{\formj^n c^n}{n!}\, \wt{\text{T}}_n (\ssA^n)
=
\sum_{n=0}^\infty \sum_{r=0}^\infty 
\dfrac{\formg^{r} \formj^n c^{r+n}}{r!\, n!}\,  \text{R}_{r;n} (\ssS^{\, r} ; \ssA^n)         
\]
and
\[
\mathcal{W}_{\! \formg\ssS}: V\to \cA[[\formg,\! \formj]],
\qquad
\mathcal{W}_{\! \formg\ssS}(\formj\! \ssA):=\sum_{n=0}^\infty 
\dfrac{\formj^n c^n}{n!}\, \wt{\text{T}}_n (\ssA^n)
=
\sum_{n=0}^\infty \sum_{r=0}^\infty 
\dfrac{\formg^{r} \formj^n c^{r+n}}{r! \, n!}\,  \text{A}_{r;n} (\ssS^{\, r};\ssA^n)         
\]
where
\[
\text{A}_{Y;I}(\ssS_Y;\ssA_I)
:= 
\underbrace{\text{T}_{Y\sqcup I}(  \mathtt{A}_{(Y;I)} \otimes \ssS_Y \ssA_I)
	=
	\sum_{Y_1 \sqcup Y_2=Y}  \text{T}_{Y_1 \sqcup I}( \ssS_{Y_1} \ssA_I) \star \overline{\text{T}}_{Y_2\sqcup \emptyset}(\ssS_{Y_2})}_{\text{by } \textcolor{blue}{(\refeqq{eq:retardadvan}})} 
.\]

\begin{thm}
We have
\[
\mathcal{V}_{\formg \ssS}(\formj\! \ssA)= 
\mathcal{S}^{-1}( \formg \ssS)\star \mathcal{S}(\formg \ssS +\formj\! \ssA )
\qquad \text{and} \qquad 
\mathcal{W}_{\formg \ssS}(\formj\! \ssA)= 
\mathcal{S}(\formg \ssS +\formj\! \ssA )\star \mathcal{S}^{-1}( \formg \ssS)
.\]
\end{thm}
\begin{proof}
We have
\[         
\text{R}_{r;I}(\ssS^{\, r}; \ssA_I)=
\sum_{Y_1\sqcup Y_2=[r]} \overline{\text{T}}_{Y_1 \sqcup \emptyset}(\ssS^{Y_1})   
\star 
\text{T}_{Y_2\sqcup I}( \ssS^{Y_2}\ssA_I) 
.\]
Then
\begin{align*}
\mathcal{V}_{\formg \ssS}(\formj\! \ssA)
=&\  
\sum_{n=0}^\infty \sum_{r=0}^\infty 
\dfrac{\formg^{r} \formj^n c^{r+n}}{r!\, n!}\,  \text{R}_{r;n} (\ssS^{\, r} ; \ssA^n) \\[6pt]
=&\  
\sum_{n=0}^\infty \sum_{r=0}^\infty 
\dfrac{\formg^{r} \formj^n c^{r+n}}{r!\, n!} 
\sum_{Y_1\sqcup Y_2=[r]} \overline{\text{T}}_{Y_1 \sqcup \emptyset}(\ssS^{Y_1})   
\star 
\text{T}_{Y_2\sqcup [n]}(\ssS^{Y_2}\ssA^n) \\[6pt]
=& \ 
\sum^\infty_{r=0} \dfrac{\formg^{r} c^r}{r!} \overline{\text{T}}_{r+0}(\ssS^{\, r}) 
\star
\sum_{n=0}^\infty \sum_{r=0}^\infty \dfrac{c^n}{n!} \text{T}_{r+n} (\ssS^{\, r} \ssA^n) 
\\[6pt]
=& \ 
\mathcal{S}^{-1}( \formg \ssS)\star \mathcal{S}(\formg \ssS +\formj\! \ssA )
\end{align*}
The proof for $\mathcal{W}_{\formg \ssS}(\formj\! \ssA)$ is similar.
\end{proof}

\begin{cor}[Bogoliubov's Formula {\cite[Chapter 4]{Bogoliubov59}}] 
We have
\begin{equation} \label{eq:Bog}
\widetilde{\text{T}}_i(\ssA)
=
\dfrac{1}{c}\,  \dfrac{d}{d \formj }  \Bigr|_{\formj=0} \mathcal{V}_{\formg \ssS}(\formj\! \ssA)
.
\end{equation}
\end{cor}
\begin{proof}
We have
\[
\dfrac{d}{d \formj }\mathcal{V}_{\formg \ssS}(\formj\! \ssA)
=
\dfrac{d}{d \formj } \sum_{n=0}^\infty 
\dfrac{\formj^n c^n}{n!}\, \widetilde{\text{T}}_n (\ssA^n)
=
\sum_{n=1}^\infty 
\dfrac{\formj^{n-1} c^n}{(n-1)!}\,  \widetilde{\text{T}}_n (\ssA^n)
.\]
Then, putting $\formj=0$, we obtain
\[
\dfrac{d}{d \formj }  \Bigr|_{\formj=0}  \mathcal{V}_{\formg \ssS}(\formj\! \ssA) 
= 
c\,  \widetilde{\text{T}}_1 (\ssA)
.\qedhere \]
\end{proof}

This formula was originally motivated by the path integral heuristic, see e.g. \cite[Remark 15.16]{perturbative_quantum_field_theory}. %which is closely related to the idea that (interacting) vacuum expectation values are obtained by summing over all possible interactions of virtual particles. %The more modern and rigorous version of this is the so-called worldline formalism, which is a perspective on pQFT that comes from perturbative string theory.

%%%%%%%%%%%%%%%%%%%%%%%%%%%%%%%%%%%%%%%%%%%%%%%%%%%%%%%%%%%%%%%%%%%%%%%%
\subsection{$\text{R}$-Products and $\text{A}$-Products} 
%%%%%%%%%%%%%%%%%%%%%%%%%%%%%%%%%%%%%%%%%%%%%%%%%%%%%%%%%%%%%%%%%%%%%%%%
The linear maps $\text{R}(Y;I)$ which are given by
\[
\text{R}(Y;I):\textbf{E}_V^{[Y]}[I] \to \cA
,\qquad
\ssS_Y \ssA_I \mapsto \text{R}_{Y;I}(\ssS_Y \ssA_I)
\] 
are called R\emph{-products}. In the case of singletons $I=\{i\}$, the maps $\text{R}(Y;i)$ are called \emph{total} R\emph{-products}. By \textcolor{blue}{(\refeqq{eq:retardadvan})}, $\text{R}$-products are given in terms of $\text{T}$-products and reverse $\text{T}$-products by
\[
\text{R}(Y;I)=\sum_{Y_1 \sqcup Y_2 =Y} \overline{\text{T}}(Y_1) \star  \text{T}(Y_2\sqcup I)  
.\]
Then
\[
\widetilde{\text{T}}(I)= \sum_{r=0}^\infty  \dfrac{c^r}{r!} \text{R}(r;I)
.\]
In a similar way, we can define the A-\emph{products} $\text{A}(Y;I)$, so that 
\[
\text{A}(Y;I)=\sum_{Y_1 \sqcup Y_2 =Y}  \text{T}(Y_1 \sqcup  I)   \star \overline{\text{T}}(Y_2)
.\]
The total R-products are both R-products and generalized R-products, which is due to the double description appearing in \autoref{rem:double}. A related result is \cite[Proposition 109]{aguiar2013hopf}.

\begin{remark}
In the literature, the total retarded products in our sense are sometimes called retarded products, and the retarded products in our sense are then called generalized retarded products, e.g. \cite{polk58}, \cite[Exercise 3.3.16]{dutsch2019perturbative}. 
\end{remark}

%%%%%%%%%%%%%%%%%%%%%%%%%%%%%%%%%%%%%%%%%%%%%%%%%%%%%%%%%%%%%%%%%%%%%%%%
%%%%%%%%%%%%%%%%%%%%%%%%%%%%%%%%%%%%%%%%%%%%%%%%%%%%%%%%%%%%%%%%%%%%%%%%
%%%%%%%%%%%%%%%%%%%%%%%%%%%%%%%%%%%%%%%%%%%%%%%%%%%%%%%%%%%%%%%%%%%%%%%%
%%%%%%%%%%%%%%%%%%%%%%%%%%%%%%%%%%%%%%%%%%%%%%%%%%%%%%%%%%%%%%%%%%%%%%%%
%%%%%%%%%%%%%%%%%%%%%%%%%%%%%%%%%%%%%%%%%%%%%%%%%%%%%%%%%%%%%%%%%%%%%%%%
\part{Perturbative Algebraic Quantum Field Theory} \label{part 2}
%%%%%%%%%%%%%%%%%%%%%%%%%%%%%%%%%%%%%%%%%%%%%%%%%%%%%%%%%%%%%%%%%%%%%%%%
%%%%%%%%%%%%%%%%%%%%%%%%%%%%%%%%%%%%%%%%%%%%%%%%%%%%%%%%%%%%%%%%%%%%%%%%
%%%%%%%%%%%%%%%%%%%%%%%%%%%%%%%%%%%%%%%%%%%%%%%%%%%%%%%%%%%%%%%%%%%%%%%%
%%%%%%%%%%%%%%%%%%%%%%%%%%%%%%%%%%%%%%%%%%%%%%%%%%%%%%%%%%%%%%%%%%%%%%%%
%%%%%%%%%%%%%%%%%%%%%%%%%%%%%%%%%%%%%%%%%%%%%%%%%%%%%%%%%%%%%%%%%%%%%%%%

We now apply the theory we have developed to the case of a real scalar quantum field on a Minkowski spacetime, as described by pAQFT.\footnote{\ although pAQFT deals more generally with perturbative \hbox{Yang-Mills} gauge theory on curved spacetimes} Mathematically, the important extra property is a causal structure on the vector space of decorations $V$, which allows one to impose causal factorization. Connections between QFT and species have been previously studied in \cite{MR2036353}, \cite{MR2862982}, \cite{MR3753672}.

Our references for pAQFT are \cite{dutfred00}, \cite{rejzner2016pQFT}, \cite{dutsch2019perturbative}, \cite{perturbative_quantum_field_theory}. We mainly adopt the notation and presentation of \cite{perturbative_quantum_field_theory}. Key features of pAQFT are its local, i.e. \hbox{sheaf-theoretic}, approach, the (closely related) use of adiabatic switching of interaction terms to avoid IR-divergences, and the interpretation of renormalization as the extension of distributions to the fat diagonal to avoid UV-divergences. The Wilsonian cutoff, sometimes called heuristic quantum field theory, may be rigorously formulated within pAQFT \cite{dutfred09}, \cite{dut12}, \cite[Section 3.8]{dutsch2019perturbative}, \cite[Section 16]{perturbative_quantum_field_theory}. 

%%%%%%%%%%%%%%%%%%%%%%%%%%%%%%%%%%%%%%%%%%%%%%%%%%%%%%%%%%%%%%%%%%%%%%%%
\section{Spacetime and Field Configurations} 
%%%%%%%%%%%%%%%%%%%%%%%%%%%%%%%%%%%%%%%%%%%%%%%%%%%%%%%%%%%%%%%%%%%%%%%%

Let $\cX\cong \bR^{1,p}$ denote a $(p+1)$-dimensional Minkowski spacetime, for $p\in \bN$. Thus, $\cX$ is a real vector space equipped with a metric tensor which is a symmetric nondegenerate bilinear form $\cX\times \cX\to \bR$ with signature $(1,p)$. The bilinear form gives rise to a volume form on $\cX$, which we denote by $\text{dvol}_\cX\in \Omega^{p+1}(\cX)$. For regions of spacetime $X_1,X_2\subset \cX$, we write 
\[
X_1\! \vee\! \! \wedge X_2
\] 
if one cannot travel from $X_1$ to $X_2$ on a future-directed timelike or lightlike curve. We have the set valued species $\cX^{(-)}$ given by
\[
I\ \mapsto \ \cX^I:= \big \{ \text{functions}\ I\to \cX \big\}   
.\]

%\begin{remark} 
%If $p=0$, or we just take the time components, then the space of configurations \hbox{$\lambda:I\to \cX$} (we do not require injectivity) modulo translations of $\cX$ is the (essentialized) braid arrangement over $I$. This configuration space is a tropical algebraic torus, so we may take the toric compactification with respect to the braid arrangement fan, denoted $\bT \Sigma$. The boundary at infinity consists of would-be limiting configurations where points are separated by infinite times. Compositions may be identified with limiting configurations as follows, e.g. the composition $(12,3)$ is the configuration where $1$ and $2$ coincide, with an infinite time separating them from $3$. Then a natural Hopf monoidal structure\footnote{\ in the sense of \cite[Section 4.3]{aguiar2013hopf}} on $\bT\Sigma$ induces the structure of $\Sig$, see \cite[Introduction]{norledge2019hopf}. 
%\end{remark}

%References for this section are \cite[Sections 3 and 7]{perturbative_quantum_field_theory}, \cite{bar15}. Given a one-dimensional real vector space $\text{Fib}\in \textsf{Vec}$,\footnote{\ pAQFT does deal more generally with perturbative Yang-Mills theory} we denote its dual space by $\text{Fib}^\ast$, and the pairing by
%\[
%\la - , - \ra :\text{Fib}^\ast \otimes \text{Fib}\to \bR
%.\] 

For simplicity, we restrict ourselves to the Klein-Gordan real scalar field on $\cX$. Therefore, let $E\to \cX$ be a smooth real vector bundle over $\cX$ with one-dimensional fibers. An (off-shell) \emph{field configuration} $\Phi$ is a smooth section of the bundle $E\to \cX$,
\[
\Phi:\cX\hookrightarrow E
,\qquad
x\mapsto \Phi(x)
.\]
The space of all field configurations, denoted $\Gamma(E)$, has the structure of a Fr\'echet topological (real) vector space. 

\begin{remark}
We can always pick an isomorphism $(E\to \cX) \cong (\cX\times \bR\to \cX)$, which induces an isomorphism $\Gamma(E)\cong C^\infty(\cX,\bR)$, so that field configurations are modeled as smooth functions $\cX\to \bR$. 
\end{remark}

Let $E^\ast\to \cX$ denote the dual vector bundle of $E$, and let the canonical pairing be denoted by
\[
\la -,- \ra : E^\ast \otimes E \to \bR
.\]
Let a \emph{compactly supported distributional section} $\ga$ be a distribution of field configurations
\[
\ga:\Gamma(E)\to \bR
,\] 
i.e. an element of the topological dual vector space of $\Gamma(E)$, which is modeled as a sequence $(\ga_j)_{j\in \bN}$ of smooth compactly supported sections of the dual bundle $E^\ast\to \cX$,
\[
\ga_j:\cX \hookrightarrow E^\ast, \qquad j\in \bN
,\]
where the modeled distribution is recovered as the following limit of integrals,
\[
\Gamma(E)\to \bR, \qquad \Phi \mapsto \underbrace{\int_{x\in \cX}  \big \la \ga(x), \Phi(x)  \big \ra   \text{dvol}_\cX := \lim_{j\to \infty} \int_{x\in \cX} \la \ga_j(x) , \Phi(x) \ra    \text{dvol}_\cX}_{\text{sometimes called \emph{generalized function} notation}}
.\]
The space of all compactly supported distributional sections is denoted $\Gamma_{\text{cp}}'(E^\ast)$. By e.g. \cite[Lemma 2.15]{bar15}, \emph{all} distributions $\Gamma(E)\to \bR$ may be obtained as compactly supported distributional sections in this way. 

%We formalize this realization of the sequence of smooth sections as distributions (called linear observables) via a (generalized) system of products on $\textbf{X}$ below.

We can pullback the vector bundle $E^\ast$ to $\cX^I$ along each canonical projection 
\[
\cX^I\to \cX^{\{i\}}\cong \cX 
,\qquad 
i\in I
.\] 
The tensor product of these $n$ many pullback bundles is the exterior tensor product bundle $(E^\ast)^{ \boxtimes I }$. This defines a presheaf of smooth vector bundles on $\sfS$,  
\[
\sfS^{\op}\to \textsf{Diff}_{/ \cX}
, \qquad 
I\mapsto (E^\ast)^{ \boxtimes I }
.\]
By taking compexified compactly supported distributional sections ${{\Gamma'}^\bC}_{\! \! \! \! \! \! \text{cp}}(-):=\Gamma_{\text{cp}}'(-)\otimes_\bR \bC$, we obtain the complex vector species ${{\boldsymbol{\Gamma}'}^\bC}_{\! \! \! \! \! \! \text{cp}}(E^\ast)$, given by
\[
{{\boldsymbol{\Gamma}'}^\bC}_{\! \! \! \! \! \! \text{cp}}(E^\ast)[I]:= {{\Gamma'}^\bC}_{\! \! \! \! \! \! \text{cp}}\Big ((E^\ast)^{ \boxtimes I }\Big)
.\]
Of course, ${{\boldsymbol{\Gamma}'}^\bC}_{\! \! \! \! \! \! \text{cp}}(E^\ast)$ does not `factorize' in the sense that it is not a monoidal functor,
\begin{equation} \label{fact}
{{\boldsymbol{\Gamma}'}^\bC}_{\! \! \! \! \! \! \text{cp}}(E^\ast)[I] \ncong {{\boldsymbol{\Gamma}'}^\bC}_{\! \! \! \! \! \! \text{cp}}(E^\ast)[i_1]\otimes \dots \otimes {{\boldsymbol{\Gamma}'}^\bC}_{\! \! \! \! \! \! \text{cp}}(E^\ast)[i_n]  
\end{equation}
where $I=\{ i_1,\dots,i_n \}$. There are more distributional sections then just those coming from the tensor product.

%We formalize a realization of these distributional sections as observables (called polynomial observables) via a (generalized) system of products on $\textbf{E}$ below.

%%%%%%%%%%%%%%%%%%%%%%%%%%%%%%%%%%%%%%%%%%%%%%%%%%%%%%%%%%%%%%%%%%%%%%%%
\section{Observables} \label{sec:obs}
%%%%%%%%%%%%%%%%%%%%%%%%%%%%%%%%%%%%%%%%%%%%%%%%%%%%%%%%%%%%%%%%%%%%%%%%

%References for this section are \cite[Section 7]{perturbative_quantum_field_theory}, \cite{bar15}, \cite[Section 1.2-1.3]{dutsch2019perturbative}. See also \cite{horm90} for theory of distributions. 

An off-shell \emph{observable} $\emph{\textsf{O}}$ is a smooth functional of field configurations into the complex numbers,
\[     
\emph{\textsf{O}}:\Gamma(E)\to \bC
,\qquad 
\Phi\mapsto \emph{\textsf{O}}(\Phi)  
.\]
The space of all observables is denoted $\text{Obs}$. We can pointwise multiply observables, sometimes called the \emph{normal ordered product}, so that observables form a commutative $\bC$-algebra,
\[ 
\text{Obs}\otimes \text{Obs} \to \text{Obs}, \qquad \emph{\textsf{O}}_1\otimes \emph{\textsf{O}}_2 \mapsto  \emph{\textsf{O}}_1 \cdot \emph{\textsf{O}}_2         
\]
where
\[      
\emph{\textsf{O}}_1 \cdot \emph{\textsf{O}}_2(\Phi):=  \! \! \! \underbrace{\emph{\textsf{O}}_1(\Phi)\emph{\textsf{O}}_2(\Phi)}_{\text{multiplication in $\bC$}} \! \!  \!  
.\]
Thus, we may form the commutative algebra in species $\textbf{U}_{\text{Obs}}$, given by $\textbf{U}_{\text{Obs}}[I]=\text{Obs}$. 

A \emph{linear observable} $\emph{\textsf{O}}\in \text{Obs}$ is an observable which is additionally a linear functional, that is
\[
\emph{\textsf{O}}( \Phi_1 + \Phi_2 ) = \emph{\textsf{O}}(\Phi_1) + \emph{\textsf{O}}(\Phi_2)
\qquad \text{and} \qquad
\emph{\textsf{O}}(c \Phi) = c \emph{\textsf{O}}(\Phi) \qquad \text{for}\quad  c\in \bC
.\] 
The space of linear observables is denoted $\text{LinObs}\subset \text{Obs}$. In particular, for each spacetime event $x\in \cX$, we have the \emph{field observable} $\boldsymbol{\Phi}(x) \in \text{LinObs}$, given by
\[    
\boldsymbol{\Phi}(x):    \Gamma(E)\to \bC,  \qquad  \Phi\mapsto \Phi(x)
.\] 
We now show how linear observables and so-called polynomial observables arise species-theoretically, via (generalized) systems of products for the species $\textbf{E}$ and $\textbf{X}=\mathcal{P}(\textbf{E})$. 

Let $\textbf{X}$ denote the species given by $\textbf{X}\big [\{i\}\big]:=\bC$ for singletons and $\textbf{X}[I]:=0$ otherwise. We denote $\tH_i:=1\in \textbf{X}\big [\{i\}\big]$. We have the following morphism of species,
\[
\textbf{X} \otimes {{\boldsymbol{\Gamma}'}^\bC}_{\! \! \! \! \! \! \text{cp}}(E^\ast) \to \textbf{U}_{\text{Obs}},
\qquad
\tH_i \otimes \ga \mapsto 
\bigg(  \Phi \mapsto   \int_{x\in \cX} \big \la \ga(x), \Phi(x) \big \ra     \text{dvol}_\cX  \bigg)
.\]
This is like a system of products for $\textbf{X}$, however ${{\boldsymbol{\Gamma}'}^\bC}_{\! \! \! \! \! \! \text{cp}}(E^\ast)$ does not factorize \textcolor{blue}{(\refeqq{fact})}, and so cannot be written in the form $\textbf{E}_V$. It follows from \cite[Lemma 2.15]{bar15} that the colimit (as defined in \cite[Remark 15.7]{aguiar2010monoidal}) of the species which is the image of this morphism is the space of linear observables $\text{LinObs}$. The currying of this map is given by
\[
\textbf{X}\to \cH \big ( {{\boldsymbol{\Gamma}'}^\bC}_{\! \! \! \! \! \! \text{cp}}(E^\ast) ,  \textbf{U}_{\text{Obs}}\big  ),
\qquad 
\tH_i\mapsto \boldsymbol{\Phi}_i=\boldsymbol{\Phi}    \]
where
\[
\boldsymbol{\Phi}(\ga):=\bigg( \Phi \mapsto  \int_{x\in \cX}\big \la \ga(x), \Phi(x) \big \ra     \text{dvol}_\cX   \bigg)
.\]
If we restrict $\boldsymbol{\Phi}$ to bump functions $b\in \Gamma_{\text{cp}}(E^\ast)\otimes_{\bR} \bC$, also called `smearing functions', then one might call the linear map
\[
\boldsymbol{\Phi}: \Gamma_{\text{cp}}(E^\ast)\otimes_{\bR} \bC \to   \text{Obs}
,\qquad
  b\mapsto \boldsymbol{\Phi}(b)
\] 
an `observable-valued distribution', and this is sometimes referred to as `the (smeared) field'. The field observable $\boldsymbol{\Phi}(x)$ is recovered by evaluating $\boldsymbol{\Phi}$ on the Dirac delta function $\delta_x$ localized at $x$. One views $b$ as the smearing of a Dirac delta function, hence smearing functions and smeared field. 
%\cite[Lemma 2.15]{bar15}, Compactly supported distributional sections account for all linear observables $\Gamma(E)\to \bC$. 
%Note that smooth and compactly supported functions $b\in \cX \to \bC$, \emph{bump functions} or `smearing functions', are particular cases. 
%By integrating field configurations against compactly supported distributional sections, we obtain an isomorphism    
%\[
%\text{LinObs} \cong \Gamma'_{\text{cp}}(E^\ast)\cong  {{\boldsymbol{\Gamma}'}^\bC}_{\! \! \! \! \! \! \text{cp}}(E^\ast) [i]
%.\] 
%See \cite[Proposition 3.6]{Moerdijk91}. 
%\begin{figure}[t]
%	\centering
%	\includegraphics[scale=0.42]{prop}
%	\caption{The various propagators for a real scalar field (see e.g. \cite[Section A.2]{dutsch2019perturbative}) in the case $p=0$ are $\bC$-valued generalized functions on the one-dimensional braid arrangement $\bR^2/(1,1)$. The imaginary part is shown in blue.}
%	\label{fig:prop}
%\end{figure}

We extend the smeared field by replacing $\textbf{X}$ with $\textbf{E}$ to define the following morphism of species,
\[
\textbf{E} \otimes {{\boldsymbol{\Gamma}'}^\bC}_{\! \! \! \! \! \! \text{cp}}(E^\ast) \to \textbf{U}_{\text{Obs}}
,\qquad
\tH_I\otimes \ga_I 
\mapsto
\bigg(  \Phi \mapsto   \int_{\cX^I}
\big \la 
\ga_I(x_{i_1}, \dots, x_{i_n}),
\Phi(x_{i_1})\dots \Phi(x_{i_n}) 
\big \ra     \text{dvol}_{\cX^I}  \bigg)  
.\]
This is like a system of products for $\textbf{E}$, but again without factorization. The colimit of the species which is the image of this morphism is the vector space of \emph{polynomial observables}, as defined in e.g. \cite[Definition 7.13]{perturbative_quantum_field_theory}, denoted
\[
\text{PolyObs}\subset \text{Obs}
.\]
(Alternatively, if we restrict the limit of this map $\mathscr{S}({{\boldsymbol{\Gamma}'}^\bC}_{\! \! \! \! \! \! \text{cp}}(E^\ast)) \to \text{Obs}[[\formj]]$ to finite series and set $\formj=1$, then we recover \cite[Definition 1.2.1]{dutsch2019perturbative}.) The space of \emph{microcausal polynomial observables} $\mathcal{F}$ is the subspace 
\[
\mathcal{F} \subset \text{PolyObs}
\] 
consisting of those polynomial observables which satisfy a certain microlocal-theoretic condition called \emph{microcausality}, see \cite[Definition 1.2.1 (ii)]{dutsch2019perturbative}. Following \cite[Definition 1.3.4]{dutsch2019perturbative}, the space of \emph{local observables}
\[
\mathcal{F}_{\text{loc}}\subset \text{Obs}
\] 
consists of those observables obtained by integrating a polynomial with real coefficients in the field and its derivatives (`field polynomials') against a bump function $b\in \Gamma_{\text{cp}}(E^\ast)\otimes_\bR \bC$. Importantly, we have a natural inclusion
\[
\mathcal{F}_{\text{loc}} \hookrightarrow \mathcal{F}
,\qquad
\ssA \mapsto\  :\ssA: 
.\] 
Let $\mathcal{F}_{\text{loc}}[[\hbar]]$ and $\mathcal{F}[[\hbar]]$ denote the spaces of formal power series in $\hbar$ with coefficients in $\mathcal{F}_{\text{loc}}$ and $\mathcal{F}$ respectively, and let $\mathcal{F}((\hbar))$ denote the space of Laurent series in $\hbar$ with coefficients in $\mathcal{F}$.

Applying Moyal deformation quantization with formal Planck's constant $\hbar$, $\mathcal{F}[[\hbar]]$ is a formal power series $\ast$-algebra, called the (abstract, off-shell) \emph{Wick algebra}, with multiplication the Moyal star product \cite[Definition 2.1.1]{dutsch2019perturbative} defined with respect to the Wightman propagator $\Delta_{\text{H}}$ for the Klein-Gordan field \cite[Section 2.2]{dutsch2019perturbative}, 
\[
\mathcal{F}[[\hbar]] \otimes \mathcal{F}[[\hbar]] \to \mathcal{F}[[\hbar]]
,\qquad
\emph{\textsf{O}}_1 \otimes \emph{\textsf{O}}_2\\
\mapsto
\emph{\textsf{O}}_1\star_{\text{H}}\emph{\textsf{O}}_2
.\] 
We may form the algebra in species $\textbf{U}_{\mathcal{F}[[\hbar]]}$, or, allowing negative powers of $\hbar$, $\textbf{U}_{\mathcal{F}((\hbar))}$. 

%Briefly, the wave front set of each $\ga_I$ must not include those points where all $n$ wave vectors are in the closed future cone or all in the closed past cone. 

%%%%%%%%%%%%%%%%%%%%%%%%%%%%%%%%%%%%%%%%%%%%%%%%%%%%%%%%%%%%%%%%%%%%%%%%
\section{Time-Ordered Products and S-Matrix Schemes} \label{sec:Time-Ordered Products}
%%%%%%%%%%%%%%%%%%%%%%%%%%%%%%%%%%%%%%%%%%%%%%%%%%%%%%%%%%%%%%%%%%%%%%%%

%The \emph{retarded propagator} $\Delta^{ \text{ret} }$ is the fundamental solution,
%\[    (   \square +m^2  ) \Delta^{ \text{ret} }=-\delta  , \qquad \text{with}\quad   \text{supp} %(\Delta^{ \text{ret} })\subseteq \overline{V}_+.  \]
%This is uniquely determined.

For $\ssA\in \mathcal{F}_{\text{loc}}[[\hbar]]$, let $\text{supp}(\ssA)$ denote the spacetime support of $\ssA$. Given a composition $G$ of $I$, we say that $\ssA_I\in \textbf{E}_{\mathcal{F}_{\text{loc}}[[\hbar]]}[I]$ \emph{respects} $G$ if
\[
\text{supp}({\ssA_{i_1}})\vee\! \! \wedge\ \text{supp}(\ssA_{i_2}) 
\qquad \quad \text{for all}\quad 
(i_1,i_2)\quad \text{such that}\quad  G|_{\{i_1, i_2\}}= (i_1,i_2) 
.\footnote{\ $G|_{\{i_1, i_2\}}= (i_1,i_2)$ means that $i_1$ and $i_2$ are in different lumps, with the lump containing $i_1$ appearing to the left of the lump containing $i_2$}\]  
Consider a system of $\text{T}$-products (as defined in \autoref{sec:T-Products, Generalized T-Products, and Generalized R-Products}) of the form
\[
\text{T}: \textbf{E}^\ast_+ \otimes \textbf{E}_{\mathcal{F}_{\text{loc}}[[\hbar]]}\to 
\textbf{U}_{\mathcal{F}((\hbar))},
\qquad
\tH_{(I)}\otimes \ssA_I\mapsto \text{T}_I(\tH_{(I)}\otimes \ssA_I)=\text{T}_I(\ssA_I)
.\]

%\begin{figure}[t]
%	\centering
%	\includegraphics[scale=0.45]{respects}
%	\caption{For $p=0$, configurations $\lambda\in \cX^I/\cX=\bR^I/(1,\dots,1)$ shown in red which respect the composition $G$ shown in white (i.e. \hbox{$\lambda(i_1)\vee\! \! \wedge\ \lambda(i_2)$} for all $(i_1,i_2)$ such that $G|_{\{i_1, i_2\}}=(i_1,i_2)$), depicted on the tropical toric compactification of the tropical torus $\bR^I/(1,\dots,1)=(\bT^\times)^I/\bT^\times$ with respect to the braid arrangement fan.}
%	\label{fig:respects}
%\end{figure}

%When interactions are adiabatically switched on, and we take the adiabatic limit, the region of factorization of scattering amplitudes (\autoref{prop:scattering}) will tend towards the boundary facets, where configurations are separated by infinite times

\noindent Since $\Sig$ is the free algebra on $\textbf{E}^\ast_+$, we have the unique extension to a system of generalized \hbox{$\text{T}$-products}
\[
\text{T}: \Sig \otimes \textbf{E}_{\mathcal{F}_{\text{loc}}[[\hbar]]}\to 
\textbf{U}_{\mathcal{F}((\hbar))},
\qquad
\text{T}_I(\tH_{F}\otimes \ssA_I):=\text{T}_{S_1}(\ssA_{S_1})\star_{\text{H}} \dots \star_{\text{H}}  \text{T}_{S_k}(\ssA_{S_k}) 
.\]
Then:
\begin{enumerate}[label={\arabic*.}]
\item 
(perturbation) we say that $\text{T}$ satisfies \emph{perturbation} if the singleton components $\text{T}_{i}$ are isomorphic to the inclusion $\mathcal{F}_{\text{loc}}[[\hbar]]\hookrightarrow \mathcal{F}((\hbar))$, that is
\[   
{\text{T}}_{i}( \ssA  )=  \,   :\! \ssA \! :       
\]
\item
(causal factorization) we say that $\text{T}$ satisfies \emph{causal factorization} if for all compositions $(S,T)$ of $I$ with two lumps, if $\ssA_I \in \textbf{E}_{\mathcal{F}_{\text{loc}}[[\hbar]]}[I]$ respects $(S,T)$\footnote{\ explicitly, $\text{supp}({\ssA_{i_1}})\vee\! \! \wedge\ \text{supp}(\ssA_{i_2})$ for all $i_1\in S$ and $i_2\in T$} then
\begin{equation}\label{eq:causalfac}
\text{T}_I(\tH_{(I)}\otimes\ssA_I)
=\text{T}_I(\tH_{(S,T)}\otimes\ssA_I).\footnote{\ or equivalently $\text{T}_I(\ssA_I)=\text{T}_S(\ssA_S)\star_{\text{H}} \text{T}_T(\ssA_T)$} 
\end{equation}
%(Recall that $\text{T}(\tH_{(S,T)}\otimes\ssA_I)=\text{T}(\tH_{(S)}\otimes \ssA_S)\star_{\text{H}} \text{T}(\tH_{(T)}\otimes\ssA_T)$.)
\end{enumerate}
Let a (fully normalized) \emph{system of time-ordered products} be a system of $\text{T}$-products which satisfies perturbation and causal factorization. The corresponding unique extension of $\text{T}$ to $\Sig$ is called the associated \emph{system of generalized time-ordered products}. After currying
\[     
\Sig \to \cH(  \textbf{E}_{\mathcal{F}_{\text{loc}}[[\hbar]]} ,  \textbf{U}_{\mathcal{F}((\hbar))} ), 
\qquad     
\tH_{F}\mapsto  \text{T}(S_1)\dots \text{T}(S_k)   
\]
the linear maps
\[
\text{T}(S_1)\dots \text{T}(S_k):  \mathcal{F}_{\text{loc}}[[\hbar]]^{\otimes I}  
\to
\mathcal{F}((\hbar))
,\qquad
\ssA_I\mapsto \text{T}_I(\tH_{F}\otimes \ssA_I)
\]
are called \emph{generalized time-ordered products}. The linear maps $\text{T}(I)$ are called \emph{time-ordered products}. After fixing a field polynomial, so that each $\ssA_{i_j}$ of $\ssA_I$ is determined by a bump function $b_{i_j}$, they are usually presented in generalized function notation as follows,
\[
\text{T}_I (\ssA_{i_1}  \otimes \cdots \otimes \ssA_{i_n} )
= 
\int_{\cX^I}   \text{T}(x_{i_1}, \dots , x_{i_n})  b_{i_1}( x_{i_1} ) \dots  b_{i_n}(x_{i_1}) 
dx_{i_1} \dots dx_{i_n}
\]
where $(x_{i_1}, \dots , x_{i_n}) \mapsto \text{T}(x_{i_1}, \dots , x_{i_n})$ is an `operator-valued' generalized function. See e.g. \cite[Section 1.2]{ep73roleofloc}.

Given compositions $F=(S_1, \dots , S_{k})$ and $G=(U_1,\dots, U_{l})$ of $I$, let
\[  
\tH_{F}\triangleright \tH_{G}  : =\tH_{(  U_1\cap S_1, \dots, U_{l}\cap S_1, \dots \dots, U_{1}\cap S_{k}, \dots ,   U_{l}\cap S_{k}        )_+}  
.\]
This is called the \emph{Tits product}, going back to Tits \cite{Tits74}. See \cite[Section 13]{aguiar2013hopf} for more on the structure of the Tits product, where it is shown it is given by the action of $\Sig$ on itself by Hopf powers. See also \cite[Section 1.4.6]{brown08} for the context of other Coxeter systems and Dynkin types.

%Causal factorization is a condition on $\text{T}(I)$ in terms of $\text{T}(S)$ and $\text{T}(T)$. In this way, time-ordered products may be constructed inductively, as in \cite{ep73roleofloc}. %This result appears to be part of the physicists motivation for the Tits product. %This is an interesting question, since the species-theoretic Hopf structure of $\Sig$ does not appear explicitly in physics.

\begin{prop}\label{prob:causalfac}
Let 
\[
\text{T}: \textbf{E}^\ast_+ \otimes \textbf{E}_{\mathcal{F}_{\text{loc}}[[\hbar]]}\to 
\textbf{U}_{\mathcal{F}((\hbar))}
\] 
be a system of $\text{T}$-products which satisfies causal factorization. Given a composition \hbox{$G=(U_1,\dots,U_k)$} of $I$, and $\ssA_I \in \textbf{E}_{\mathcal{F}_{\text{loc}}[[\hbar]]}[I]$ which respects $G$, then 
\[ 
\text{T}_I( \mathtt{a} \otimes\ssA_I)
= 
\text{T}_I( \mathtt{a} \triangleright \tH_{G} \otimes\ssA_I)
\qquad  \text{for all} \quad 
\mathtt{a}\in \Sig[I]     
.\]
\end{prop}
\begin{proof}
We have
\[
\text{T}_I(\tH_{G}\otimes\ssA_I)
=
\underbrace{\text{T}_{U_1}(\ssA_{U_1})\star_{\text{H}} \dots \star_{\text{H}} \text{T}_{U_k}(\ssA_{U_k})
=
\text{T}_I(\ssA_I)}_{\text{by repeated applications of causal factorization}}
.\]
Observe that the action $\tH_F \mapsto \tH_F \triangleright \tH_{G}$, for $F\in \Sigma[I]$, replaces the lumps of $F$ with their intersections with $G$. But we just saw that $\text{T}_I(\ssA_I)=\text{T}_I(\tH_{G}\otimes\ssA_I)$, and so it follows that
\[
\text{T}_I( \tH_F \otimes\ssA_I)
= 
\text{T}_I( \tH_F \triangleright \tH_{G} \otimes\ssA_I)
.\]
Since the claim is true for the $\tH$-basis, it is true for all $\mathtt{a}\in \Sig[I]$. 
\end{proof}

\begin{cor}\label{cor:cor}
If $\mathtt{a} \triangleright \tH_{G} =0$, then 
\[
\text{T}_I(\mathtt{a} \otimes\ssA_I)=0
\]
for all $\ssA_I\in \textbf{E}_{\mathcal{F}_{\text{loc}}[[\hbar]]}[I]$ which respect $G$. 
\end{cor}

The restriction of $\text{T}$ to the primitive part Lie algebra is called the associated \emph{system of generalized retarded products},
\[
\text{R}:\Zie\otimes\textbf{E}_{\mathcal{F}_{\text{loc}}[[\hbar]]} \to \textbf{U}_{\mathcal{F}((\hbar))}  
.\]
The image of the Dynkin elements $\mathtt{D}_\cS$ under the currying of $\text{R}$ are the \emph{generalized retarded products} $\text{R}_\cS$, see e.g. \cite[Equation 79]{ep73roleofloc}. It follows from \autoref{cor:cor} and the structure of Dynkin elements under the Tits product that generalized retarded products have nice support properties. This is described in \cite{epstein1976general}.
 
Given a system of generalized time-ordered products
\[
\text{T}: \Sig \otimes \textbf{E}_{\mathcal{F}_{\text{loc}}[[\hbar]]}\to 
\textbf{U}_{\mathcal{F}((\hbar))}
\]
the $\text{T}$-exponential $\mathcal{S}=\mathcal{S}_{\mathtt{G}(1/\text{i}\hbar)}$ (defined in \textcolor{blue}{(\refeqq{eq:tsexp})}) for the group-like series 
\[
\mathtt{G}(1/\text{i}\hbar): \textbf{E} \to \Sig
,\qquad
\tH_{I} \mapsto   \dfrac{1}{\text{i}\hbar}\tH_{(I)}
\]
is called the associated perturbative \emph{S-matrix scheme}. Thus, $\mathcal{S}$ is the function 
\[
\mathcal{S}: \mathcal{F}_{\text{loc}}[[\hbar]] \to \mathcal{F}((\hbar))[[\formj]] 
,\qquad        
\ssA\mapsto  \mathcal{S}(\formj\! \ssA)  
:=
\sum^\infty_{n=0} \bigg(\dfrac{1}{\text{i}\hbar}\bigg)^{\! n}\dfrac{\formj^n}{n!} \text{T}_n(\ssA^n)     
.\]

%\footnote{\ here we let $\Sig$ be defined over $\bC=\bC((\hbar))$, i.e. the field of restricted Laurent series in $\hbar$ with coefficients in $\bC$} 

%Given a choice of vector $\ssS_{\text{int}}\in \mathcal{F}_{\text{loc}}[[\hbar]]$, now called the adiabatically switched \emph{interaction}, recall from \autoref{sec:Perturbation of T-Products by Up Coderivation} that we also have the $\text{T}$-exponential perturbed by the up coderivation of $\textbf{E}^\ast$,
%\[
%\mathcal{S}_{\formg \ssS_{\text{int}}}: \mathcal{F}_{\text{loc}}[[\hbar]] \to \mathcal{F}((\hbar))[[\formg,\! \formj]] 
%,\qquad        
%\ssA\mapsto  \mathcal{S}_{\formg \ssS_{\text{int}}}(\formj\! \ssA)  
%=
%\mathcal{S}( \formg \ssS_{\text{int}}+ \formj \! \ssA)     
%.\]
%As in \cite[Section 15]{perturbative_quantum_field_theory}, we may extend $\bC[[\hbar,\formg,\! \formj]]$-linearly to obtain a map
%\[
%\mathcal{F}_{\text{loc}}[[\hbar, \formg,\! \formj]]\la \formg,\! \formj \ra  \to \mathcal{F}((\hbar))[[\formg,\! \formj]] 
%\]
%where $\mathcal{F}_{\text{loc}}[[\hbar, \formg,\! \formj]]\la \formg,\! \formj \ra$ denotes formal power series which are at least linear in $\formg$ or $\! \formj$. 

%%%%%%%%%%%%%%%%%%%%%%%%%%%%%%%%%%%%%%%%%%%%%%%%%%%%%%%%%%%%%%%%%%%%%%%%
\section{Interactions}    \label{sec:Interactions}
%%%%%%%%%%%%%%%%%%%%%%%%%%%%%%%%%%%%%%%%%%%%%%%%%%%%%%%%%%%%%%%%%%%%%%%%

Given a choice of adiabatically switched \emph{interaction} $\ssS_{\text{int}} \in \mathcal{F}_{\text{loc}}[[\hbar]]$, and a system of fully normalized generalized \hbox{time-ordered} products
\[
\text{T}: \Sig \otimes \textbf{E}_{\mathcal{F}_{\text{loc}}[[\hbar]]}\to 
\textbf{U}_{\mathcal{F}((\hbar))}
,\]
we have the new system of interacting generalized time-ordered products which is obtained by the construction of \autoref{sec:Perturbation of T-Products by Steinmann Arrows},
\[ 
\wt{\text{T}}
:
\Sig\otimes\textbf{E}_{\mathcal{F}_{\text{loc}}[[\hbar]]} 
\to 
\textbf{U}_{\mathcal{F}((\hbar))[[\formg]]}
.\]
The associated \emph{generating function scheme} $\mathcal{Z}_{\formg \ssS_{\text{int}}}$ for interacting field observables, and more generally for time-ordered products of interacting field observables, is the new $\text{T}$-exponential for the group-like series $\mathtt{G}(1/\text{i}\hbar)$, denoted $\mathcal{V}_{\formg \ssS_{\text{int}}}$ in \autoref{sec:Perturbation of T-Products by Steinmann Arrows}. Thus, $\mathcal{Z}_{\formg \ssS_{\text{int}}}$ is the function
\[
\mathcal{Z}_{\formg \ssS_{\text{int}}}
: 
\mathcal{F}_{\text{loc}}[[\hbar]] \to \mathcal{F}((\hbar))[[\formg,\! \formj]] 
,\qquad
\ssA \mapsto \mathcal{Z}_{\formg \ssS_{\text{int}}}(\formj\! \ssA)
\]
where
\[
\mathcal{Z}_{\formg \ssS_{\text{int}}}(\formj\! \ssA)
\! :=\!
\sum_{n=0}^\infty \!  \bigg(\dfrac{1}{\text{i}\hbar}\bigg)^{\! \!  n}\! \dfrac{\formj^n}{n!} \wt{\text{T}}_n(\ssA_n)
=\! 
\sum_{n=0}^\infty \sum_{r=0}^\infty \!  
\bigg(\dfrac{1}{\text{i}\hbar}\bigg)^{\! \!  r+n}\! 
\dfrac{\formg^{r} \formj^n}{r!\, n!}\, \!   \text{R}_{r;n} (\ssS_{\text{int}}^{\, r} ; \ssA^n)    
=
\mathcal{S}^{-1}( \formg \ssS_{\text{int}})\star_{\text{H}} \mathcal{S}(\formg \ssS_{\text{int}} +\formj\! \ssA )
.\]
Then 
\[
\ssA_{\text{int}}:
=
\wt{\text{T}}_i(\ssA)
=
\sum_{r=0}^\infty 
\bigg(\dfrac{1}{\text{i}\hbar}\bigg)^{\! r} 
\dfrac{\formg^{r}}{r!} \text{R}_{r+1} (\ssS^{\, r}_{\text{int}}; \ssA)  
\in \mathcal{F}((\hbar))[[\formg]]
\] 
is the \emph{local interacting field observable} of $\ssA$. Bogoliubov's formula \textcolor{blue}{(\refeqq{eq:Bog})} now reads
\[
\ssA_{\text{int}}
=
\text{i} \hbar \, \dfrac{d}{d\formj}\Bigr|_{\formj=0}  \mathcal{Z}_{\formg \ssS_{\text{int}}}(\formj\! \ssA)
.\]
One views $\ssA_{\text{int}}$ as the deformation of the local observable $\ssA$ due to the interaction $\ssS_{\text{int}}$. One can show that $\wt{\text{T}}$ does indeed land in $\textbf{U}_{\mathcal{F}[[\hbar,\formg]]}$ \cite[Proposition 2 (ii)]{dutfred00}. The perturbative interacting quantum field theory then has a classical limit \cite{collini2016fedosov}, \cite{MR4109798}. 

%This means that interacting field observables have a classical limit, being a formal deformation quantization of the classical interacting field theory.
%(, Hawkins-Rejzner 16, cor. 5.2). 

% and the map $\ssA\mapsto \ssA_{\text{int}}$ is somtimes called a quantum M{\o}ller operator. 

%%%%%%%%%%%%%%%%%%%%%%%%%%%%%%%%%%%%%%%%%%%%%%%%%%%%%%%%%%%%%%%%%%%%%%%%
\section{Scattering Amplitudes}\label{sec:scatterung}
%%%%%%%%%%%%%%%%%%%%%%%%%%%%%%%%%%%%%%%%%%%%%%%%%%%%%%%%%%%%%%%%%%%%%%%%

We finish with a translation of a standard result in pAQFT (see \cite[Example 15.12]{perturbative_quantum_field_theory}) into our notation, which relates S-matrix schemes as presented in \autoref{sec:Time-Ordered Products} to S-matrices used to compute scattering amplitudes, which are predictions of pAQFT that are tested with scattering experiments at particle accelerators. 

Following \cite[Definition 2.5.2]{dutsch2019perturbative}, the \emph{Hadamard vacuum state} $\la - \ra_0$ is the linear map given by
\[ 
\la - \ra_0
: 
\mathcal{F}[[\hbar,\formg]] \to \bC[[\hbar,\formg]]
,\qquad
\emph{\textsf{O}}\mapsto  \la  \emph{\textsf{O}}\,  \ra_0:= \emph{\textsf{O}}\, (\Phi=0)
.\]
Let $\ssS_{\text{int}} \in \mathcal{F}_{\text{loc}}[[\hbar]]$. We say that the Hadamard vacuum state $\la - \ra_0$ is \emph{stable} with respect to the interaction $\ssS_{\text{int}}$ if for all $\emph{\textsf{O}}\in \mathcal{F}[[\hbar,\formg]]$, we have
\begin{equation}\label{eq:vacstab}
\big\la\emph{\textsf{O}}\star_{\text{H}}\mathcal{S}(\formg \ssS_{\text{int}}) \big \ra_0
=
\big \la\emph{\textsf{O}}\, \big\ra_0
\big\la    \mathcal{S}(\formg \ssS_{\text{int}})   \big \ra_0    
\qquad \text{and} \qquad
\big\la    
\mathcal{S}^{-1}(\formg \ssS_{\text{int}})\star_{\text{H}} \emph{\textsf{O}}  \,  
\big\ra_0
=
\dfrac{1}{\big \la \mathcal{S}(\formg \ssS_{\text{int}})\big \ra_0 } \big \la \emph{\textsf{O}}\, \big \ra_0
.
\end{equation}
In situations where 
\[
\ssS_{\text{int}} \otimes \ssA_I\in \textbf{E}'_{\mathcal{F}_{\text{loc}}[[\hbar]]}[I]
\qquad \text{respects the composition} \qquad 
(S,\ast,T)
\] 
%and $
%\text{supp}(\ssA_{i_1}) >< \text{supp}(\ssA_{i_2})
%$ 
%if $\{i_1,i_2\}\subseteq S$ or $\{i_1,i_2\}\subseteq T$, 
we can interpret free particles/wave packets labeled by $T$ coming in from the far past, interacting in a compact region according to the adiabatically switched interaction $\ssS_{\text{int}}$, and then emerging into the far future, labeled by $S$. For $\ssA_I\in \textbf{E}_{\mathcal{F}_{\text{loc}}[[\hbar]]}[I]$, let
\[\text{G}_I(\ssA_I):=\big\la 
\widetilde{\text{T}}(\ssA_I)  
\big\ra_0
.\]
If we fix the field polynomial of local observables to be $\text{P}(\Phi)=\Phi$, then $\ssA_I\mapsto \text{G}_I(\ssA_I)$ is the \hbox{time-ordered} $n$-point correlation function, or Green's function. They are usually presented in generalized function notation as follows,
\[
\text{G}_I (b_{i_1}  \otimes \cdots \otimes b_{i_n} )
= 
\int_{\cX^I} \Big\la   \text{T} \big ( \boldsymbol{\Phi}(x_{i_1}) \dots \boldsymbol{\Phi}(x_{i_n})\big ) \Big  \ra_0   b_{i_1}( x_{i_1} ) \dots  b_{i_n}(x_{i_n}) 
dx_{i_1} \dots dx_{i_n}
.\]
Note that to obtain the realistic Green's functions, we still have to take the adiabatic limit. 

%Note that $ \ssS_{\text{int}}$ is adiabatically switched, and so one should now take the algebraic adiabatic limit. 

\begin{prop} \label{prop:scattering}
If the Hadamard vacuum state $\la-\ra_0$ is stable with respect to $\ssS_{\text{int}} \in \mathcal{F}_{\text{loc}}[[\hbar]]$, and if $\ssS_{\text{int}} \otimes \ssA_I\in \textbf{E}'_{ \mathcal{F}_{\text{loc}[[\hbar]]}}[I]$ respects the composition $(S, \ast, T)$, then
\[  
\text{G}_I(\ssA_I)=  
\dfrac{1}
{\Big\la
\mathcal{S}(\formg\ssS_{\text{int}}) 
\Big\ra_0 }   
\Big\la \text{T}_S(\ssA_S)\star_{\text{H}}  \mathcal{S}(\formg\ssS_{\text{int}})\star_{\text{H}}  
\text{T}_T(\ssA_T)
\Big\ra_0
.\footnote{\ the element $\mathcal{S}(\formg\ssS_{\text{int}})\in \mathcal{F}((\hbar))[[\formg]]$ is called the perturbative S\emph{-matrix}}\]
\end{prop}
\begin{proof}
We have
\begin{align*}
\text{G}_I(\ssA_I) 
&=
\big\la 
\widetilde{\text{T}}(\ssA_I)  
\big\ra_0 \\[6pt] 
&= 
\bigg\la 
\sum_{r=0}^\infty \dfrac{\formg^{r}}{r!} 
\text{R}_{r;I}( \ssS^{\, r}_{\text{int}} ; \ssA_{I} )
\bigg\ra_0  \\[6pt] 
&= 
\bigg\la 
\sum_{r=0}^\infty \sum_{r_1 + r_2 =r}   \dfrac{\formg^{r}}{r_1! \, r_2!} 
\overline{\text{T}}_{[r_1] \sqcup \emptyset}(\ssS^{\, r_1}_{\text{int}}) \star_{\text{H}}  \text{T}_{[r_2] \sqcup I}(\ssS^{\, r_2}_{\text{int}} \ssA_{I}) 
\bigg\ra_0.
\end{align*}
To obtain the final line, we expanded the retarded products according to \textcolor{blue}{(\refeqq{eq:retardprod})}. Then, by causal factorization \textcolor{blue}{(\refeqq{eq:causalfac})}, we have
\[
\text{T}_{[r_2]\sqcup I}(\ssS^{\, r_2}_{\text{int}} \ssA_{I})
=
\text{T}_{S}(\ssA_{S})
\star_{\text{H}}
\text{T}_{[r_2]\sqcup  \emptyset}(\ssS^{\, r_2}_{\text{int}})
\star_{\text{H}}
\text{T}_{T}(\ssA_{T})
.\]
Therefore
\begin{align*}
\text{G}_I(\ssA_I) 
&=
\bigg\la 
\sum_{r=0}^\infty \sum_{r_1 + r_2 =r}   \dfrac{\formg^{r}}{r_1! \, r_2!} 
\overline{\text{T}}_{[r_1]\sqcup \emptyset}(\ssS^{\, r_1}_{\text{int}}) \star_{\text{H}} 
\text{T}_{S}(\ssA_{S})
\star_{\text{H}}
\text{T}_{[r_2] \sqcup \emptyset}(\ssS^{\, r_2}_{\text{int}})
\star_{\text{H}}
\text{T}_{T}(\ssA_{T})
\bigg\ra_0. \\[6pt]
&=
\bigg\la 
\sum_{r=0}^\infty \dfrac{\formg^{r}}{r!} 
\overline{\text{T}}_{[r]\sqcup \emptyset}(\ssS^{\, r}_{\text{int}}) \star_{\text{H}} 
\text{T}_{S}(\ssA_{S})
\star_{\text{H}}
\sum_{r=0}^\infty \dfrac{\formg^{r}}{r!} 
\text{T}_{[r]\sqcup \emptyset}(\ssS^{\, r}_{\text{int}})
\star_{\text{H}}
\text{T}_{T}(\ssA_{T})
\bigg\ra_0. \\[6pt]
&= 
\Big \la  
\mathcal{S}^{-1}(\formg \ssS_{\text{int}}) 
\star_{\text{H}} 
\text{T}_S(\ssA_S)
\star_{\text{H}}
\mathcal{S}(\formg \ssS_{\text{int}}) 
\star_{\text{H}}
\text{T}_T(\ssA_T)
\Big \ra_0  \\[6pt] 
&=
\dfrac{1}
{  \Big \la \mathcal{S}(\formg\ssS_{\text{int}}) \Big  \ra_0 }   
\Big\la \text{T}_S(\ssA_S)\star_{\text{H}}  \mathcal{S}(\formg\ssS_{\text{int}})\star_{\text{H}}  
\text{T}_T(\ssA_T)
\Big\ra_0.
\end{align*}
For the final step, we used vacuum stability \textcolor{blue}{(\refeqq{eq:vacstab})}. 
\end{proof}

%\newpage
\bibliographystyle{alpha}

\begin{thebibliography}{CHDD{\etalchar{+}}20}
	
	\bibitem[AB08]{brown08}
	Peter Abramenko and Kenneth~S. Brown.
	\newblock {\em Buildings}, volume 248 of {\em Graduate Texts in Mathematics}.
	\newblock Springer, New York, 2008.
	\newblock Theory and applications.
	
	\bibitem[Abd04]{MR2036353}
	Abdelmalek Abdesselam.
	\newblock Feynman diagrams in algebraic combinatorics.
	\newblock {\em S\'{e}m. Lothar. Combin.}, 49:Art. B49c, 45, 2002/04.
	
	\bibitem[AM10]{aguiar2010monoidal}
	Marcelo Aguiar and Swapneel Mahajan.
	\newblock {\em Monoidal functors, species and {H}opf algebras}, volume~29 of
	{\em CRM Monograph Series}.
	\newblock American Mathematical Society, Providence, RI, 2010.
	\newblock With forewords by Kenneth Brown, Stephen Chase and Andr\'{e} Joyal.
	
	\bibitem[AM13]{aguiar2013hopf}
	Marcelo Aguiar and Swapneel Mahajan.
	\newblock Hopf monoids in the category of species.
	\newblock In {\em Hopf algebras and tensor categories}, volume 585 of {\em
		Contemp. Math.}, pages 17--124. Amer. Math. Soc., Providence, RI, 2013.
	
	\bibitem[AM17]{aguiar2017topics}
	Marcelo Aguiar and Swapneel Mahajan.
	\newblock {\em Topics in hyperplane arrangements}, volume 226 of {\em
		Mathematical Surveys and Monographs}.
	\newblock American Mathematical Society, Providence, RI, 2017.
	
	\bibitem[AM20]{aguiar2020bimonoids}
	Marcelo Aguiar and Swapneel Mahajan.
	\newblock {\em Bimonoids for Hyperplane Arrangements}, volume 173.
	\newblock Cambridge University Press, 2020.
	
	\bibitem[Ara61]{Huz1}
	Huzihiro Araki.
	\newblock Generalized retarded functions and analytic function in momentum
	space in quantum field theory.
	\newblock {\em Journal of Mathematical Physics}, 2(2):163--177, 1961.
	
	\bibitem[Bar78]{barratt1978twisted}
	M.~G. Barratt.
	\newblock Twisted {L}ie algebras.
	\newblock In {\em Geometric applications of homotopy theory ({P}roc. {C}onf.,
		{E}vanston, {I}ll., 1977), {II}}, volume 658 of {\em Lecture Notes in Math.},
	pages 9--15. Springer, Berlin, 1978.
	
	\bibitem[B{\"a}r15]{bar15}
	Christian B{\"a}r.
	\newblock Green-hyperbolic operators on globally hyperbolic spacetimes.
	\newblock {\em Comm. Math. Phys.}, 333(3):1585--1615, 2015.
	
	\bibitem[BC19]{borges2019generalized}
	Francisco Borges and Freddy Cachazo.
	\newblock Generalized planar {F}eynman diagrams: collections.
	\newblock {\em arXiv preprint arXiv:1910.10674}, 2019.
	
	\bibitem[BDF09]{dutfred09}
	R.~Brunetti, M.~D\"{u}tsch, and K.~Fredenhagen.
	\newblock Perturbative algebraic quantum field theory and the renormalization
	groups.
	\newblock {\em Adv. Theor. Math. Phys.}, 13(5):1541--1599, 2009.
	
	\bibitem[BF00]{klaus2000micro}
	Romeo Brunetti and Klaus Fredenhagen.
	\newblock Microlocal analysis and interacting quantum field theories:
	renormalization on physical backgrounds.
	\newblock {\em Comm. Math. Phys.}, 208(3):623--661, 2000.
	
	\bibitem[Bjo15]{MR3467341}
	Anders Bjorner.
	\newblock Positive sum systems.
	\newblock In {\em Combinatorial methods in topology and algebra}, volume~12 of
	{\em Springer INdAM Ser.}, pages 157--171. Springer, Cham, 2015.
	
	\bibitem[BK05]{Kreimer05}
	Christoph Bergbauer and Dirk Kreimer.
	\newblock The {H}opf algebra of rooted trees in {E}pstein-{G}laser
	renormalization.
	\newblock {\em Ann. Henri Poincar\'{e}}, 6(2):343--367, 2005.
	
	\bibitem[BL75]{bros}
	J.~Bros and M.~Lassalle.
	\newblock Analyticity properties and many-particle structure in general quantum
	field theory. {II}. {O}ne-particle irreducible {$n$}-point functions.
	\newblock {\em Comm. Math. Phys.}, 43(3):279--309, 1975.
	
	\bibitem[BLL98]{bergeron1998combinatorial}
	F.~Bergeron, G.~Labelle, and P.~Leroux.
	\newblock {\em Combinatorial species and tree-like structures}, volume~67 of
	{\em Encyclopedia of Mathematics and its Applications}.
	\newblock Cambridge University Press, Cambridge, 1998.
	\newblock Translated from the 1994 French original by Margaret Readdy, with a
	foreword by Gian-Carlo Rota.
	
	\bibitem[BMM{\etalchar{+}}12]{billera2012maximal}
	L.J. Billera, J.~Tatch Moore, C.~Dufort Moraites, Y.~Wang, and K.~Williams.
	\newblock Maximal unbalanced families.
	\newblock {\em arXiv preprint arXiv:1209.2309}, 2012.
	
	\bibitem[Bor11]{Borcherds10}
	Richard~E. Borcherds.
	\newblock Renormalization and quantum field theory.
	\newblock {\em Algebra Number Theory}, 5(5):627--658, 2011.
	
	\bibitem[BP99]{MR1731815}
	Alexander Barvinok and James~E. Pommersheim.
	\newblock An algorithmic theory of lattice points in polyhedra.
	\newblock In {\em New perspectives in algebraic combinatorics ({B}erkeley,
		{CA}, 1996--97)}, volume~38 of {\em Math. Sci. Res. Inst. Publ.}, pages
	91--147. Cambridge Univ. Press, Cambridge, 1999.
	
	\bibitem[Bro09]{Brouder10}
	Christian Brouder.
	\newblock Quantum field theory meets {H}opf algebra.
	\newblock {\em Math. Nachr.}, 282(12):1664--1690, 2009.
	
	\bibitem[Bro17]{MR3713351}
	Francis Brown.
	\newblock Feynman amplitudes, coaction principle, and cosmic {G}alois group.
	\newblock {\em Commun. Number Theory Phys.}, 11(3):453--556, 2017.
	
	\bibitem[BS59]{Bogoliubov59}
	N.~N. Bogoliubov and D.~V. Shirkov.
	\newblock {\em Introduction to the theory of quantized fields}.
	\newblock Authorized English edition. Revised and enlarged by the authors.
	Translated from the Russian by G. M. Volkoff. Interscience Monographs in
	Physics and Astronomy, Vol. III. Interscience Publishers, Inc., New York;
	Interscience Publishers Ltd., London, 1959.
	
	\bibitem[CGUZ19]{cachazo2019planar}
	Freddy Cachazo, Alfredo Guevara, Bruno Umbert, and Yong Zhang.
	\newblock Planar matrices and arrays of {F}eynman diagrams.
	\newblock {\em arXiv preprint arXiv:1912.09422}, 2019.
	
	\bibitem[CHDD{\etalchar{+}}19]{caron2019cosmic}
	Simon Caron-Huot, Lance~J Dixon, Falko Dulat, Matt Von~Hippel, Andrew~J McLeod,
	and Georgios Papathanasiou.
	\newblock The cosmic {G}alois group and extended {S}teinmann relations for
	planar {N}=4 {SYM} amplitudes.
	\newblock {\em Journal of High Energy Physics}, 2019(9):61, 2019.
	
	\bibitem[CHDD{\etalchar{+}}20]{Caron-Huot:2020bkp}
	Simon Caron-Huot, Lance~J. Dixon, James~M. Drummond, Falko Dulat, Jack Foster,
	\"Omer G\"urdo\u{g}an, Matt von Hippel, Andrew~J. McLeod, and Georgios
	Papathanasiou.
	\newblock {The Steinmann Cluster Bootstrap for $N=4$ Super Yang-Mills
		Amplitudes}.
	\newblock {\em PoS}, CORFU2019:003, 2020.
	
	\bibitem[CJ70]{MR272301}
	John~T. Cannon and Arthur~M. Jaffe.
	\newblock Lorentz covariance of the {$\lambda (\varphi^{4})_{2}$} quantum field
	theory.
	\newblock {\em Comm. Math. Phys.}, 17:261--321, 1970.
	
	\bibitem[Col16]{collini2016fedosov}
	Giovanni Collini.
	\newblock Fedosov quantization and perturbative quantum field theory.
	\newblock {\em arXiv preprint arXiv:1603.09626}, 2016.
	
	\bibitem[DF01]{dutfred00}
	M.~D\"{u}tsch and K.~Fredenhagen.
	\newblock Algebraic quantum field theory, perturbation theory, and the loop
	expansion.
	\newblock {\em Comm. Math. Phys.}, 219(1):5--30, 2001.
	
	\bibitem[DF04]{dutfredretard04}
	Michael D\"{u}tsch and Klaus Fredenhagen.
	\newblock Causal perturbation theory in terms of retarded products, and a proof
	of the action {W}ard identity.
	\newblock {\em Rev. Math. Phys.}, 16(10):1291--1348, 2004.
	
	\bibitem[DFG18]{drummond2018cluster}
	James Drummond, Jack Foster, and {\"O}mer G{\"u}rdo{\u{g}}an.
	\newblock Cluster adjacency properties of scattering amplitudes in {N}=4
	supersymmetric {Y}ang-{M}ills theory.
	\newblock {\em Physical review letters}, 120(16):161601, 2018.
	
	\bibitem[DFKR14]{FredHopf14}
	Michael D\"{u}tsch, Klaus Fredenhagen, Kai~Johannes Keller, and Katarzyna
	Rejzner.
	\newblock Dimensional regularization in position space and a forest formula for
	{E}pstein-{G}laser renormalization.
	\newblock {\em J. Math. Phys.}, 55(12):122303, 37, 2014.
	
	\bibitem[D{\"u}t12]{dut12}
	Michael D{\"u}tsch.
	\newblock Connection between the renormalization groups of
	{S}t\"{u}ckelberg-{P}etermann and {W}ilson.
	\newblock {\em Confluentes Math.}, 4(1):1240001, 16, 2012.
	
	\bibitem[D{\"u}t19]{dutsch2019perturbative}
	Michael D{\"u}tsch.
	\newblock {\em From classical field theory to perturbative quantum field
		theory}, volume~74 of {\em Progress in Mathematical Physics}.
	\newblock Birkh\"{a}user/Springer, Cham, 2019.
	\newblock With a foreword by Klaus Fredenhagen.
	
	\bibitem[Dys52]{Dyson52}
	F.~J. Dyson.
	\newblock Divergence of perturbation theory in quantum electrodynamics.
	\newblock {\em Phys. Rev. (2)}, 85:631--632, 1952.
	
	\bibitem[Ear19]{early2019planar}
	Nick Early.
	\newblock Planar kinematic invariants, matroid subdivisions and generalized
	{F}eynman diagrams.
	\newblock {\em arXiv preprint arXiv:1912.13513}, 2019.
	
	\bibitem[EG73]{ep73roleofloc}
	H.~Epstein and V.~Glaser.
	\newblock The role of locality in perturbation theory.
	\newblock {\em Ann. Inst. H. Poincar\'{e} Sect. A (N.S.)}, 19:211--295 (1974),
	1973.
	
	\bibitem[EGS75]{epstein1976general}
	H.~Epstein, V.~Glaser, and R.~Stora.
	\newblock {General properties of the n-point functions in local quantum field
		theory}.
	\newblock In {\em {Institute on Structural Analysis of Multiparticle Collision
			Amplitudes in Relativistic Quantum Theory Les Houches, France, June 3-28,
			1975}}, pages 5--93, 1975.
	
	\bibitem[Eps16]{epstein2016}
	Henri Epstein.
	\newblock Trees.
	\newblock {\em Nuclear Phys. B}, 912:151--171, 2016.
	
	\bibitem[Far11]{MR2862982}
	William~G. Faris.
	\newblock Combinatorial species and {F}eynman diagrams.
	\newblock {\em S\'{e}m. Lothar. Combin.}, 61A:Art. B61An, 37, 2009/11.
	
	\bibitem[Fau01]{Fauser01}
	Bertfried Fauser.
	\newblock On the {H}opf algebraic origin of {W}ick normal ordering.
	\newblock {\em J. Phys. A}, 34(1):105--115, 2001.
	
	\bibitem[GBL00]{Bondia00}
	Jose~M Gracia-Bondia and Serge Lazzarini.
	\newblock Connes-{K}reimer-{E}pstein-{G}laser renormalization.
	\newblock {\em arXiv preprint hep-th/0006106v2}, 2000.
	
	\bibitem[GJ68]{MR247845}
	James Glimm and Arthur Jaffe.
	\newblock A {$\lambda \phi^{4}$} quantum field without cutoffs. {I}.
	\newblock {\em Phys. Rev. (2)}, 176:1945--1951, 1968.
	
	\bibitem[GJS74]{MR363256}
	James Glimm, Arthur Jaffe, and Thomas Spencer.
	\newblock The {W}ightman axioms and particle structure in the {${\mathscr
			P}(\phi)_{2}$} quantum field model.
	\newblock {\em Ann. of Math. (2)}, 100:585--632, 1974.
	
	\bibitem[GK18]{MR3753672}
	John Gough and Joachim Kupsch.
	\newblock {\em Quantum fields and processes}, volume 171 of {\em Cambridge
		Studies in Advanced Mathematics}.
	\newblock Cambridge University Press, Cambridge, 2018.
	\newblock A combinatorial approach.
	
	\bibitem[GLZ57]{GLZ1957}
	V.~Glaser, H.~Lehmann, and W.~Zimmermann.
	\newblock Field operators and retarded functions.
	\newblock {\em Nuovo Cimento (10)}, 6:1122--1128, 1957.
	
	\bibitem[HK64]{haagkas64}
	Rudolf Haag and Daniel Kastler.
	\newblock An algebraic approach to quantum field theory.
	\newblock {\em J. Mathematical Phys.}, 5:848--861, 1964.
	
	\bibitem[Hol08]{MR2455327}
	Stefan Hollands.
	\newblock Renormalized quantum {Y}ang-{M}ills fields in curved spacetime.
	\newblock {\em Rev. Math. Phys.}, 20(9):1033--1172, 2008.
	
	\bibitem[HR20]{MR4109798}
	Eli Hawkins and Kasia Rejzner.
	\newblock The star product in interacting quantum field theory.
	\newblock {\em Lett. Math. Phys.}, 110(6):1257--1313, 2020.
	
	\bibitem[IS78]{Slavnov78}
	V.~A. Ilyin and D.~A. Slavnov.
	\newblock Algebras of observables in the {$S$}-matrix approach.
	\newblock {\em Teoret. Mat. Fiz.}, 36(1):32--41, 1978.
	
	\bibitem[Joy81]{joyal1981theorie}
	Andr\'{e} Joyal.
	\newblock Une th\'{e}orie combinatoire des s\'{e}ries formelles.
	\newblock {\em Adv. in Math.}, 42(1):1--82, 1981.
	
	\bibitem[Joy86]{joyal1986foncteurs}
	Andr\'{e} Joyal.
	\newblock Foncteurs analytiques et esp\`eces de structures.
	\newblock In {\em Combinatoire \'{e}num\'{e}rative ({M}ontreal, {Q}ue.,
		1985/{Q}uebec, {Q}ue., 1985)}, volume 1234 of {\em Lecture Notes in Math.},
	pages 126--159. Springer, Berlin, 1986.
	
	\bibitem[LM00]{losevmanin}
	A.~Losev and Y.~Manin.
	\newblock New moduli spaces of pointed curves and pencils of flat connections.
	\newblock volume~48, pages 443--472. 2000.
	\newblock Dedicated to William Fulton on the occasion of his 60th birthday.
	
	\bibitem[LNO19]{lno2019}
	Zhengwei Liu, William Norledge, and Adrian Ocneanu.
	\newblock The adjoint braid arrangement as a combinatorial {L}ie algebra via
	the {S}teinmann relations.
	\newblock {\em arXiv preprint arXiv:1901.03243}, 2019.
	
	\bibitem[NO19]{norledge2019hopf}
	William Norledge and Adrian Ocneanu.
	\newblock Hopf monoids, permutohedral cones, and generalized retarded
	functions.
	\newblock {\em to appear in Annals IHP D - Comb., Phys. and their
		Interactions}, 2019.
	
	\bibitem[Nor20]{norledge2020species}
	William Norledge.
	\newblock Species-theoretic foundations of perturbative quantum field theory.
	\newblock {\em arXiv preprint arXiv:2009.09969}, 2020.
	
	\bibitem[Ocn18]{oc17}
	Adrian Ocneanu.
	\newblock {H}igher {R}epresentation {T}heory.
	\newblock {P}hysics 267, Fall 2017 Harvard {C}ourse, and supplementary
	materials and presentations, 2017-2018.
	\newblock The course is available on YouTube, and supplementary materials are
	in preparation for publication.
	
	\bibitem[Pet13]{MR3134040}
	Dan Petersen.
	\newblock The operad structure of admissible {$G$}-covers.
	\newblock {\em Algebra Number Theory}, 7(8):1953--1975, 2013.
	
	\bibitem[Pin00]{MR1845168}
	G.~Pinter.
	\newblock The {H}opf algebra structure of {C}onnes and {K}reimer in
	{E}pstein-{G}laser renormalization.
	\newblock {\em Lett. Math. Phys.}, 54(3):227--233, 2000.
	
	\bibitem[Pol58]{polk58}
	J.~C. Polkinghorne.
	\newblock Generalized retarded products.
	\newblock {\em Proc. Roy. Soc. London Ser. A}, 247:557--561, 1958.
	
	\bibitem[PS16]{stora16}
	G.~Popineau and R.~Stora.
	\newblock A pedagogical remark on the main theorem of perturbative
	renormalization theory.
	\newblock {\em Nuclear Phys. B}, 912:70--78, 2016.
	
	\bibitem[Rej16]{rejzner2016pQFT}
	Kasia Rejzner.
	\newblock {\em Perturbative algebraic quantum field theory}.
	\newblock Mathematical Physics Studies. Springer, Cham, 2016.
	\newblock An introduction for mathematicians.
	
	\bibitem[Rue61]{Ruelle}
	D.~Ruelle.
	\newblock Connection between {W}ightman functions and {G}reen functions in
	{$p$}-space.
	\newblock {\em Nuovo Cimento (10)}, 19:356--376, 1961.
	
	\bibitem[Sch93]{Bill93}
	William~R. Schmitt.
	\newblock Hopf algebras of combinatorial structures.
	\newblock {\em Canad. J. Math.}, 45(2):412--428, 1993.
	
	\bibitem[Sch95]{MR1359058}
	G.~Scharf.
	\newblock {\em Finite quantum electrodynamics}.
	\newblock Texts and Monographs in Physics. Springer-Verlag, Berlin, second
	edition, 1995.
	\newblock The causal approach.
	
	\bibitem[Sch18]{schultka2018toric}
	Konrad Schultka.
	\newblock Toric geometry and regularization of {F}eynman integrals.
	\newblock {\em arXiv preprint arXiv:1806.01086}, 2018.
	
	\bibitem[Sch20]{perturbative_quantum_field_theory}
	Urs Schreiber.
	\newblock {G}eometry of physics - perturbative quantum field theory.
	\newblock
	\url{https://ncatlab.org/nlab/show/geometry+of+physics+--+perturbative+quantum+field+theory},
	September 2020.
	\newblock Revision 197.
	
	\bibitem[Ste60a]{steinmann1960zusammenhang}
	Othmar Steinmann.
	\newblock \"{U}ber den {Z}usammenhang zwischen den {W}ightmanfunktionen und den
	retardierten {K}ommutatoren.
	\newblock {\em Helv. Phys. Acta}, 33:257--298, 1960.
	
	\bibitem[Ste60b]{steinmann1960}
	Othmar Steinmann.
	\newblock Wightman-{F}unktionen und retardierte {K}ommutatoren. {II}.
	\newblock {\em Helv. Phys. Acta}, 33:347--362, 1960.
	
	\bibitem[Ste71]{steinbook71}
	Othmar Steinmann.
	\newblock {\em Perturbation expansions in axiomatic field theory}.
	\newblock Springer-Verlag, Berlin-New York, 1971.
	\newblock Lecture Notes in Physics, Vol. 11.
	
	\bibitem[Sto93a]{stora1993differential}
	R~Stora.
	\newblock Differential algebras in {L}agrangean field theory.
	\newblock {\em ETH-Z{\"u}rich Lectures}, 1993.
	
	\bibitem[Sto93b]{stover1993equivalence}
	Christopher~R. Stover.
	\newblock The equivalence of certain categories of twisted {L}ie and {H}opf
	algebras over a commutative ring.
	\newblock {\em J. Pure Appl. Algebra}, 86(3):289--326, 1993.
	
	\bibitem[SZ11]{shadrin2011group}
	Sergey Shadrin and Dimitri Zvonkine.
	\newblock A group action on {L}osev-{M}anin cohomological field theories.
	\newblock In {\em Annales de l'Institut Fourier}, volume~61, pages 2719--2743,
	2011.
	
	\bibitem[Tit74]{Tits74}
	Jacques Tits.
	\newblock {\em Buildings of spherical type and finite {BN}-pairs}.
	\newblock Lecture Notes in Mathematics, Vol. 386. Springer-Verlag, Berlin-New
	York, 1974.
	
\end{thebibliography}
\newcommand{\etalchar}[1]{$^{#1}$}

%\nocite{*}
%\printindex

\end{document}